\documentclass{article}[11pt]

\usepackage{xcolor}
\usepackage{fullpage}
\usepackage{amsmath}
\usepackage{amsfonts}
\usepackage{amsthm}
\usepackage{bbm}
\usepackage{dirtytalk}
\usepackage{natbib}
\usepackage{graphicx}
\usepackage{comment}
\usepackage[linesnumbered,ruled,vlined]{algorithm2e}
\usepackage{thm-restate}
\usepackage{hyperref}
\usepackage[capitalize]{cleveref}  
\usepackage{tikz}
\usepackage{subfig}
\usepackage{floatrow}
\usepackage{blindtext}
\usepackage{mathtools}  
\usepackage{caption}
\usepackage{tabularx}
\usepackage[export]{adjustbox}
\usepackage{soul}

\allowdisplaybreaks
\newfloatcommand{capbtabbox}{table}[][\FBwidth]

\newtheorem{theorem}{Theorem}[section]
\theoremstyle{definition}
\newtheorem{definition}{Definition}
\theoremstyle{plain}

\newtheorem{claim}{Claim}
\newtheorem{lemma}[theorem]{Lemma}
\newtheorem{fact}[theorem]{Fact}
\theoremstyle{definition}
\newtheorem{example}{Example}

\theoremstyle{plain}
\newtheorem{corollary}[theorem]{Corollary}

\newcommand{\opt}[0]{\mathrm{OPT}}
\newcommand{\feas}{\mathcal{F}}
\newcommand{\dist}{\mathcal{D}}
\newcommand{\ds}{T}
\newcommand{\ind}{I}
\newcommand{\mc}{\mathcal{M}}
\newcommand{\argmin}{\operatorname{argmin}}
\newcommand{\argmax}{\operatorname{argmax}}
\newcommand{\optf}{f}
\newcommand{\dd}[2]{\frac{d}{d#1}#2}
\newcommand{\surr}{\rho}

\newcommand{\expect}[1]{\operatorname{E}\left[#1\right]}
\newcommand{\expectt}[2]{\operatorname{E}_{#1}\left[#2\right]}
\newcommand{\prob}[1]{\operatorname{Pr}\left[#1\right]}

\newcommand{\real}{\mathbb{R}}
\newcommand{\matroid}{\mathrm{M}}
\newcommand{\comset}[1]{\mathcal{C}\left(#1\right)}
\newcommand{\instance}{\mathbb{I}}
\newcommand{\pbox}{\mathcal{B}}

\newcommand{\X}{\mathcal{X}}
\newcommand{\Z}{\mathcal{Z}}
\newcommand{\T}{\mathcal{T}}
\newcommand{\A}{\mathcal{A}}
\newcommand{\R}{\mathcal{R}}
\newcommand{\I}{\mathcal{I}}
\newcommand{\vecM}{\mathbb{M}}
\newcommand{\cg}{\mathrm{CG}}
\usepackage[ezlist, cleveref, delimiters, notation, probability style=rm]{ezlib}
\newcommand{\insp}{\mathrm{open}}
\newcommand{\alg}{\mathrm{alg}}

\newcommand{\grab}{\mathrm{grab}}
\DeclareMathOperator{\Bernoulli}{Bernoulli}
\newcommand{\Dimitris}[1]{{\color{red}[Dimitris] #1}}

\newcommand{\comms}{\mathcal{P}}
\newcommand{\traj}{\Upsilon}

\title{Combinatorial Selection with Costly Information}
\author{
Shuchi Chawla\footnote{University of Texas at Austin. The authors were supported in part by NSF award CCF-2225259.} \\ {\tt shuchi@cs.utexas.edu} 
\and 
Dimitris Christou$^*$  \\ {\tt christou@cs.utexas.edu} 
\and
Amit Harlev\footnote{Cornell University. This author was supported by the Department of Defense (DoD) through the National Defense Science $\&$ Engineering Graduate (NDSEG) Fellowship Program.} \\ {\tt ah843@cornell.edu}
\and
Ziv Scully\footnote{Cornell University. This author was supported by NSF award CMMI-2307008.} \\ {\tt zivscully@cornell.edu}
}
\date{}

\begin{document}

\maketitle
\thispagestyle{empty}
\addtocounter{page}{-1}
\begin{abstract}


We consider a class of optimization problems over stochastic variables where the algorithm can learn information about the value of any variable through a series of costly steps; we model this information acquisition process as a Markov Decision Process (MDP). The algorithm's goal is to minimize the cost of its solution plus the cost of information acquisition, or alternately, maximize the value of its solution minus the cost of information acquisition. Such {\em bandit superprocesses} have been studied previously but solutions are known only for fairly restrictive special cases.

We develop a framework for approximate optimization of bandit superprocesses that applies to arbitrary acyclic MDPs with a matroid feasibility constraint. Our framework establishes a bound on the optimal cost through a novel cost amortization; it then couples this bound with a notion of local approximation that allows approximate solutions for each component MDP in the superprocess to be composed without loss into a global approximation. 

We use this framework to obtain approximately optimal solutions for several variants of bandit superprocesses for both maximization and minimization. We obtain new approximations for combinatorial versions of the previously studied  Pandora's Box with Optional Inspection and Pandora's Box with Partial Inspection; the less-studied Additive Pandora's Box problem; as well as a new problem that we call the Weighing Scale problem.
    
\end{abstract}

\newpage
\thispagestyle{empty}
\addtocounter{page}{-2}
\tableofcontents

\newpage
\section{Introduction}\label{sec:intro}

Many real-world settings involve algorithmic decision-making under incomplete information about the underlying input and future outcomes. In some cases, additional information can be acquired at a cost. A fundamental question arises: how should this costly information acquisition be incorporated into the algorithmic process? 

Consider, for example, a biotechnology company running multiple drug discovery projects. These projects require upfront investment and entail a series of choices on which directions to pursue {\em before} their potential is known and realized. Other examples include an oil company's search for an optimal drilling site, a construction firm paying for multiple architectural designs before choosing one to pursue, and a manufacturer conducting market research before deciding which products to produce and in what quantity.

A classical model for costly information acquisition is the Pandora's Box problem, introduced by \citet{W79}. In this model, an algorithm faces an optimization problem over $n$ stochastic alternatives, each concealed inside a closed box. The algorithm can open a box at a cost to observe the value of the alternative within, and can select an alternative only after its box has been opened. The goal is to minimize the total cost, which includes both the cost of opening the boxes and the values of the chosen alternatives.\footnote{For maximization problems, the objective is to maximize the function value minus the cost of opening the boxes. In this section, we will focus our exposition on minimization problems although, as we show later in the paper, our framework and techniques extend to maximization problems as well.} The Pandora's Box model (henceforth, PB) is well understood and admits a simple optimal solution for a broad class of optimization problems. However, it is too simplistic to capture the complexities of real-world information acquisition: it assumes that each alternative can be revealed in only one way and that the algorithm either learns everything or nothing about it—there is no concept of partial information acquisition.

Recent research has explored extensions of PB that allow for multiple modes and stages of information acquisition.\footnote{See, for example, \citet{OW15, KWG16, S17, D18, BK19, EHLM19, GJSS19, BFLL20, AJS20, FLL22, BC22, BEFF23, BW24, DS24}.} These generalizations are typically NP-hard \citep{FLL22}, necessitating the development of approximation algorithms. \citet{BC24-survey} provide a survey of recent advances in this domain. However, a significant limitation of prior work is that existing solutions are tailored to specific problem variants and do not extend to broader settings. For example, \citet{AJS20} address a two-stage inspection process, but their techniques do not generalize to three stages. Similarly, while previous results separately handle cases with partial inspection and optional inspection, it is unclear how to approach an instance that contains alternatives of both kinds. 

\begin{center}
\parbox{0.95\textwidth}{
\noindent
In this paper, we develop a general framework for modeling \textit{arbitrarily complex} protocols for information acquisition and designing approximation algorithms for them. Using this framework, we obtain improved approximations for several PB variants and develop algorithms for new models where non-trivial approximations were previously unknown.
}
\end{center}

\paragraph{A Framework for Modeling Information Acquisition.}

We model the process of costly information acquisition for each alternative $i \in [n]$ using a finite-horizon Markov decision process (MDP), denoted by $\mathcal{M}_i$. The current state of this MDP encapsulates all information the algorithm has about alternative $i$ and determines the set of available actions for further exploration. Some actions lead to terminal states, where the alternative can be selected. For instance, both a standard optional-inspection PB and a three-stage inspection PB can be modeled as MDPs (see Figure~\ref{fig:mdp-examples}). A key assumption in our framework is the independence of these MDPs: exploring one alternative does not affect the state of another.

This framework defines a broad family of optimization problems, which we call \textit{Costly Information Combinatorial Selection} (CICS). In CICS, we are given $n$ stochastic alternatives, each associated with an MDP $\mathcal{M}_1, \dots, \mathcal{M}_n$. The goal is to select a subset of alternatives satisfying a given feasibility constraint; our work focuses on matroid feasibility constraints. The algorithm iteratively chooses an alternative $i$ and takes an action in $\mathcal{M}_i$ to gather information. At any time, it can stop and select a feasible subset of alternatives that have reached terminal states. The objective function is the total expected value of the selected alternatives plus the cost of all actions taken.

\begin{figure}[h!]
    \centering
    \includegraphics[width=0.8\textwidth]{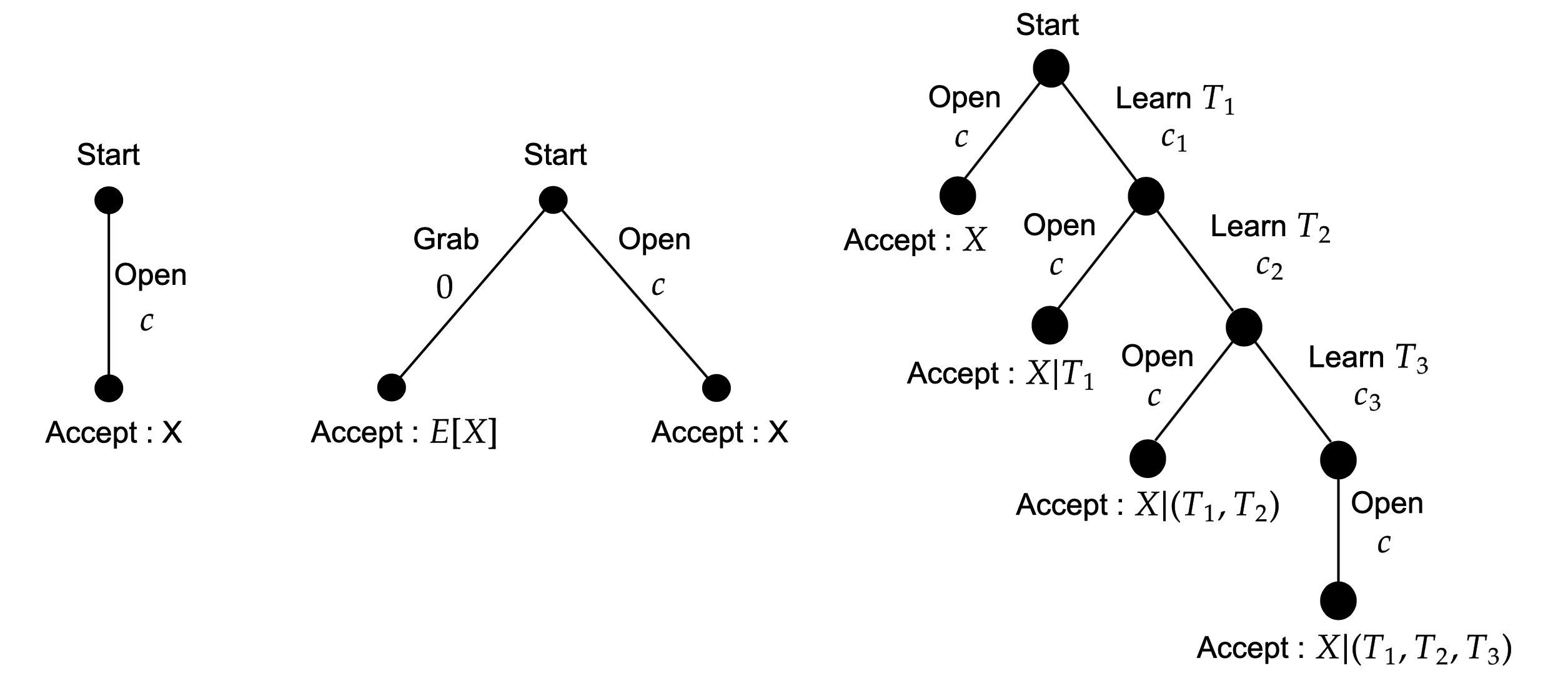}
    \caption{Three examples of a Pandora's Box style information acquisition protocol. (Left:) A classical PB is modeled as a Markov chain with a single costly action (opening the box) resulting in a random terminal state of value $X$. (Middle:) An optional-inspection PB modeled as an MDP; in addition to opening the box, there is a grab action that reveals no information and results in a single terminal state of deterministic value $\expect{X}$. (Right:) A three-stage inspection protocol where the box can be opened directly, or it can be partially explored by sequentially revealing hidden variables $T_i$ at a cost.}
    \label{fig:mdp-examples}
\end{figure}

The primary challenge in solving CICS arises from the interplay between two levels of decision-making: (i)~\textit{global decision-making}, which determines which alternative to explore at each step, and (ii)~\textit{local decision-making}, which dictates the costly actions taken within each alternative's MDP. These decisions are typically adaptive, responding to the information that was acquired in previous steps. Since CICS generalizes MDPs, we have no hope of obtaining a general purpose solution to the problem. Instead, what we aim for is a decomposition of the problem into smaller pieces that can each be solved independently and composed together into a global solution. We obtain such a decomposition by analyzing a class of algorithms called \textit{committing policies}.

\paragraph{Committing Policies and the Commitment Gap.}
A \emph{committing policy} for CICS is a policy that, in essence, predetermines all local decisions before making any global decisions.
In particular, it preselects an action  for each state in every MDP $\mathcal{M}_i$, effectively reducing the MDPs to Markov chains.
This is valuable because for the class of feasibility constraints we are interested in, the special case of CICS where each MDP $\mathcal{M}_i$ is effectively a Markov chain, i.e. has just one action in every state, admits a simple optimal policy \citep{DTW03, GJSS19}. By fixing local decisions in advance, we significantly reduce the complexity of global decision-making. This prompts two questions:
\* What is the best approximation ratio achievable by committing policies for a given class of CICS problems? We call this ratio the \emph{commitment gap}.
\* Can the local decision commitments be made in a \emph{composable} way? That is, can the commitments for each MDP $\mathcal{M}_i$ be determined using only the characteristics of $\mathcal{M}_i$, independent of other MDPs?
\*/
These questions have been posed and investigated for optional-inspection PB \citep{BK19, DS24}, but little is known beyond that setting. In this paper, we address both questions by developing a set of \emph{local} sufficient conditions for bounding the commitment gap in general CICS.
By ``local'', we mean that our conditions depend only on the properties of each individual MDP $\mathcal{M}_i$, without reference to the other MDPs or the feasibility constraint.
Furthermore, our conditions yield policies that determine their commitments locally (or, in one case, with a simple global algorithm after local pre-processing).

We introduce two key conceptual contributions that enable bounds on the commitment gap: a method of bounding the optimum, and notions of approximation ratio that can be checked locally for each MDP.

\paragraph{Bounding the Optimum.}
Our first conceptual contribution is a lower bound on the cost of the (unrestricted, non-committing) optimum. This bound is analogous to the Whittle integral \citep{W80} for infinite-horizon discounted-reward MDPs. While prior proofs of the Whittle integral exist for the discounted-reward setting \citep{BS13, HMR15}, we present an algorithmic proof for the finite-horizon case, which is both simpler and more insightful. 

Importantly, our approach extends the notion of \textit{surrogate costs} from PB literature to general MDPs and connects it to the Whittle integral,  offering valuable algorithmic insights into bounding the commitment gap. Surrogate costs were first defined for classical PB by \cite{KWG16} and \cite{S17}, and can be extended easily to arbitrary Markov chains \citep{W92, DTW03, GJSS19}. However, the extension to MDPs is highly non-trivial and is a novel contribution of our work. We describe the details in Section~\ref{sec:intro-tech}.

\paragraph{Local Approximation.}
Our second contribution is the development of a \textit{local approximation} technique, which expresses the global approximation factor of a committing policy $(\pi_1, \dots, \pi_n)$ in terms of the local properties of each constituent policy $\pi_i$ relative to its corresponding MDP $\mathcal{M}_i$. This concept was first introduced by \citet{DS24} for the optional-inspection PB; we generalize it to arbitrary MDPs.


Crucially, local approximations compose: if each $\pi_i$ achieves an $\alpha$-local approximation for $\mc_i$, then the committing policy $(\pi_1, \dots, \pi_n)$ achieves an $\alpha$-global approximation. This observation answers, for example, our aforementioned question of how to compose partial-inspection Pandora's Boxes with optional-inspection ones within a single instance. Furthermore, our global approximation guarantee extends to any combinatorial setting that admits a ``greedy-style'' algorithm via techniques similar to \citep{S17, GJSS19}.

Although local approximation can be defined directly with respect to the Whittle integral, as in \cite{DS24}, proving guarantees using this definition appears very challenging. One of our main contributions is to connect this concept to surrogate costs in a manner that enables simpler analysis. We develop multiple approaches for bounding local approximation factors that are suitable for different contexts, and demonstrate their use through the applications discussed in Section~\ref{sec:intro_applications}.


\paragraph{Example Application: Shortest Paths.}
To illustrate the power of our framework, consider the \textit{Pandora's Shortest Path} problem introduced by \citet{S17}: the edge weights of a given graph correspond to independent Pandora's boxes and the objective is to accept a set of (open) boxes that form a path between two given vertices $s$ and $t$, while minimizing the total cost of opening boxes plus the cost of the path. As path constraints do not admit frugal algorithms that enable the generalization of Weitzman's result to combinatorial settings, this problem is notoriously challenging and no approximation results are known.

When the underlying graph is a union of disjoint $s$-$t$ paths, this problem becomes an instantiation of CICS: each path corresponds to an alternative, whose MDP encapsulates the different orders in which the constituent ``edge boxes'' can be opened. A committing policy, in this case, pre-defines a protocol for the order of inspecting the edges within a path, that can adapt only to the instantiations of the edges from the same path. We show that the commitment gap of this special case of Pandora's Shortest Paths is at most $2$.


\subsection{Our technical contributions and their relationship to prior work}
\label{sec:intro-tech}
We now describe our technical contributions in more detail in the context of closely related prior work. 
A thorough discussion of other related work is presented in~\Cref{sec:related}. 

\subsubsection*{Bandit superprocesses and the Whittle integral}

CICS is closely related to a class of sequential decision-making problems called bandit superprocesses (BSPs), introduced by \cite{N73}. A BSP is composed from multiple independent MDPs; at every step, the algorithm advances one of the MDPs, receives a (discounted) reward, and the state of the underlying MDP evolves stochastically. A primary difference between the two settings is that a CICS algorithm must terminate by selecting a feasible set of alternatives, whereas a BSP can potentially have an infinite horizon. Moreover, since the states of the MDPs in a CICS model information acquisition, we assume that each constituent MDP is acyclic; our techniques rely on this acyclicity. In contrast, BSPs can be defined over arbitrary MDPs.

In the special case where the MDPs are Markov chains, BSP reduces to the extensively-studied Multi-Armed Bandit (MAB) problem. MAB admits an astonishingly simple optimal solution \citep{G79}: every state in each constituent Markov chain is assigned an index independent of the other chains; at every step, the Markov chain with the maximum index in its current state is advanced. \cite{DTW03} show that this ``indexability'' extends to the finite-horizon setting as well, and \cite{GJSS19} further extend it to combinatorial settings. We restate this result with a new proof as~\Cref{thm:MC-opt}. \cite{W79}'s algorithm for classical PB is essentially a special case of this result. 

General BSPs, in contrast, are not indexable \citep{G82}.\footnote{This is because MAB only involves global decision-making and does not require any local decision-making as there is only one way to advance the chosen Markov chain.} Nevertheless, the local structure of the problem provides insight into the optimal cost: given an MDP $\mc$, \cite{W80}
considers a local problem $(\mc, y)$, where at every step the algorithm can either advance $\mc$ or terminate with a cost of $y$. The values of these local games, formalized as ``optimality curves'' $\optf_\mc(y)$, one for each constituent MDP in the BSP, can be combined to obtain the Whittle integral, a lower bound on the global optimum \citep{W80,BS13}.

Our work provides a new interpretation for Whittle's optimality curves by connecting them with a mapping from trajectories in the MDP to ``surrogate'' costs.
The surrogate costs essentially allow an algorithm to amortize the costs of actions and pay (a part of) them only when the MDP terminates. We develop a recursive water filling algorithm for cost amortization that ensures that good, i.e. low cost or high reward, terminal states are responsible for paying most of the cost share. A crucial property we achieve is the {\em independence} of the distribution of surrogate costs from the actions chosen by an algorithm in the MDP in the local or global game. In particular, the costs capture the local structure of the MDP without limiting in any way how an algorithm ``solves'' the MDP. 
This independence allows us to compose them into a global bound for the optimal (adaptive, non-committing) policy. 


Our approach is heavily influenced by \cite{KWG16} and \cite{S17}'s  amortization viewpoint on Weitzman's index for classical PB, as well as its extension to Markov chains. However, we emphasize that the extension to MDPs while ensuring action independence is an important and novel contribution of our work that enables our approximation bounds. 


\subsubsection*{Commitment gap and local approximation}

We follow the approach of \cite{DS24} to establish a bound on the commitment gap by quantifying the performance of each commitment $\pi_i$ relative to the optimal policy for MDP $\mc_i$ in the local game $(\mc_i, y)$. Such an approach helpfully disentangles global and local decision making.


Local guarantees for $(\mc_i, y)$ need to be established with care, however. \cite{W80} showed that if every $\mc_i$ admits a commitment $\pi_i$ that is optimal for the local game $(\mc_i, y)$ regardless of the value $y$, then the tuple $\comms=(\pi_1,\dotsc , \pi_n)$ is globally optimal. In other words, unambiguous local optimality composes into global optimality. However, the reverse doesn't hold: the globally optimal policy might take actions that are unambiguously {\em suboptimal} in the local game! (See~\Cref{app:example}.) Likewise, simply relating the optimal cost $\optf_{\mc_i}(y)$ to the cost of the commitment $\optf_{\mc_i^{\pi_i}}(y)$ at every $y$ is not sufficient, as we demonstrate in~\Cref{app:example}. 

Following \cite{DS24} we say that $\pi_i$ $\alpha$-locally approximates $\mc_i$ if for all $y$ we have $\optf_{\mc_i^{\pi_i}}(\alpha y)\le \alpha\cdot \optf_{\mc_i}(y)$. Local approximations defined in this manner compose into global guarantees on the commitment gap. We use this approach to obtain new approximations for PB with partial inspections (defined below).



Unfortunately, however, local approximations can be very challenging to prove for general MDPs. We develop a new (but weaker) local condition based on surrogate costs that we call {\em pointwise approximation}. This allows us to analyze more complicated MDPs with multiple rounds of actions for which optimality curves are tricky to establish.

Finally, the commitment gap can sometimes be strictly smaller than what can be established through local approximation.  We investigate this issue in the context of PB with optional inspection and show that improved results can be obtained through a {\em semi-local} argument.


\subsubsection*{Applications of our framework}\label{sec:intro_applications}
We instantiate our framework to obtain new bounds on the commitment gap of four PB extensions. The first three variants are new problems studied in the minimization setting, and we obtain the following results:

\begin{itemize}
    \item {\bf PB with Partial Inspection}. In addition to opening the box, the algorithm can choose to ``peek'' into it at a cheaper cost and learn its value; boxes must still be opened before being selected. We prove that the commitment gap is at most $\sqrt{2}$.

    \item {\bf Additive PB}. The value of the box is given by the sum of independent random variables; each component must be separately inspected at a cost and the box can be selected only after all components have been inspected. This models the shortest paths problem over disjoint paths described earlier. We prove that the commitment gap is at most $2$.

    \item {\bf Weighing Scale Problem}. The box cannot be opened, but can be compared at a cost against a threshold $t$ (chosen by the algorithm); the comparison reveals whether $X_i\geq t$ or not. Any sequence and number of comparisons are allowed. We obtain a logarithmic bound (in key parameters of the problem) on the commitment gap.
\end{itemize}

\noindent All these results apply to any set of matroid feasibility constraints, and can be seamlessly extended to any set of constraints that admit a frugal algorithm, as defined by~\cite{S17}. Furthermore, they immediately extend to settings that consider combinations of these variants: an instance consisting of some (normal) boxes, some partial inspection boxes and some additive boxes under matroid feasibility constraints still has a commitment gap of at most $2$.

Furthermore, we prove that the commitment gap of {\bf PB with Optional Inspection} under matroid feasibility constraints in the maximization setting is at least $0.58$. This problem has been extensively studied in the literature; for the single-selection setting, it is known that the problem admits a PTAS \citep{FLL22,BC22}. For matroid constraints, \cite{BK19} show that the commitment gap is at least $0.63$; however, their approach cannot be used to determine the committing policy that achieves it, as it requires an enumeration of the exponentially many committing policies. In contrast, the committing policy achieving our bound can be efficiently computed and leads to the first approximation result that is strictly better than the (trivial) bound of $0.5$.

\subsection{Further related work}
\label{sec:related}

We have already discussed the relationship of CICS to {\bf bandit superprocesses} and Markovian multi-armed bandit.  
There is little work on BSPs that do not satisy Whittle's ``unambiguous local optimality'' condition, outside of the variants of Pandora's Box. The only such work we are aware of is by \cite{KVB18}, who look at a problem with two specific symmetric MDPs and an arbitrary outside option, and develop an adaptive (but complicated) exact algorithm for this problem. We refer the reader to the survey by \cite{HMR15} for further discussion of BSPs.

{\bf Combinatorial variants} of Pandora's Box were first studied by \cite{S17} who showed that for feasibility constraints admitting greedy-style  or ``frugal'' approximation algorithms, an extension of Weitzman's indexing algorithm provides the same approximation factor. This was further extended to CICS over Markov chains by \cite{GJSS19}.

The idea of augmenting the inspection process of a Pandora's box in order to explore more interesting decision-making settings has been a very active line of research over the last years. Arguably, the most studied variant is \textbf{optional inspection}. Most literature focuses on the single-item selection, as opposed to combinatorial variants. Study of it was initiated by \citet{GMS08}, who give a $4/5$-approximation for the maximization setting; and \citet{D18}, who characterized the solution to the single-box problem and proved certain conditions under which the Gittins policy remains optimal for single-item selection in the maximization setting (though the results naturally extend to the minimization setting). Beyond this results are separated by whether they are for the minimization or maximization setting. For the minimization setting, \citet{DS24} proved a composition theorem for a special case of local approximation (\cref{def:local-approx}) and used it to construct a committing policy with a $4/3$-approximation guarantee. In fact, their result extends beyond the single-item selection setting to the combinatorial Pandora's box setting studied by \citet{S17}. 
In the maximization setting, \citet{BK19} and \citet{GMS08} give approximation guarantees for committing policies. Furthermore, \citet{FLL22} and \citet{BC22} introduced polynomial time approximation schemes that for any $\varepsilon > 0$ compute a policy that is at least a $(1-\varepsilon)$-approximation. However, all of these results are for the single-item selection setting. For matroid feasibility constraints, \citet{BK19} observe that the commitment gap is at least $(1-1/e)$ and give an efficient policy that achieves a $0.5$ approximation.

The {\bf partial inspection} variant of Pandora's Box has also primarily been studied in the context of maximization. \cite{AJS20} provide a $(1-1/e)$-approximation via a committing policy, and show that, in fact, \textit{any} committing algorithm or its negation (flipping which box should be partially opened versus fully opened) admits a $(1/2)$-approximation to the optimal utility. We note that this already highlights a significant difference between the minimization and maximization settings for this variant. Whereas for maximization one can essentially flip a coin to decide which of the two actions to commit to, obtaining the optimal's utility for the committed action and non-negative utility for the action that was not selected, the same approach cannot be applied for minimization as the cost suffered by a bad flip could result to arbitrarily bad approximations. \citet{B19} introduces a more general inspection model, where the searcher has $k$ different methods for inspecting each box, and can select at most one of them. They provide a $(1-1/e)$-approximation that applies to $k$-element selection, but is computationally inefficient when $k$ is large. To our knowledge, partial inspection has not been studied in the context of minimization.

The \textbf{additive PB} model has been studied previously \citep{BW24,BW24_2} in the context of maximization for $k=2$ random variables per alternative with matching constraints. In this case, each alternative admits only two commitments: opening one of the two random variables followed by the other. The authors show via a global argument that picking one of the two commitments randomly independently for each alternative yields a bound of $2$ on the commitment gap (a further factor of $2$ is lost due to the matching constraint). 
For boxes with $k>2$ components each, applying the same approach implies a commitment gap of $O(k!)$. We leave as an open question whether a better bound can be obtained using our techniques. The minimization variant of this problem has not been explored in prior research.

It is noteworthy that a common approach in the literature to bounding the {\bf commitment gap} in maximization settings involves reducing the problem to finding the optimal committing policy for a stochastic submodular maximization instance and applying the bound of~\cite{AN16} on the adaptivity gap for such problems~\citep{BK19,AJS20,B19}. This method is inherently global, as it requires efficient optimization over the entire instance once mapped to a stochastic submodular maximization framework. Moreover, the optimization scales with the number of committing policies and does not extend beyond the single-selection case.

{\bf Other extensions} of Pandora's Box include settings with combinatorial rewards \citep{OW15}; combinatorial costs \citep{BEFF23}; correlated values \citep{CGTTZ19, CGMT21, GT24}; or constraints on the order of inspection \citep{EHLM19, BFLL20,BW24}. We refer the reader to the survey by \cite{BC24-survey} for further discussion of Pandora's Box.

The {\bf weighing scale} problem has not been studied previously, although similar settings where the algorithm is provided with a \textit{budget} on the total number of weighings it can perform have been considered. In particular, \citet{HSS24} study a setting, where each alternative can be weighed against the \textit{median} (or any other fixed quantile) of its distribution (conditioned on any past weighings). \citet{HS23} also consider the setting where the alternatives are identically distributed and the decision-maker is allowed to weigh them against any threshold of their choosing (much like the weighing scale problem). In both of these works, using the weighing scale is free; the algorithm is provided with a budget on the number of weighings it can perform; the objective is to find the best alternative subject to the budget; and constant factor approximations are obtained.

\paragraph{Simultaneous Work.} In simultaneous and independent work,~\cite{BLW25} also study CICS in the maximization setting. Rather than approaching the problem through the study of the commitment gap, they consider an ex-ante relaxation of the optimal policy that allows them to efficiently obtain a non-committing approximation policy in settings where the underlying combinatorial constraint admits a prophet inequality; in particular, they obtain a $0.5$-approximation under matroid constraints. A key difference between our frameworks is that (i) our results and local approximation also apply to the minimization setting, and (ii) the bounds we obtain on the commitment gap also depend on the structure of the underlying MDPs, allowing for improved approximations in some cases (such as PB with Optional Inspection).

\subsection{Outline of the rest of the paper}



In~\Cref{sec:prelims}, we provide formal definitions for CICS and the commitment gap. Our amortization framework and bound on the optimal policy for the minimization setting are presented in~\Cref{sec:amortization}; corresponding results for the maximization setting are established in~\Cref{app:maxim}. The local approximation framework and its composition to bounds for the commitment gap are presented in~\Cref{sec:local_approx}. Subsequent sections are dedicated to bounding the commitment gap for different instantiations of CICS: PB with Partial Inspection (\Cref{sec:pvo}), Additive PB (\Cref{sec:additive_pb}), Weighing Scale Problem (\Cref{sec:ws}) and Optional Inspection PB (\Cref{sec:gvo}). 


\newpage
\section{Definitions}\label{sec:prelims}

We begin by defining costly information acquisition for a random variable as a Markov decision process (MDP). 
In a Costly Information MDP, states represent the information the algorithm has gained about the corresponding random variable. Accordingly, we associate each state with the conditional value distribution it represents. Formally:

\begin{definition}[\bf Costly Information MDP]
    A Costly Information MDP for a random variable $X$ is a tuple $\mc_X=(S, \sigma, A, c, \dist, V, T)$, where $S$ is a set of states, $\sigma\in S$ is the starting state, $A(\cdot)$ maps states to sets of actions, $c(\cdot)$ is a cost function mapping actions to costs, and $\dist$ is a transition matrix. For each pair of states $s,s'\in S$ and each action $a\in A(s)$, $\dist(s,a,s')\in [0,1]$ specifies the probability of transitioning to $s'$ upon taking action $a$ in state $s$; naturally, $\sum_{s'}\dist(s,a,s') = 1$ for all states $s\in S$ and actions $a\in A(s)$. 
    
    $V(\cdot)$ is a function mapping states to distributions over values. $V(\sigma)$ is the (unconditional) distribution of $X$ and $V(s)$ is the posterior distribution of $X$ conditioned on being in state $s\in S$. As such, $V$ satisfies the rules of conditional probability: for all states $s\in S$ and all actions $a\in A(s)$, we have $V(s) = \sum_{s'} \dist(s,a,s') V(s')$. We also write $v(s) := \expect{V(s)}$. Finally, $T\subseteq S$ is the set of terminal states. Terminal states have only one action available, called the ``accept'' action. This accept action comes at no cost; results in a value of $v(s)$ at terminal state $s\in T$; and terminates the MDP. In the special case where there is only one action (accept or advance) available at every state $s\in S$, we call $\mc_X$ a Costly Information Markov chain. When clear from the context, we will drop the subscript $X$.
\end{definition}


\paragraph{Assumptions.} Note that each costly information MDP $\mc_X$ is inherently acyclic (i.e. the underlying graph is a DAG), as each costly action further refines the algorithm's knowledge of the underlying random variable. For simplicity of exposition and without loss of generality, we assume that the state spaces, action sets, and the support of the random variable $X$ are finite, although our framework and results extend seamlessly to continuous settings. We will also assume a bounded horizon. In particular, there exists a constant $H$ such that any state reached after a sequence of $H$ actions is terminal.

\begin{definition}[\bf Costly Information Combinatorial Selection] 
    The Costly Information Combinatorial Selection problem (henceforth, CICS) is defined over a ground set of $n$ nonnegative random variables, $X_1, X_2, \cdots, X_n$; a costly information MDP $\mc_i := \mc_{X_i}=(S_i, \sigma_i, A_i, c_i, \dist_i, V_i, T_i)$ for each variable $X_i$; and a feasibility constraint $\feas\subseteq 2^{[n]}$. The constraint $\feas$ corresponds to an upwards closed set in the minimization version (henceforth, min-CICS), and to a downwards closed set in the maximization version (henceforth, max-CICS). We use $\vecM := (\mc_1,\cdots , \mc_n)$ to denote the set of MDPs and $\instance = (\vecM , \feas)$ to denote the CICS instance.
\end{definition}

\paragraph{General Policies for CICS.} A policy (a.k.a. algorithm) for CICS proceeds as follows. Let $S\subseteq [n]$ denote the set of indices of all terminated MDPs.
Let $s_i$ denote the current state of the MDP $\mc_i$ at any point of time during the process. Initially, $S=\emptyset$ and $s_i=\sigma_i$ for all $i\in [n]$. The algorithm chooses at every step an index $i\in [n]\setminus S$ corresponding to a non-terminated MDP $\mc_i$ and an action $a_i\in A_i(s_i)$. It then follows the action at a cost of $c_i(a_i)$. If $a_i$ is the accept action (i.e. $s_i\in T_i$ is a terminal state), it adds $i$ to $S$ and collects the value $v_i(s_i)$. Otherwise, it updates the state of $\mc_i$ to a new state drawn from the distribution $\dist(s_i, a_i, \cdot)$, while the states of all other MDPs $\mc_{i'}$ with $i'\ne i$ remain unchanged. Observe that the algorithm can make both of these choices -- the index of the MDP to move in and the action to take in that MDP -- adaptively, depending on the evolution of all $n$ MDPs in previous steps.

\begin{itemize}
    \item For min-CICS, the algorithm terminates as soon as $S\in \feas$. The objective of min-CICS is to find an algorithm with \textit{minimum total cost}, defined as the expectation (over the randomness of the algorithm and the underlying processes) of the total cost of all actions undertaken by the algorithm until termination \textit{plus} the values accrued from accept actions. 

   \item For max-CICS, the algorithm needs to ensure feasibility by selecting at every step indices $i\in [n]\setminus S$ such that $S\cup \{i\}\in\feas$. The objective of max-CICS is to find an algorithm with \textit{maximum utility}, defined as the expectation (over the randomness of the algorithm and the underlying processes) of the total value accrued from accept actions {\em minus} the total cost of all actions undertaken by the algorithm until termination. 
\end{itemize}

\noindent We use $\opt(\instance)$ to denote the cost/utility of the optimal policy for the CICS instance $\instance = (\vecM,\feas)$.

\paragraph{Committing Policies for CICS.} 
A commitment $\pi_i$ for the MDP $\mc_i$ is a mapping from every state $s_i\in S_i$ of $\mc_i$ to a distribution $\pi_i(s_i)$ over the actions in $A_i(s_i)$. A committing policy is then characterized by a tuple of commitments $\comms=(\pi_1,\cdots , \pi_n)$; at each step, the policy selects an MDP $\mc_i$ at current state $s_i$ and takes an action that is sampled from $\pi_i(s_i)$. 
Formally, for $i\in [n]$, let $\mc_i^{\pi_i}$ denote the Markov chain resulting from applying the commitment $\pi_i$ to MDP $\mc_i$. Then any policy for the instance $\instance_{|\comms}:= (\mc_1^{\pi_1},\cdots, \mc_n^{\pi_n},\feas)$ is a committing policy for $\instance$ consistent with the tuple $\comms$. Note that while committing policies must make all local decisions as dictated by $\comms$, the index of the MDP to move in can be selected adaptively: committing policies are therefore adaptive algorithms. 

We use $\comset{\mc_i}$ to denote the set of all possible commitments $\pi_i$ for the MDP $\mc_i$ and $\comset{\instance} = \prod_i \comset{\mc_i}$ to denote the set of all possible tuples of commitments $\comms$ for the CICS instance $\instance$. The commitment gap quantifies the performance loss of committing policies relative to the optimal policy.


\begin{definition}[\bf Commitment Gap]
    The commitment gap for any min-CICS (or max-CICS) instance $\instance = (\vecM ,\feas)$ is defined as 
    \[\cg(\instance) := \min_{\comms\in\comset{\instance}}\frac{\mathrm{OPT}(\instance_{|\comms})}{\mathrm{OPT}(\instance)}\geq 1 \;\;\;\;\Big(\text{or }\;\cg(\instance):=\max_{\comms\in\comset{\instance}}\frac{\mathrm{OPT}(\instance_{|\comms})}{\mathrm{OPT}(\instance)}\leq 1\Big).
    \]
\end{definition}

\paragraph{Matroid-min-CICS.} From now on, we focus primarily on matroid feasibility constraints: each MDP corresponds to an element of some known matroid $\matroid = ([n], \mathcal{I})$. For min-CICS, $\feas$ is the collection of all sets that contain a basis of $\matroid$; for max-CICS, $\feas$ contains all the independent sets of $\matroid$. Observe that single-element selection, where the algorithm's goal is to accept one MDP, is a special case of matroid selection. Following the work of~\citet{S17},
our results will apply to any feasibility constraint that admits an efficient {\em frugal approximation algorithm}. We describe this extension in~\Cref{app:combinatorial}. For conciseness, in the main body of the paper, we state our results for the minimization setting (matroid-min-CICS); the analogous framework and results for the maximization setting are discussed in Appendix \ref{app:maxim}.

\newpage

\section{An Amortization Framework}\label{sec:amortization}

In this section, we develop a novel cost amortization technique that will allow us to lower bound the cost of the optimal adaptive algorithm for matroid-min-CICS. We first establish our amortization for the special case of Markov chains (Section~\ref{sec:mc-amortization}). We then connect our amortization framework to the notion of optimality curves (Section~\ref{sec:optimality-curves}) and use this connection to extend our approach to general MDPs (Section~\ref{sec:mdp-amortization}). All proofs are deferred to~\Cref{app:amortization}, except for the proof of \Cref{lem:mdp-wf-char-new} which can be found at the end of this section.

\subsection{Amortized Surrogate Costs for Markov Chains}\label{sec:mc-amortization}
In this section we consider instances $\instance = (\vecM ,\feas)$ of the matroid-min-CICS where every MDP $\mc_i$ is a Markov chain. We begin with some notation. For a Markov chain $\mc=(S, \sigma, A, c, \dist, V, T)$ and a state $s\in S$, we use $c(s)$ to denote the cost of the unique action in $A(s)$. A {\em trajectory} $\tau$ in this Markov chain is a sequence of states that results from advancing the chain until it terminates. For a state $s\in S$, let $\traj(s)$ denote the set of all possible trajectories starting in $s$ and let $\traj :=\traj(\sigma)$. For state $s\in S$ and trajectory $\tau\in \traj(s)$, let $p(\tau)$ denote the probability that a random trajectory starting in the initial state $\sigma$ visits the state $s$ and then continues along $\tau$ until terminating. Let $p(s):=\sum_{\tau\in\traj(s)} p(\tau)$ denote the probability of visiting $s$. For $\tau\in\traj$ and $s\in\tau$, let $\tau_s$ denote the suffix of $\tau$ starting with $s$. Finally, let $v(\tau)$ be the value $v(t)$ of the terminal state $t$ that $\tau$ ends in.

We will now define an amortization of action costs in a Markov chain over the resulting (random) trajectories that follow. The intention is that in the amortized setting, the algorithm does not incur any action costs when the action is taken but pays the amortized cost of the instantiated trajectory {\em only if} it follows the trajectory all the way to termination. 

\begin{definition}[\bf Cost Amortization]\label{def:mc-amort}
A cost amortization of a Markov chain $\mc=(S, \sigma, A, c, \dist, V, T)$ is a non-negative vector $b=\{b_{s\tau}\}_{s\in S, \tau\in \traj(s)}$ with the property that $\sum_{\tau\in \traj(s)} p(\tau)b_{s\tau}= p(s)c(s)$ for all states $s\in S$. Based on this amortization, we define:
\begin{itemize}
    \item The amortized cost of a trajectory $\tau\in \traj$ as $\surr_b(\tau) := v(\tau) + \sum_{s\in\tau} b_{s\tau_s}$. 
    
    \item The surrogate cost of the Markov chain $\mc$ as the random variable $\surr_{\mc,b}$ that takes on value $\surr_b(\tau)$ for $\tau\in \traj$ with probability $p(\tau)$.
    
    \item The index of a state $s\in S$ of the Markov chain $\mc$ as $\ind_{\mc, b}(s):=\min_{\tau\in \traj: s\in\tau}\surr_b(\tau)$.
\end{itemize}
\end{definition}
Note that the ``cost shares'' $b_{s\tau}$ distributed by a state $s$ towards its resulting trajectories fully cover the cost of $s$'s action and no more. However, in the amortized setting, these cost shares are paid only when the  trajectory terminates, and go unpaid if the algorithm stops advancing the Markov chain midway through the trajectory. We therefore obtain the following lemma.
\begin{restatable}{lemma}{mclbtag}
\label{lem:MC-lb}
    Consider any instance $\instance = (\vecM ,\feas)$  of matroid-min-CICS over Markov chains and let $b_i$ be any cost amortization of $\mc_i$ with surrogate cost $\surr_i := \surr_{\mc_i, b_i}$ for all $i\in [n]$. Then, $\opt(\instance)\geq\expect{\min_{S\in\feas} \sum_{i\in S}\surr_i}$.
\end{restatable}

We will now exhibit a specific cost amortization and a corresponding algorithm that achieves the lower bound of~\Cref{lem:MC-lb}, proving optimality. The {\bf water filling cost amortization} is described algorithmically in a bottom-up fashion. We consider the states of the chain in reverse topological order, starting from the terminals up towards the root $\sigma$. Trajectories $\tau\in\traj(t)$ for terminal states $t\in T$ are singletons and do not carry cost shares. Consider a state $s$ such that the cost shares for all states reachable from it have been determined. The state distributes its total cost $c(s)$ across its downstream trajectories $\tau\in\traj(s)$, starting from those with the \textbf{lowest current total cost}, until the equation $\sum_{\tau\in \traj(s)} p(\tau)b_{s\tau}= p(s)c(s)$ is satisfied. Formally, let $\tau' = \tau\setminus s$ be the sequence of states in $\tau$ following $s$; its current downstream cost is $\surr(\tau') = v(\tau')+\sum_{s'\in\tau'} b_{s'\tau'_{s'}}$. We find the lowest water level $g$ such that the cost shares $b_{s\tau}:=(g-\surr(\tau\setminus s))^+$ satisfy the cost equation above. The new total cost of $\tau$ becomes $\surr(\tau)  = b_{s\tau}+\surr(\tau\setminus s) = \max\{g, \surr(\tau\setminus s)\}$. We continue in this manner until $\sigma$'s cost is amortized. 

We use $W^*_\mc$ to denote the water filling surrogate cost of a Markov chain $\mc$, and $\ind^*_\mc(s)$ to denote the water filling index of a state $s$ in $\mc$. Finally, we can describe an index-based policy based on the water filling amortization.

\begin{definition}[\bf Water Filling Index Policy]
The water filling index policy for an instance $\instance = (\vecM ,\feas)$ of matroid-min-CICS over Markov chains selects at every step the Markov chain $i^*=\argmin_{i\in\feas_S} \ind^*_i(s_i)$, where $s_i$ is the current state of Markov chain $\mc_i$; $S$ is the set of terminated (selected) Markov chains; and $\feas_S=\{i: \text{rank}(S\cup\{i\})>\text{rank}(S)\}$.
\end{definition}

Note that, by definition, the indices of states encountered on any trajectory weakly increase with each next action. The water filling amortization ensures that if $b_{s\tau}>0$ for some state $s\in S$ and $\tau\in \traj(s)$, then the index of all states in the downstream trajectory $\tau$ stays {\em equal to} the index of $s$. Since the policy described above always advances the Markov chain of minimum index, this implies that for any action taken by the policy, any downstream trajectory that ``owes'' a non-zero cost share to that action will terminate with acceptance if instantiated. In other words, all the amortized cost shares are paid in expectation, establishing the following optimality result:

\begin{restatable}{theorem}{mcopttag}
\label{thm:MC-opt}
   For any matroid-min-CICS instance $\instance = (\vecM ,\feas)$ over Markov chains, the expected cost of the water filling index policy is equal 
    to $\expect{\min_{S\in\feas} \sum_{i\in S}W^*_{\mc_i}}$. The policy is therefore optimal for $\instance$.
\end{restatable}
To conclude this section, we will illustrate water filling through an example.
\begin{example}\label{ex:mdp}
Consider an MDP $\mc$ whose starting state is given a choice between following one of two Markov chains $\mc_1$ and $\mc_2$. Each Markov chain has a single costly action and results in a distribution over two distinct states, as illustrated in~\Cref{fig:mdp_example}. The action costs, terminal values, and instantiation probabilities are indicated in the figure. For each chain, the water filling amortization computes a unique index $g$ (here, $g_1=2$ and $g_2=1$), corresponding to the value for which the highlighted area equals the cost of the amortized action. Each trajectory leads to a unique terminal node, so we can associate cost shares with terminals. For each terminal $t$, the corresponding cost share is given by $(g-v(t))^+$; here, $b_{11} = 4/3$, $b_{21} = 1/2$ and $b_{12}=b_{22}=0$. Therefore, the surrogate cost $W^*_{\mc_1}$ is $2$ with probability $(3/4)$ and $4$ with probability $(1/4)$. Likewise, $W^*_{\mc_2}$ is $1$ with probability $(1/4)$ and $3$ with probability $(3/4)$. 
\end{example}
\begin{figure}[h!]
    \centering
    \includegraphics[width=0.9\textwidth]{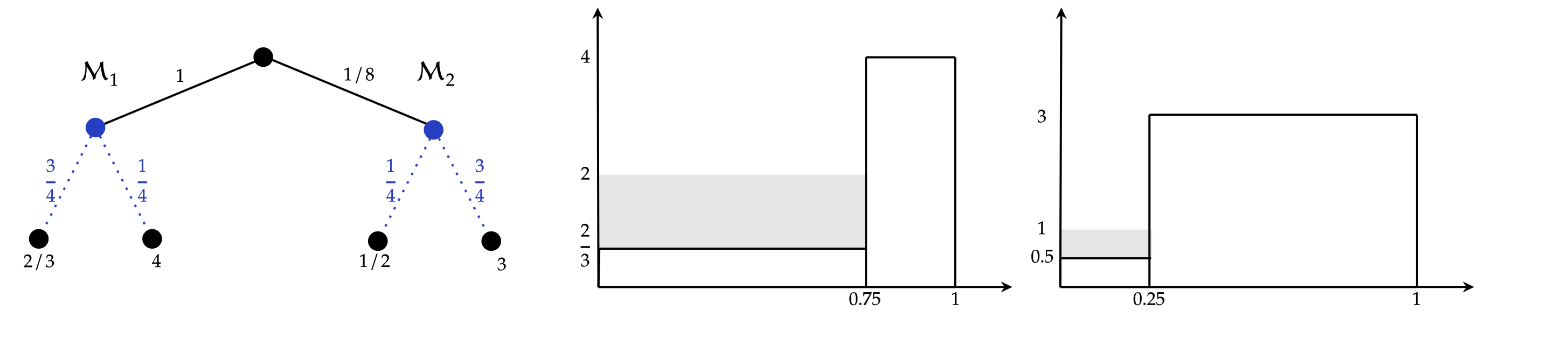}
    \caption{ The water filling amortization for Markov chains $\mc_1$ (left branch) and $\mc_2$ (right branch). The solid black edges denote costly actions (of costs $1$ and $1/8$) and the dashed blue edges denote random transitions to terminal states upon taking these actions, with the corresponding probabilities of transition. Each terminal state has a corresponding value, depicted in the bottom line. The graphs on the right display the water-filling operation for each chain with the $x$-axis representing probabilities of instantiation of the respective trajectories and the $y$-axis representing the terminal values for these trajectories. The area of the gray region equals the action cost and its height is the water level that determines cost shares. By definition of water-filling, cost shares are first transferred to the trajectories of minimum value.}
    \label{fig:mdp_example}
\end{figure}

\subsection{Local Games and Optimality Curves}\label{sec:optimality-curves}

We now connect the water filling amortization described above to the notion of optimality curves for MDPs, defined in the work of~\citet{W80} and its follow ups. We first define a ``local game'' for an MDP $\mc$.

\begin{definition}[\bf Local Game and Optimality Curve]\label{def:opt-curve}
    The local game $(\mc, y)$ is a single-selection min-CICS with two MDPs, one of which is the MDP $\mc$. The second MDP, a.k.a. the outside option, has a single terminal state with deterministic value $y$. A policy for the local game advances $\mc$ for some (zero or non-zero) number of steps and either accepts the deterministic option $y$ or the value from $\mc$. Let $\optf_\mc(y)$ denote the expected cost of the optimal policy for the local game $(\mc, y)$. We refer to the function $\optf_\mc$ as the optimality curve of the MDP $\mc$.
\end{definition}

\noindent
The surrogate cost of the outside option $y$ (under any cost amortization) is simply $y$; as a consequence of Theorem~\ref{thm:MC-opt}, we obtain the following characterization:
\begin{corollary}\label{corollary:mc-curve}
    For any Markov chain $\mc$ and any $y\in \real$, it holds that
    $\optf_\mc(y) = \expect{\min(y, W^*_\mc)}.$ 
\end{corollary}

Observe from this characterization that the CDF of the surrogate cost $W^*_{\mc}$ can be derived\footnote{In particular, let $H$ and $h$ denote the CDF and PDF of $W^*_\mc$ respectively, then we have $\optf_\mc(y) = y(1-H(y)) +\int_0^y zh(z)dz$, from which the statement follows by differentiating both sides.} from the optimality curve as $1-\dd{y}{\optf_\mc(y)}$.
Inspired by this connection, we can extend the definition of water filling surrogate costs to arbitrary MDPs:


\begin{definition}[\bf Water Filling Surrogate Costs]
\label{def:mdp-wf}
    Let $\mc$ be an MDP with optimality curve $\optf_\mc$. The water filling surrogate cost for $\mc$ is the random variable $W^*_\mc$ generated from the CDF $1-\dd{y}{\optf_\mc(y)}$. That is, $W^*_\mc$ is the random variable satisfying $\optf_\mc(y) = \expect{\min(y,W^*_\mc)}$ for all $y\in\real$.
\end{definition}

\noindent Our eventual goal is to prove the following analog of Theorem~\ref{thm:MC-opt} for general CICS over MDPs.
\begin{restatable}{theorem}{mdplbtag}
\label{thm:MDP-lb}
    In any instance $\instance = (\vecM ,\feas)$  of matroid-min-CICS, $\opt(\instance)\geq \expect{\min_{S\in\feas}\sum_{i\in S} W^*_{\mc_i}}$.
\end{restatable}
The quantity $\expect{\min_i W^*_{\mc_i}}$ is called the Whittle integral in the context of single selection over infinite-horizon discounted-reward MDPs. It was first suggested as a lower bound on the optimal cost by \cite{W80} and proved in that context by \cite{BS13} and \cite{HMR15}. Unfortunately, \Cref{def:mdp-wf} does not give us insight into how the surrogate costs $W^*_{\mc_i}$ relate to the cost of an adaptive algorithm for matroid-min-CICS. 
We will instead prove the theorem by relating optimality curves to the water filling procedure.

\subsection{Water Filling and Surrogate Costs for General MDPs}\label{sec:mdp-amortization}

In this section we will prove~\Cref{thm:MDP-lb}. In the case of Markov chains, the water filling surrogate cost of a randomly sampled trajectory recovers the optimality curve of the Markov chain. The challenge to extending this argument for general MDPs is that each sequence of actions creates a different distribution over the trajectories and consequently also over the surrogate costs, as illustrated in \Cref{ex:mdp}. 

Our main observation is that if we are willing to cover action costs only partially through the amortization, then we can define an appropriate amortization for {\em every possible} deterministic commitment in $\mc$ such that the distribution over surrogate costs generated by that commitment is independent of the commitment itself. In other words, {\em every} sequence of actions in the MDP generates the {\em same} distribution over surrogate costs. This distribution is a fundamental property of the MDP itself, and not of the algorithm operating on the MDP. We formalize this property through the following lemma.

\begin{lemma}
\label{lem:mdp-wf-char-new}
    For any MDP $\mc=(S, \sigma, A, c, \dist, V, T)$ and any deterministic commitment $\pi\in\comset{\mc}$, generating a Markov chain $\mc^\pi$ with states $S_\pi\subseteq S$, realizable trajectories $\traj_\pi\subseteq\traj$ and a distribution $p_\pi$ over trajectories and reachable states, there exists an amortized cost function $\surr^\pi:\traj_\pi \mapsto \Delta(\real)$ mapping trajectories to distributions over costs and a non-negative cost sharing vector $b^\pi = \{b^\pi_{s\tau}\}_{s\in S_\pi, \tau\in \traj_\pi(s)}$, such that the following properties hold:
    \begin{enumerate}
    \smallskip
    \item \textbf{\emph{Cost Sharing:}} the amortized cost of a trajectory pays for its own acceptance value and the cost shares it sends to upstream states. For all $\tau\in \traj_\pi$, it holds that $\expect{\surr^\pi(\tau)} = v(\tau) + \sum_{s\in \tau} b^\pi_{s\tau_s}$.
    \smallskip

    \item \textbf{\emph{Cost Dominance:}} the cost shares received by any state pay towards its action cost, but do not overpay. For all $s\in S_\pi$, it holds that $\sum_{\tau\in \traj_\pi(s)} p_\pi(\tau) b^\pi_{s\tau}\le p_\pi(s)c(\pi(s))$.
    \smallskip

    \item \textbf{\emph{Action Independence:}} sampling a trajectory $\tau\sim p_\pi$ and then sampling from the amortized cost distribution $\surr^\pi(\tau)$ generates a random surrogate cost for the MDP. This random variable is distributed identically to the water filling surrogate cost $W^*_\mc$. 
\end{enumerate}
\end{lemma}

As mentioned earlier, the amortized costs $\surr^\pi(\tau)$ defined by the lemma are necessarily different from the water filling costs in the respective Markov chains $\mc^\pi$ as they must satisfy action independence. 
Observe that in the absence of the last condition connecting the surrogate costs $\surr^\pi(\tau)$ to the water filling costs $W^*_\mc$ as defined in~\Cref{def:mdp-wf}, the lemma is trivial to satisfy. Indeed we can define all of the cost shares to be $0$ and still satisfy cost dominance and action independence. The key contribution of the lemma is to show that we can achieve these properties while recovering the ``optimal'' surrogate costs as determined by the optimality curve $\optf_\mc$. \Cref{fig:mdp-ex-2} illustrates how the lemma applies to the MDP in \Cref{ex:mdp}. 

Once we establish the lemma, the proof of~\Cref{thm:MDP-lb} follows in much the same way as the proof of~\Cref{lem:MC-lb}. At a high level, we can view the algorithm's choices {\em in hindsight} as corresponding to some deterministic commitment, and relate the algorithm's cost to the surrogate costs as instantiated through that commitment; the details are deferred to the appendix. We conclude this section with the proof of Lemma~\ref{lem:mdp-wf-char-new}. A casual reader can safely skip this proof and proceed to \Cref{sec:local_approx}.

\begin{figure}[h!]
\centering
\includegraphics[width=0.95\linewidth,valign=m]{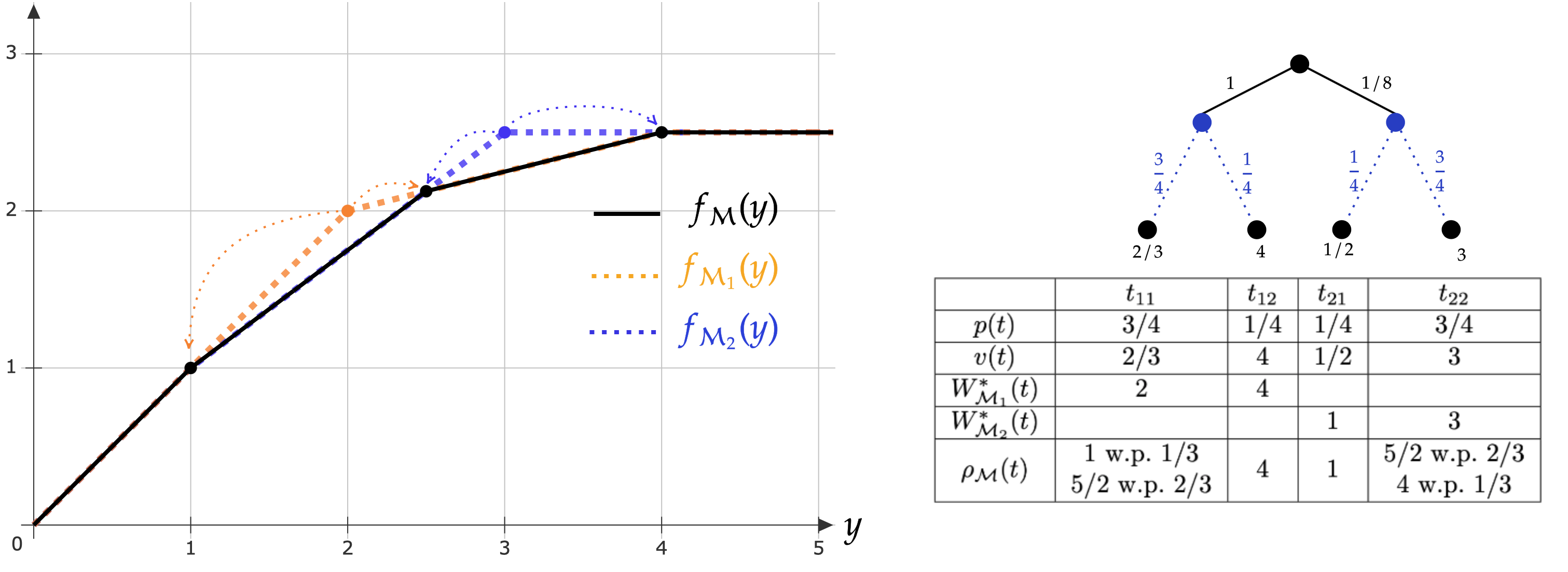} 
\caption{We consider the same setting as in~\Cref{ex:mdp}. The optimality curve for MDP $\mc$ is obtained by the minimum of the optimality curves for $\mc_1$ (orange dashed line) and $\mc_2$ (blue dashed line). Each action/commitment leads to two possible trajectories. The table on the right displays the four trajectories along with their respective probabilities of instantiation, values, and amortized costs under their respective commitments (commitment 1 for the first two columns and commitment 2 for the last two columns). Observe that resulting surrogate costs have different distributions. The last row in the table defines an alternate set of (random) surrogate costs $\surr_\mc(t)$. Observe that for each $t$, we have $\operatorname{E}[\surr_\mc(t)]\le v(t)$. This implies cost dominance. 
The distribution of $\surr_\mc(t)$ when $t$ is picked according to commitment 1 is $1$ with probability $(3/4)\cdot(1/3) = (1/4)$, $2.5$ with probability $(3/4)\cdot(2/3) = (1/2)$ and $4$ with probability $(1/4)\cdot(1) = (1/4)$. We obtain the same distribution if $t$ is picked according to commitment 2. The same distribution also arises from the solid black optimality curve in the graph on the left, exhibiting the action independence property.
\label{fig:mdp-ex-2}}
\end{figure}

\paragraph{Proof of~\Cref{lem:mdp-wf-char-new}.}
Fix any MDP $\mc=(S, \sigma, A, c, \dist, V, T)$. The proof follows by induction on the horizon $H$ of the MDP. If $H=1$, then $\mc$ is a singleton MDP where $\sigma\in T$ is the unique state, $\pi_\sigma$ is the unique (trivial) commitment and $\traj_{\pi_\sigma} = \{\{\sigma\}\}$. Therefore, the lemma is immediately satisfied by $\surr^{\pi_\sigma}(\{\sigma\}):=v(\sigma)$ and $b^{\pi_{\sigma}}_{\sigma\{\sigma\}}:=0$. We now assume that the conditions of the lemma hold for all MDPs of horizon up to $H$ and extend it to MDPs of horizon $H+1$.

Let $\mc$ be of horizon $H+1$ and let $A(\sigma)=\{a_1,\cdots , a_k\}$ be the set of actions available at the starting state $\sigma$. Let $R_j$ denote the set of all states that can be reached directly by taking action $a_j$ on $\sigma$; that is, $R_j:=\{s\in S: \dist(\sigma, a_j, s) > 0\}$. For each $s\in R_j$, we use $\mc_s$ to denote the subprocesses of $\mc$ that starts in state $s$ and $W^*_{\mc_s}$ to denote the water filling surrogate cost of $\mc_s$. We begin by computing the water filling amortization of each action $a_j$.

Let $Z_j$ denote the random variable that first draws a state $s$ from the distribution $\dist(\sigma, a_j, \cdot)$ and then draws a value from the distribution $W^*_{\mc_s}$. Let $g_j$ be the solution to the equation
\[c(a_j) + \expect{Z_j} = \expect{\max\{g_j, Z_j\}}.\]
Then, $\hat{Z}_j := \max\{g_j, Z_j\}$ is the water filling surrogate cost of the action $a_j$. We observe that for any $y\ge g_j$, we have $c(a_j)+\expect{\min\{y,Z_j\}}=\expect{\min\{y, \hat{Z}_j\}}$, and because $\hat{Z}_j$ is always at least $g_j$, we get for all $y$:
\begin{align}
\label{eq:lem2-1}
\min\{y, c(a_j)+\expect{\min\{y,Z_j\}}\}=\expect{\min\{y, \hat{Z}_j\}}.
\end{align}
\noindent
We can now write the optimality curve of $\mc$ as:
\begin{align*}
\optf_\mc(y) &= \min\left\{y, \min_{j\in [k]} \bigg(c(a_j) + \sum_{s\in R_j}\dist(\sigma,a_j,s)\cdot f_{\mc_s}(y)\bigg)\right\}\\
    &= \min\left\{ y , \min_{j\in [k]} \bigg(c(a_j) + \sum_{s\in R_j}\dist(\sigma,a_j,s)\cdot \expect{\min\{y, W^*_{\mc_s}\}}\bigg)\right\} \\
    &= \min\left\{y, \min_{j\in [k]}\bigg( c(a_j) + \expect{\min\{y, Z_j\}}\bigg)\right\} 
    = \min\left\{y, \min_{j\in [k]} \expect{\min\bigg(y, \hat{Z}_j\bigg)}\right\}.
\end{align*}
Here the first equation follows from noting that in the local game $(\mc, y)$, the algorithm can either choose the outside option $y$ or takes one of the actions $a_j$ from $\sigma$ and then proceeds optimally in the game $(\mc_{s}, y)$ where $s$ is instantiated from $R_j$. The second equation follows by the definition of $W^*_{\mc_s}$; the third by the definition of $Z_j$; and the fourth by Equation~\eqref{eq:lem2-1}. 

Recall that $\optf_\mc(y) = \expect{\min\{y, W^*_\mc\}}$, and so, we conclude that \[\expect{\min\{y, W^*_\mc\}}\le \expect{\min\{y, \hat{Z}_j\}}\] for all $j\in [k]$. This implies that the random variable $\hat{Z}_j$ second-order stochastically dominates the random variable $W^*_\mc$, allowing us to use the following lemma. 
\begin{restatable}{lemma}{stdomtag}
\label{lemma:sdom}
\textrm{\normalfont (Second Order Stochastic Dominance.)} 
Let $X, Z$ be discrete random variables that satisfy the property $\expect{\min\{y, X\}} \leq \expect{\min\{y, Z\}}$ for all $y\in\real$. There exists a mapping $m: \mathrm{supp}(Z) \mapsto \Delta(\mathrm{supp}(X))$ from the support of $Z$ to distributions over the support of $X$ such that:
\begin{enumerate}
    \item $X$ is obtained by sampling from $m(z)$ for a randomly sampled $z\sim Z$.
    \item For all $z\in\mathrm{support}(Z)$, it holds that $\expect{m(z)}\leq z$.
\end{enumerate}
\end{restatable}
\noindent
We note that the lemma is standard (see for example~\citep{S65,FS16}) but we provide a constructive proof in~\Cref{app:amortization} for the sake of intuition and completeness. 
We apply~\Cref{lemma:sdom} to all tuples $(W^*_\mc , \hat{Z}_j)$ to obtain mappings $m_j(\cdot)$.

We are finally ready to define the amortized cost functions $\surr^\pi$ and the cost sharing vectors $b^\pi$. Fix any deterministic commitment $\pi\in\comset{\mc}$ and let $j = \pi(\sigma)\in [k]$ be the fixed action that $\pi$ takes at state $\sigma$. For each state $s\in R_j$, we use $\pi|s\in\comset{\mc_s}$ to denote the commitment $\pi$ after we transition to state $s$; note that this is also a deterministic commitment. By definition, each MDP $\mc_s$ has horizon up to $H$ and thus by the induction hypothesis, it admits a mapping $\surr^{\pi|\sigma}(\cdot)$ and a non-negative cost sharing vector $b^{\pi|\sigma}$, satisfying the properties of the lemma. Now, consider any trajectory $\tau\in\traj_\pi$ and observe that $\tau = \{\sigma, \tau_s\}$ for some $s\in R_j$ and some $\tau_s\in\traj_{\pi}(s)$. We define
\[\surr^\pi(\tau) := m_j(\max\{g_j, \surr^{\pi|s}(\tau_s)\})\] 
and
\[b^\pi_{\sigma \tau} := \expect{\surr^\pi(\tau)}-\expect{\surr^{\pi|s}(\tau_s)}\]
and for all other $s\in S_\pi\setminus\{\sigma\}$ and $\tau\in\traj_\pi(s)$, we use the same cost share $b^\pi_{s\tau} = b^{\pi|s}_{s\tau}$ that was used in $\mc_s$.

\medskip

\begin{itemize}   
    \item {\bf Cost sharing:} This holds trivially by the definition of the cost shares in the starting state $\sigma$, the fact that we don't change the cost shares in any other state $s\neq \sigma$, the fact that $v(\tau)=v(\tau_s)$ for any suffix $\tau_s$ of $\tau$ and the induction hypothesis.
    \smallskip
    \item {\bf Cost dominance:} By the induction hypothesis and the fact that we maintain the cost shares of every state $s\neq \sigma$, we only need to prove the inequality for $\sigma$. We have
    \begin{align*}
        \sum_{\tau \in \traj_\pi}p_\pi(\tau)b^\pi_{\sigma \tau} &=
        \sum_{s\in R_j, \tau_s\in \traj_\pi(s)} p_\pi(\{\sigma, \tau_s\})(\expect{\surr^\pi(\{\sigma,\tau_s\})}-\expect{\surr^{\pi|s}(\tau_s)})\\
        & = \sum_{\tau \in \traj_\pi}p^\pi(\tau)\expect{\surr^\pi (\tau)} - \sum_{s\in R_j} \dist(\sigma, a_j, s) \sum_{\tau_s\in \traj_{\pi}(s)} p^{\pi|s}(\tau_s)\expect{\surr^{\pi|s}(\tau_s)}\\
        & = \expect{W^*_\mc} - \sum_{s\in R_j} \dist(\sigma, a_j, s) \expect{W^*_{\mc_s}}\\
        & \le \expect{\hat{Z}_j} - \expect{Z_j} = c(a_j).
    \end{align*}
     and the statement follows by noting $p^\pi(\sigma)=1$. Here the first equation follows from the definition of $b^\pi_{\sigma \tau}$ and a decomposition of the trajectories $\tau$; the second just rewrites the terms separately; the third is by the action independence of $\surr^\pi$ (proven below), and by the induction hypothesis applied to $\mc_s$ similarly; the fourth uses property (2) in Lemma~\ref{lemma:sdom} for the first term and the definition of $Z_j$ for the second term; and the last equality follows from the definition of $\hat{Z}_j$. 
    \smallskip
    \item {\bf Action independence:} Drawing a trajectory $\tau\sim p_\pi$ is equivalent to first drawing a state $s\in R_j$ from the distribution $\dist(\sigma, a_j,\cdot)$ and then drawing a trajectory $\tau_s\in\traj_\pi(s)$ from $p_{\pi|s}$. Consider drawing $\tau$ in this manner and then sampling from the distribution $\surr^{\pi|s}(\tau_s)$. By the induction hypothesis and the definition of $Z_j$, this provides us with a sample drawn from $Z_j$. Then, $\max\{g_j,\surr^{\pi|s}(\tau_s)\}$ with $\tau$ drawn in this manner corresponds to an instantiation of $\hat{Z}_j$, and $m_j$ applied to that instantiation results in an instantiation of $W^*_\mc$ by the definition of $m_j$ and property (1) in Lemma~\ref{lemma:sdom}. 
\end{itemize}
This concludes the proof of Lemma~\ref{lem:mdp-wf-char-new}.

\newpage

\section{Local Approximation and Composition Theorems}\label{sec:local_approx}

Recall that the commitment gap of a CICS instance is defined as the ratio between the cost of the optimal committing policy and the cost of the optimal (non-committing) policy. Let $\instance=(\vecM , \feas)$ be an instance of matroid-min-CICS and let $\pi_i$ be a commitment for $\mc_i\in\vecM$. We will establish bounds on the commitment gap achieved by $\comms=(\pi_1, \cdots, \pi_n)$ by quantifying for each $i\in [n]$ the performance of $\pi_i$ in the local game $(\mc_i, y)$. 

Recall that the performance of an optimal algorithm for the local game $(\mc, y)$ is described by the optimality curve $\optf_\mc(y)$. The performance of the commitment $\pi$ in the same game is given by $\optf_{\mc^\pi}(y)$. \cite{DS24} define local approximation by relating these two quantities:\footnote{Their definition is provided in the context of optional-inspection PB, and we extend the same definition to general MDPs.}
\begin{definition}[\textbf{Local Approximation}]\label{def:local-approx}
    Let $\mc$ be any MDP and let $\alpha\geq 1$. We say that a commitment $\pi\in\comset{\mc}$ is an $\alpha$-local approximation for $\mc$ if 
    $\optf_{\mc^\pi}(\alpha y)\leq \alpha\cdot \optf_\mc(y)$ for all $y\in\real$.
\end{definition}
Restating the inequality in terms of the water-filling costs provides a more intuitive definition as well as a composition theorem. Theorems~\ref{thm:MC-opt} and~\ref{thm:MDP-lb} together imply that we can bound the commitment gap as:
\[\cg(\instance) \leq \min_{(\pi_1, \cdots, \pi_n)\in\comset{\instance}}\frac{\expect{\min_{S\in\feas}\sum_{i\in S} W^*_{\mc^{\pi_i}_i}}}{\expect{\min_{S\in\feas}\sum_{i\in S} W^*_{\mc_i}}}.\]
On the other hand, using~\Cref{def:mdp-wf}, an $\alpha$-local approximation for $\pi_i$ with respect to $\mc_i$ can be restated as a second-order stochastic dominance condition: 
\begin{equation}
\label{eq:why-horizontal-scaling}
\expect{\min\{y, W^*_{\mc_i^{\pi_i}}\}} \leq \expect{\min\{y, \alpha W^*_{\mc_i}\}} \;\;\forall y\in\real.
\end{equation}

\noindent Combining these inequalities provides the following composition result.
\begin{restatable}{theorem}{loccomptag}\label{thm:la-comp}
Let $\instance = (\vecM , \feas)$ be any instance of matroid-min-CICS, where each constituent MDP $\mc_i$ admits an $\alpha$-local approximation under some commitment $\pi_i\in\comset{\mc_i}$. Then, $\cg (\instance)\leq \alpha$.
\end{restatable}

\noindent The proof of~\Cref{thm:la-comp} is straightforward from our previous discussion and is presented in~\Cref{app:local_approx}. The theorem also provides a recipe for designing approximation algorithms for CICS: if we can efficiently identify a commitment $\pi_i$ for each MDP $\mc_i$ that achieves an $\alpha$-local approximation, then we can simply restrict each MDP $\mc_i$ to the corresponding Markov chain $\mc_i^{\pi_i}$ and run the (efficient) water filling index policy; the resulting committing policy is guaranteed to be an $\alpha$-approximation to the unrestricted optimum.

\paragraph{Geometric Intuition for Local Approximation.}
It is worth emphasizing that $\alpha$-local approximation is \emph{not equivalent} to simply achieving a standard $\alpha$-approximation in the local game. Slightly rearranging the condition in \cref{def:local-approx}, we have that for commitment $\pi$ to be an $\alpha$-local approximation, we require for all~$y \in \bbR$,
\[
    f_{\mc^\pi}(y) \leq \alpha f_\mc(y/\alpha)
    .
\]
In contrast, for commitment $\pi$ to be a standard $\alpha$-approximation for the local game, we require only
\[
    f_{\mc^\pi}(y) \leq \alpha f_\mc(y)
    ,
\]
the difference being $y$ instead of $y/\alpha$ on the right-hand side.
As illustrated in \cref{fig:standard-vs-local}, we can understand both of the above conditions as saying that the graph of $f_{\mc^\pi}(y)$ needs to lie below a geometrically scaled version of the graph of $f_\mc(y)$. However, the scaling used is different for the two conditions.
\* For standard $\alpha$-approximation, we vertically scale up $f_\mc(y)$ by a factor of~$\alpha$.
\* For $\alpha$-local approximation, we vertically \emph{and horizontally} scale up $f_\mc(y)$ by a factor of~$\alpha$.
\*/
The additional horizontal scaling in local approximation is critical in proving \cref{eq:why-horizontal-scaling} and thereby obtaining \cref{thm:la-comp}. Standard $\alpha$-approximation implies only a weaker version of \cref{eq:why-horizontal-scaling} with the $y$ on the right-hand side replaced by $\alpha y$. Working through the proof of \cref{thm:la-comp} with this change, one would eventually obtain a commitment gap bound of $\alpha^n$, scaling exponentially\footnote{\Cref{app:example} exhibits an explicit example of this exponential scaling.} in the number of MDPs~$n$. Using the right local approximation notion is thus critical for obtaining approximation guarantees that do not depend on~$n$.

\begin{figure}
    \centering%
    \newcommand{\drawApproxGraph}[5]{%
        \begin{tikzpicture}[scale=0.3]
            \node[align=center] at (7, 10) {\itshape#5};
            \draw[<->] (0, 7) -- (0, 0) -- (14, 0) node[right] {$y$};
            \draw[ultra thick, #3] (0, 0) -- (2, 2) to[out=25, in=190] (8, 3.75) node[right] {$f_\mc(y)$}; 
            \ifblank{#1}{}{
                \begin{scope}
                    \ifstrequal{#1}{1}{
                        \newcommand{\overAlpha}{}
                        \newcommand{\nodePlace}{right}
                    }{
                        \newcommand{\overAlpha}{/\alpha}
                        \newcommand{\nodePlace}{above left}
                    }
                    \begin{scope}[xscale=#1, yscale=#2]
                        \draw[ultra thick, #4] (0, 0) -- (2, 2) to[out=25, in=190] (8, 3.75) node[\nodePlace] {$\alpha f_\mc(y \overAlpha)$}; 
                    \end{scope}
                    \end{scope}
            }
        \end{tikzpicture}%
        \ignorespaces%
    }%
    \drawApproxGraph{}{}{teal}{}{Original Cost Curve}
    \hfill
    \drawApproxGraph{1}{1.7}{teal!30}{purple}{Standard Approximation:\\vertical scaling}
    \hfill
    \drawApproxGraph{1.7}{1.7}{teal!30}{purple}{Local Approximation:\\diagonal scaling}
    \caption{For an MDP $\mc$ with a given optimality curve (left), we contrast two approximation notions: standard $\alpha$-approximation (center) vs. $\alpha$-local approximation (right). In both cases, achieving the approximation requires designing a commitment whose optimality curve lies below a scaled version of the optimal cost (purple curves). For standard $\alpha$-approximation, the curve is scaled up vertically by a factor of~$\alpha$. For $\alpha$-local approximation, the curve is scaled up \emph{diagonally}, i.e. both vertically and horizontally, by a factor of~$\alpha$.}
    \label{fig:standard-vs-local}
\end{figure}

\paragraph{Establishing Local Approximation.} We note that in some cases, the surrogate costs of an MDP might not admit a simple structure, making second-order stochastic dominance conditions difficult to establish. We can alternatively establish local approximation through a stronger but conceptually easier relationship between the surrogate costs: namely, that $\alpha W^*_\mc$ first-order stochastically dominates $W^*_{\mc_\pi}$. Formally, we define the following alternate notion of approximation, where for a quantile $q\in [0,1]$ and a random variable $X$, we let $X(q)$ denote the $q$th quantile value of the r.v.: $X(q) = \inf\{x: Pr[X\le x]\ge q\}$.

\begin{definition}[\textbf{Pointwise  Approximation}]\label{def:pw-approx}
    Let $\mc$ be any MDP and let $\alpha\geq 1$. We say that a commitment $\pi\in\comset{\mc}$ is an $\alpha$-pointwise approximation for $\mc$ if $W^*_{\mc^\pi}(q)\le \alpha\cdot W^*_{\mc}(q)$ for all $q\in [0,1]$.
\end{definition}

\noindent
The following is immediate from the fact that first-order stochastic dominance implies second-order stochastic dominance, but the converse need not hold.

\begin{restatable}{fact}{pwapproxtag}\label{lemma:pw-to-local}
    If $\pi\in\comset{\mc}$ is an $\alpha$-pointwise approximation for $\mc$, it also an $\alpha$-local approximation for $\mc$.
\end{restatable}

Finally, we note that by definition we have $\optf_\mc(y) = \min_{\pi\in\comset{\mc}} \optf_{\mc^\pi}(y)$. Moreover, because randomized commitments generate optimality curves as averages over those of deterministic commitments, we can write $\optf_\mc(y)$ as the minimum over the optimality curves of deterministic commitments. Thus, we can obtain local or pointwise approximation guarantees by showing that the commitment $\pi$ satisfies the desired definition against \textit{all other deterministic commitments} $\pi'\in\comset{\mc}$. This simple observation can prove essential when establishing these conditions, as the water-filling surrogate costs of Markov chains have a much more intuitive structure (\Cref{lem:MC-lb}) than the surrogate costs of MDPs (\Cref{lem:mdp-wf-char-new}). For example, while the surrogate costs of the MDPs corresponding to Additive PBs (\Cref{sec:additive_pb}) are very complicated, the surrogate costs of committing policies for the same setting admit an exceptionally simple form that can be exploited to prove a pointwise approximation.

\paragraph{Beyond Local Approximation.}
For one of our applications, namely the maximization version of Pandora's box with optional inspection (\cref{sec:gvo}), local approximation turns out to be too stringent of a condition. As such, we introduce a slight weakening of local approximation, which we call \emph{semilocal approximation} (\cref{def:semilocal_approx}), that is tailored to this application. Roughly speaking, local approximation requires purely multiplicative suboptimality, while semilocal approximation allows for some additive suboptimality, too. See \cref{sec:gvo} for details.

\newpage

\section{Pandora's Box with Partial Inspection (Minimization)} \label{sec:pvo}

In classical Pandora's Box, the algorithm pays some cost $c^o$ to open a box and observe its value, and is then allowed to select the box. We will now study the minimization version of the generalization called
Pandora's Box with Partial Inspection (henceforth, PBPI) where the algorithm can additionally ``peek'' into the box at a smaller cost $c^p<c^o$ and learn its value. If upon obtaining this information, the algorithm wants to select the box, it must still open the box at a cost of $c^o$ before accepting it. This presents a choice: in some cases it may be better to open the box outright, while in others it is better to pay the smaller peeking cost to potentially avoid paying the opening cost later.

Formally, we consider an instance with $n$ partial inspection boxes (henceforth, PI-boxes) $\{\pbox_i\}_{i=1}^n$, 
where each box $\pbox_i = (\dist_i, c^o_i, c^p_i)$ is characterized by a distribution $\dist_i$ over values; an opening cost $c^o_i> 0$; and a peeking cost $c^p_i\in (0, c^o_i]$\footnote{Since we can only select an opened box, if $c^p_i\geq c^o_i$, then we can safely exclude the peeking action and the box reduces to a classical Pandora's Box.}. Each PI box $\pbox_i$ can be expressed as a Costly Information MDP with two possible actions, peeking and opening, where we can interpret the accept action after peeking as incurring a cost of $c^o$. An instance of PBPI corresponds to an instance of min-CICS over the corresponding MDPs. A pictorial representation of the different states of a box and the underlying MDP is shown in~\Cref{fig:pvo-states}.

\begin{figure}[h!]
    \centering
    \includegraphics[width=0.9\textwidth]{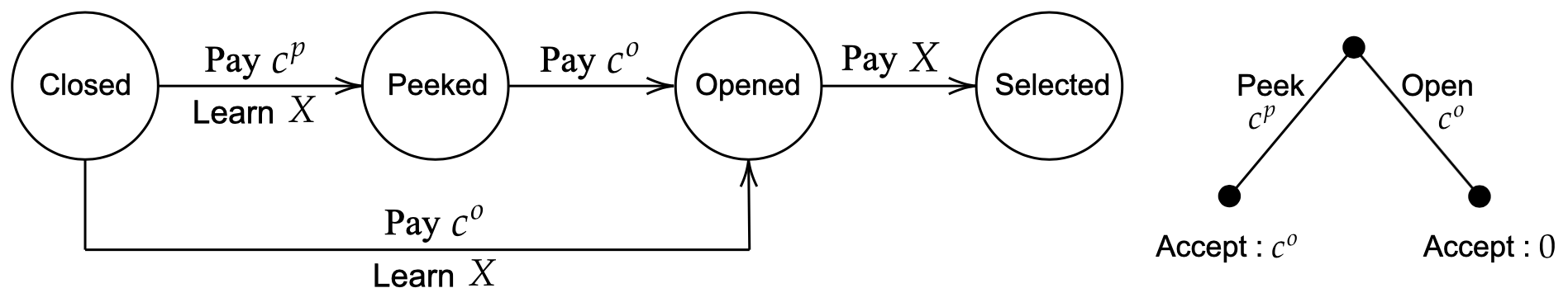}
    \caption{In PBPI, a peeking box $\pbox = (\dist , c^o, c^p)$ is initially closed. In order to learn the value realization realization $X\sim\dist$, the decision maker can either peek into the box (at a cost of $c^p$) or open it (at a cost of $c^o$). To accept the box, the decision maker must first open it and then pay its (now known) value $X$.}
    \label{fig:pvo-states}
\end{figure}





\noindent A committing policy for PBPI needs to decide in advance (perhaps randomly) whether each box $\pbox_i$ will be directly opened or peeked into and then (potentially) opened. Our main result in this section is the following:

\begin{theorem}\label{thm:pvo-approximation}
The commitment gap of matroid-PBPI is at most $\sqrt{2}$.
\end{theorem}

In Section~\ref{sec:pvo-bounds} we show a lower bound of $16/15$ for the commitment gap of matroid-PBPI (\Cref{example:PBPI-adaptivity-gap}) and a trivial upper bound of $\phi\approx 1.618$ (\Cref{claim:pvo-global-approx}). We then prove~\Cref{thm:pvo-approximation} by showing that each PI box $\pbox$ admits a $\sqrt{2}\approx 1.414$ local approximation  (\Cref{lemma:pvo-local-approx}) and employing~\Cref{thm:la-comp}; we note that the commitment achieving this guarantee decides whether each box will be directly opened or peeked before opened \textit{deterministically}, based on the parameters of the box, and can be efficiently computed.

\subsection{Lower and Upper Bounds on the Commitment Gap of PBPI}\label{sec:pvo-bounds}

We begin by showing that that there are simple instances where all committing policies are sub-optimal up to a constant factor.

\begin{example}\label{example:PBPI-adaptivity-gap}
    Consider an instance of single-selection PBPI over two boxes. The first box has an opening cost of $1$, a peeking cost of $\frac{1}{4}$ and its random value is $0$ with probability $\frac{1}{2}$ and $2$ otherwise. The second box has both an opening and peeking cost of $0$, and its random value is $2$ with probability $\frac{1}{2}$ and $\infty$ otherwise.

    \begin{itemize}
        \item Since opening the second box is free, we can assume without loss that all policies start by opening it and observing its value, call it $y$. The optimal policy will simply open and accept the first box if $y=\infty$. But if $y=2$, it can peek into the first box and only open it if it contains a value of $0$. The expected cost of this algorithm is $15/8$.

        \item Now consider any policy that commits to opening or peeking into box $1$ {\em before} observing $y$. If it commits to opening the first box, it accepts the value of this box regardless of $y$, as box $1$ always has a value smaller than $y$. The expected cost of this policy is $2$. If it commits to peeking into the first box, then at $y=\infty$ it incurs a cost of $9/4$ due to having to pay the extra peeking cost; and at $y=2$ it incurs an expected cost of $7/4$ by opening the first box only if it contains a value of $0$. It's net cost is again $2$.
    \end{itemize}
    \noindent 
    Since both deterministic committing policies have an expected cost of $2$, so does any randomized commitment. Consequently, the expected cost of the optimal adaptive policy is strictly smaller than that of the optimal committing policy.
\end{example}

\Cref{example:PBPI-adaptivity-gap} illustrates that the commitment gap of matroid-PBPI at least $(16/15)$. On the positive side, it is easy to obtain an upper bound of $2$. In particular, consider the policy that commits to peeking all boxes. This policy can mimic the optimal one as follows. Whenever the optimal algorithm peeks, so does this committing policy. Whenever the optimal algorithm opens without peeking, the committing policy peeks and then opens; on a box with costs $c^p_i$ and $c^o_i$, this policy pays $c^p_i+c^o_i$ or at most twice the amount $c^o_i$ paid by the optimal algorithm. In fact, we can further refine this argument by choosing the action we commit to more carefully: there always exists a simple commitment under which we can achieve a $\phi\approx 1.618$ approximation to the optimal adaptive policy. The proof is presented in~\Cref{app:pvo}.

\begin{restatable}{lemma}{lemmaphiapprox}\label{claim:pvo-global-approx}
    Consider any instance of matroid-PBPI and partition the $n$ boxes into two sets \[O:=\left\{i\in [n] : \frac{c_i^o}{c_i^p} \leq 1 + \frac{c_i^p}{c_i^o}\right\}\] and $P=[n]\setminus O$. The policy that commits to directly opening the boxes in $O$ and peeking before opening the boxes in $P$ achieves a $\phi$-approximation to the optimal (non-committing) policy.
\end{restatable}

We note that by using a \textit{global argument} such as the one above, i.e. an argument where we charge each action of the committing policy directly to the optimal adaptive policy, one cannot improve on this upper bound of $\phi$ (for example, consider a setting where all the boxes have $c^p=1$ and $c^o=\phi$). In order to achieve a better guarantee, we would need to leverage our knowledge of the value distributions. In the next section, we achieve this by establishing local approximation guarantees for PBPI.

\subsection{Local Approximation Guarantees for PBPI}\label{sec:pvo-local-approx}

Given a PI box $\pbox= (\dist , c^o , c^p)$, let $g^p$ denote the water filling (a.k.a. Gittins) index of the policy that commits to peeking. Equivalently, $g^p$ is the solution to the equation $c^p = \expectt{X\sim\dist}{(g^p-X-c^o)^+}$. We call $g^p$ the \textbf{peeking index} of the box.


\begin{lemma}\label{lemma:pvo-local-approx}
    Let $\pbox = (\dist, c^o, c^p)$ be a PI box with peeking index $g^p$. Let $\pi$ commit to the opening action whenever \[\frac{c^o}{c^p}\cdot \left(1-\frac{c^o}{g^p}\right) \leq 1+ \min\left(\frac{c^p}{c^o},\frac{c^o}{g^p}\right)\]
    and to the peeking action otherwise. Then, $\pi$ is a $\sqrt{2}$-local approximation to $\pbox$. 
\end{lemma}



\noindent
The rest of this section is devoted to proving Lemma~\ref{lemma:pvo-local-approx}. We first note the following characterization:
\begin{definition}[\textbf{Optimality Curves for PBPI}]
    The optimality curve of a PI box $\pbox = (\dist, c^o, c^p)$ for cost minimization is given by
    \[\optf_{\pbox}(y):= \min\{ y, \optf^o_\pbox(y), \optf^p_\pbox(y)\}\]
    where $\optf^o_\pbox(y):= c^o + \expectt{X\sim\dist}{\min\{y,X\}}$ is the optimality curve of the policy that commits to opening the box and $\optf^p_\pbox (y):= c^p + \expectt{X\sim\dist}{\min\{y, X+c^o\}}$ is the optimality curve of the policy that commits to peeking. We also define the \textbf{opening index} of the box, $g^o$, as the water filling index of the opening policy, equivalently, the solution to the equation $c^o = \expectt{X\sim\dist}{(g^o-X)^+}$. We depict these curves and indices in~\Cref{fig:pvo-curves}.
\end{definition}

\begin{figure}[h!]
    \centering
    \includegraphics[width=0.9\textwidth]{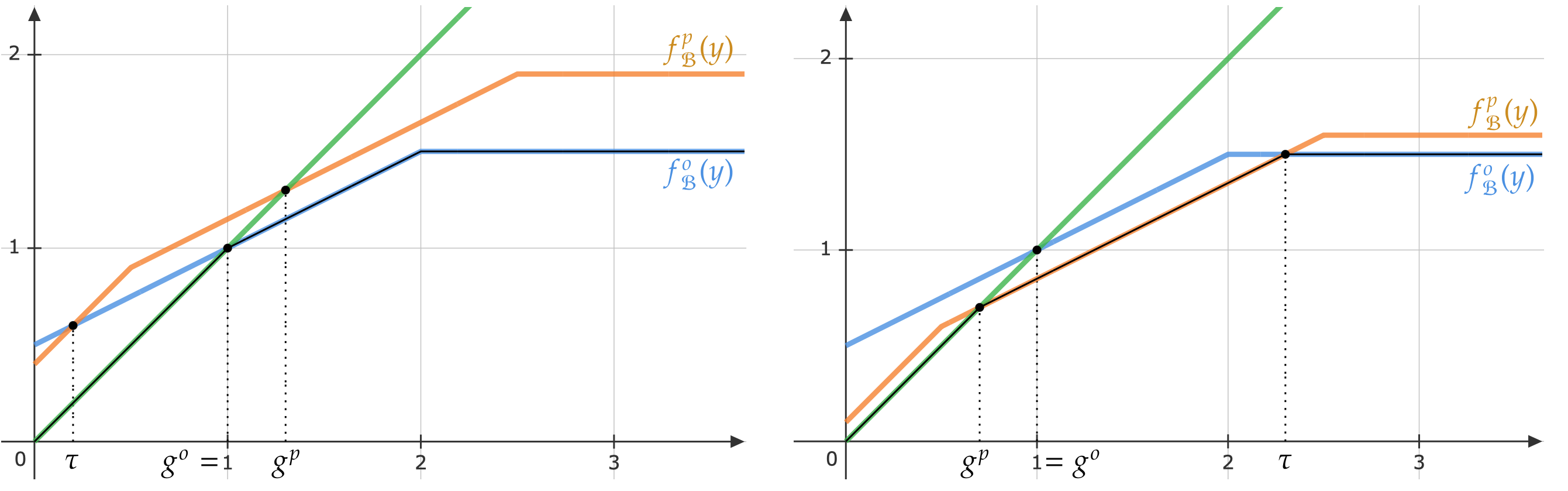} 
    
    \caption{The optimality curves for a box $\pbox$ with opening cost $c^o=0.5$, value $X=0$ with probability $0.5$ and $X=2$ with probability $0.5$ and peeking cost $c^p=0.4$ (left) and $c^p=0.1$ (right). The optimality curve of $\pbox$ is given by the minimum of the three curves. Observe that the curves of the opening and peeking action intersect at a unique point $\tau$ and that unless $g^p<g^o$, the opening action dominates the peeking action.}%
    \label{fig:pvo-curves}%
\end{figure}

\Cref{lemma:pvo-local-approx} is a consequence of the following claims, establishing local approximation guarantees for the opening and peeking actions respectively.

\begin{restatable}{claim}{pvoclone}\label{claim:pvo-open}
    For all $y\in\real$, it holds that $\min\{y, f^o_\pbox(y)\} \leq \alpha\cdot f_{\pbox}(\frac{y}{\alpha})$ for $\alpha = \frac{c^o}{c^p}\cdot \big(1-\frac{c^o}{g^p}\big)$.
\end{restatable}

\vspace{-8pt}
\begin{restatable}{claim}{pvocltwo}\label{claim:pvo-peek}
    For all $y\in\real$, it holds that $\min\{y, f^p_\pbox(y)\} \leq \alpha\cdot f_{\pbox}(\frac{y}{\alpha})$ for $\alpha = 1+ \min\left(\frac{c^p}{c^o},\frac{c^o}{g^p}\right)$.
\end{restatable}
\noindent 
The proofs of the two claims are deferred to~\Cref{app:pvo}. Note that once proven, they immediately imply an $\alpha$-local approximation with
\[\alpha = \min\left\{\frac{c^o}{c^p}\cdot \left(1-\frac{c^o}{g^p}\right) \;,\; 1+ \frac{c^p}{c^o} \; , \; 1 + \frac{c^o}{g^p}\right\}.\]
Fixing $c^o$ and $c^p$ and letting $\lambda := c^p/c^o \in (0,1)$, the minimum of the first and third terms is maximized at $g^p = c^o\cdot \frac{1+\lambda}{1-\lambda}$; so we have $\alpha \leq \min\{1+\lambda , \frac{2}{1+\lambda}\}\leq \sqrt{2}$. 

\newpage
\section{Additive Pandora's Box (Minimization)} \label{sec:additive_pb}

In this section, we study a different generalization of Pandora's Box which we call \textit{Additive Pandora's Box} (henceforth, APB). In this setting, the value of each (additive) box is given by the sum of independent random variables, all of which need to be separately probed by paying a corresponding probing cost. The objective is to minimize the total sum of opening costs plus the additive value of the accepted box. 

The primary motivation for studying APB is that it constitutes a special case of the notoriously challenging \textit{Pandora's Shortest Path} problem, where the edge weights of a given graph correspond to independent Pandora's boxes and the objective is to accept a set of boxes that form a path between two given vertices $s,t$. Due to the fact that path constraints do not admit frugal algorithms, no results are currently known for this problem and it is left as an open direction by \cite{S17}. It is not hard to see that APB corresponds to the special case of Pandora's Shortest Path where the graph consists of parallel paths between $s$ and $t$; each path is then modeled via an additive Pandora's Box (the components of which correspond to the path's edges) and the path constraint translates to a single-selection constraint over the additive boxes.

Formally, an instance of APB consists of $n$-additive boxes $\{\mathcal{B}_i\}_{i=1}^n$. Each box $\mathcal{B}_i = (\vec{\mathcal{D}}_i, \vec{c}_i)$ is characterized by $k_i\geq 1$ random variables, each distributed independently according to $\mathcal{D}_{ij}$ and with an opening cost of $c_{ij}\geq 0$ for $j\in [k_i]$. The algorithm can observe the realization $V_{ij}\sim \mathcal{D}_{ij}$ by paying $c_{ij}$ and can accept the box $\mathcal{B}_i$ only after all the $k_i$ random variables  have been observed, at a cost of $\sum_{j=1}^{k_i}V_{ij}$. At any point, the algorithm can adaptively select which additive box to explore, as well as which of its components to probe.

We can immediately cast this problem as a min-CICS instance where each additive box corresponds to an MDP $\mathcal{M}$ (which we call an additive-MDP), capturing all the different (adaptive) probing orders for the components of the box. A commitment for an additive-MDP corresponds to a protocol that determines which component to probe first, which second depending on the realization of the first value etc. Our main contribution is a to show that the commitment gap of APB is constant, and does not depend on the number of components.

\begin{theorem}\label{thm:sum-mdp-gap}
    The commitment gap of matroid-APB is at most $2$.
\end{theorem}

Our proof of~\Cref{thm:sum-mdp-gap} is existential and does not indicate how to efficiently construct the committing policy that achieves this bound. One could consider a restricted class of committing policies, which we call \textbf{static commitment policies}, where the components of each additive box are always probed in a fixed order. Trivial arguments suffice to show that static policies probing the components in increasing order of costs achieve a commitment gap of $k$. By considering the structure of water filling surrogate costs, we can obtain an improved bound on the gap:

\begin{restatable}{theorem}{additivestatic}
\label{thm:sum-mdp-static} Let $\mc$ be any additive-MDP. For $k=2$, there exists a static commitment policy that achieves a $\phi\approx 1.618$-pointwise approximation. For $k\geq 3$, the static policy of minimum index achieves a $\big(1+\lfloor \frac{k+1}{2}\rfloor\big)$-pointwise approximation.
\end{restatable}

In order to constructively achieve the approximation factor guaranteed by \Cref{thm:sum-mdp-static}, one can compute the indices of the $k!$ static policies for each constituent MDP and commit to the one with the smallest index. We leave open the questions of whether it is possible to achieve the same approximation more efficiently, and whether one could efficiently compute a committing policy that admits a constant factor  gap. The proof of~\Cref{thm:sum-mdp-static} is presented in~\Cref{app:additive_pb}.

\subsection*{Proof of~\Cref{thm:sum-mdp-gap}}
Our proof boils down to proving that each additive-MDP admits a $2$-pointwise approximation under a suitably chosen commitment and then employing~\Cref{lemma:pw-to-local} and~\Cref{thm:la-comp}. Our proof consists of two steps: first, we identify a simple structure for the surrogate costs under any commitment for APB and translate pointwise-approximation into a simple condition for the commitment. Then, we show that there exists some commitment satisfying this condition.

\paragraph{Step 1.} Let $\mathcal{M}$ be any additive-MDP with $k\geq 1$ components corresponding to values $V_i\sim\mathcal{D}_i$ and costs $c_i\geq 0$ for $i\in [k]$. Let $\pi \in C(\mathcal{M})$ be any commitment; observe that $\pi$ specifies which value to be probed first, which to be probed second depending on the first's realization and so on until all $k$-values have been probed. Therefore, $\mathcal{M}^\pi$ is a tree-structured Markov chain of height $k$, the terminal states $t(\vec{V})$ of which correspond to all the possible realizations of the value vector $\vec{V} =(V_1,\cdots , V_k)$. Furthermore, the probability of reaching such a terminal state $t(\vec{V})$ is always equal to the probability of $\vec{V}$ being realized, independently of $\pi$'s probing protocol. Therefore, the water filling surrogate cost $W^*_{\mc^\pi}$ for any commitment $\pi$ is obtained by sampling $\vec{V}\sim\vec{\mathcal{D}}$ and then returning a (commitment-specific) surrogate cost $\surr^{\pi}(t(\vec{V}))$. Due to this structure, the $\alpha$-pointwise approximation condition for additive-MDPs translates to proving that there exists some commitment $\pi\in C(\mc)$ such that
\[\surr^{\pi}(t(\vec{V})) \leq \alpha\cdot \surr^{\pi'}(t(\vec{V}))\]
for all vectors $\vec{V}\sim\vec{\mathcal{D}}$ and all {\em deterministic} commitments  $\pi'\in C(\mc)$.

Up next, we will use the inductive definition of water filling amortization to obtain structure on these surrogate costs. Fix a commitment $\pi\in C(\mc)$ and the underlying Markov chain $\mc^\pi$. For each state $s$ of $\mc^\pi$, we define the following quantities:
\begin{enumerate}
    \item The \textbf{index} $g^\pi(s)$ of state $s$ is defined as the water filling index of the sub-chain rooted at $s$, if all values that have been already probed have been realized to $0$. In other words, $g^\pi(s)$ corresponds to the index of the committing policy that proceeds as $\pi$ after state $s$ is reached, without the additive cost of already observed values. Clearly, $g^\pi(t(\vec{V})) = 0$ for all terminal states $t(\vec{V})$ and $g^\pi(r) = g^\pi$ for the root $r$ of $\mc^\pi$.

    \item The \textbf{value} $V^\pi(s)$ of state $s$ is defined as the sum of all the realized values in the trajectory from the root to $s$; clearly, we have that $V^\pi(r)=0$ and $V^\pi(t(\vec{V})) = \sum_{i=1}^k V_i$ for all terminal states $t_{\vec{V}}$.
\end{enumerate} 

The key observation is that when amortizing some state $s$, all the terminal states in it's subtree will share the same additive offset $V^\pi(s)$, corresponding to the sum of realized values leading to $s$. It is immediate by the definition of the water filling amortization that such an additive offset does not affect the cost shares of vertex $s$. Also, recall that water filling amortization is performed in a bottom-up manner. Thus, the surrogate cost of a terminal state $t(\vec{V})$ immediately after one of its ancestors $s$ gets amortized will be precisely $V^\pi(s) + \max(g^\pi(s), \surr'(t(\vec{V})))$, where $\surr'(t(\vec{V}))$ is the surrogate cost of $t(\vec{v})$ immediately after it's ancestor that is a child of $s$ was amortized. Performing this step for all ancestors of $t(\vec{V})$, we obtain that
\[\surr^\pi(t(\vec{V})) = \max_{i=1}^{k+1}\bigg( V^\pi(s_i) + g^\pi(s_{i})\bigg)\]
where $(r=s_1,s_2,\cdots, s_k, s_{k+1}=t(\vec{V}))$ is the unique path in $\mc^\pi$ that reaches $t(\vec{V})$.

From this expression, it becomes clear that $\surr^\pi(t(\vec{V})) \geq \max(g^\pi , \sum_{i=1}^k V_i)$. Furthermore, if it is the case that $g^\pi \geq g^\pi(s)$ for all states $s$, then we can easily upper bound $\surr^\pi(t(\vec{V}))\leq g^\pi + \sum_{i=1}^k V_i$. We refer to commitments that have this property as \textbf{root dominant}. Combining everything, we then obtain that for any root dominant commitment $\pi$, any commitment $\pi'$ and any terminal state $t(\vec{V})$, we have that
\[\frac{\surr^{\pi}(t(\vec{V}))}{\surr^{\pi'}(t(\vec{V}))} \leq \frac{g^\pi + \sum_{i=1}^k V_i}{\max(g^{\pi'} , \sum_{i=1}^k V_i)}\leq 1 + \frac{g^\pi}{g^{\pi'}}. \]
In other words, we have shown that any root dominant commitment $\pi$ will achieve an $\alpha$-pointwise approximation for $\alpha = 1 + g^\pi/\min_{\pi'\in C(\mc)}g^{\pi'}$. The proof of~\Cref{thm:sum-mdp-gap} is then completed by arguing that there exists a root dominant commitment $\pi$ that has minimum index across all committing policies.

\paragraph{Step 2.} Let $\pi$ be any commitment of minimum index. If $\pi$ is root dominant, then we are obviously done. Otherwise, there exists at least one state $s_1$ in $\mc^\pi$ such that $g^\pi(s_1) > g^\pi$. Among all these states, we consider one of minimum distance to the root. Notice that $s_1$ must have at least one sibling state $s_2$ with $g^\pi(s_2) \neq g^\pi(s_1)$; otherwise, the index of of the parent state of $s_1$ will be at least $g^\pi(s_1)$ (by the inductive definition of indices) and thus $s_1$ won't be a maximum height state satisfying the condition.

Let $\pi'$ be the commitment that is obtained by substituting the sub-tree rooted at $s_1$ in $\mc^\pi$ with the sub-tree rooted at $s_2$; notice that both of these subtrees correspond to commitments over the (same) set of unprobed values and thus such a substitution is always allowed (i.e. $\pi'\in C(\mc)$ as well). Up next, we will show that $g^\pi \geq g^{\pi'}$; then, $\pi'$ will also be a minimum index commitment and we can apply the same process to it. Furthermore, each substitution reduces the heteromorphity of the underlying chain (i.e. the number of different child sub-trees that a state can have) so our iterative process is guaranteed to terminate after some finite number of steps, resulting in a commitment that is both root dominant and has minimum index.

We finally argue that $g^\pi\geq g^{\pi'}$. Since $g^\pi$ is the water filling index of $\mc^\pi$, it is by definition at least as large as the minimum surrogate cost of \textit{any} valid amortization of $\mc^\pi$. Let $b=\{b_{st}\}_{s,t}$ denote the water filling cost shares of $\mc^\pi$ and $b'=\{b'_{st}\}_{s,t}$ denote the water filling cost shares of $\mc^{\pi'}$ Now, consider the following amortization $\hat{b}=\{\hat{b}_{st}\}_{s,t}$ of $\mc^\pi$:
\begin{itemize}
    \item For each state $s$ that does not belong in the sub-tree rooted at $s_1$, let $\hat{b}_{st} = b'_{st}$ for all terminals $t$. 
    \item  For each state $s$ in the sub-tree rooted at $s_1$ (including $s_1$), let $\hat{b}_{st} = b_{st}$ for all terminals $t$. 
\end{itemize}
Since all the states of $\mc^\pi$ and $\mc^{\pi'}$ that are not ancestors of state $s_1$ have the same distribution over terminal states, this is indeed a valid amortization. Furthermore, any terminal state $t$ that is not a descendant of $s_1$ will receive precisely the same cost shares as it does in $\mc^{\pi'}$ and thus it's surrogate cost will be at least $g^{\pi'}$. On the other hand, every terminal state $t$ that is a descendant of $s_1$ receives the same cost shares as it did in $\mc^\pi$ up until state $s_1$ gets amortized; by definition, this implies that it's surrogate cost will be at least $g^\pi(s_1) + V^\pi(s_1)\geq g^\pi(s_1)$. Thus, we conclude that $g^\pi \geq \min(g^{\pi'}, g^\pi(s_1)) = g^{\pi'}$ as desired.

\newpage
\bigskip
\section{The Weighing Scale Problem (Minimization)}\label{sec:ws}

We will now introduce the {\em Weighing Scale} (henceforth, WS) problem. A decision maker is presented with $n$ alternatives $(X_i,c_i)$ and a combinatorial constraint $\feas\subseteq 2^{[n]}$; $X_i\geq 0$ is the random value of the alternative realized independently by a known distribution and $c_i\geq 0$ is a weighing cost. The only way the decision maker can determine any further information about each value $X_i$ beyond its distribution, is to use a weighing scale to compare it against some fixed threshold $t$ of their choosing at the additional cost of $c_i$ and learn whether $X_i\leq t$ or not. The process terminates with the decision maker selecting a feasible set of alternatives $S\subseteq \feas$, paying their total value. Note that each alternative $(X_i,c_i)$ corresponds to a Costly Information MDP $\mc_i$ that captures the information acquisition process described above. In particular, the available actions at the starting state of each MDP $\mc_i$ are as follows:
\begin{enumerate}
    \item Pick a threshold $t\in\mathrm{support}(X_i)$ and weigh the alternative against it at a cost of $c_i$. Upon taking one of these actions, the MDP advances to one of two random subprocesses $\mc_i^{\leq t}$ and $\mc_i^{>t}$, defined over the random variables $(X_i|X_i\leq t)$ and $(X_i|X_i>t)$ respectively, based on the outcome of the weighing.

    \item Commit no more weighings of the alternative; this is a $0$-cost action resulting to a terminal state $x_i$ of value $v(x_i):=\expect{X_i}$.
\end{enumerate}
\medskip
\noindent
From this equivalence, an instance of WS corresponds to an instance of min-CICS over the corresponding MDPs.\footnote{Technically, our framework does not  capture infinite horizon MDPs. However, observe that whenever the decision maker has identified that the value of the alternative lies in some interval of length $\leq c$, performing any extra weighings is suboptimal. Thus, the corresponding MDPs for the WS problem are finite horizon without loss, even for continuous random variables $X$.} 
Our main result for this section is the following:

\begin{theorem}\label{thm:ws-approximation}
The commitment gap of matroid-WS is at most  $O(\max_i\kappa_i)$, with the parameter $\kappa_i$ for each alternative $i\in [n]$ defined as
\[\kappa_i := \frac{\mu_i}{M_i} + \log\frac{\mu_i}{g_i}\]
where $\mu_i=\expect{X_i}$ is the expected value of $X_i$, $M_i$ is the median value of $X_i$, and $g_i$ denotes the Gittins index of the alternative; i.e. the solution to the equation $c_i = \expect{(g_i-X_i)^+}$.
\end{theorem}

We prove~\Cref{thm:ws-approximation} by showing that for each alternative $i\in [n]$, the MDP $\mc_i$ admits an $O(\kappa_i)$-pointwise approximation. Furthermore, the commitment $\pi$ that achieves this condition corresponds to a very natural and efficient algorithm which we call \textit{One-Sided Halving}. Finally, we note that without any assumptions on the distributions of the values, the parameters $\kappa_i$ can be unbounded; in particular, for heavy tail distributions, $\mu_i/M_i$ can be arbitrarily large. We show in~\Cref{thm:ws-lb} (proven in~\Cref{app:ws}) that this dependence is necessary for any pointwise approximation guarantee. Better bounds may be possible via the weaker notion of local approximation or a global argument. 



\begin{restatable}{theorem}{wslbtag}\label{thm:ws-lb}
For any constant $\alpha \geq 1$, there exists a WS alternative that does not admit an $\alpha$-pointwise approximation.
\end{restatable}

\subsection*{Proof of~\Cref{thm:ws-approximation}}


We begin by introducing the committing policy for WS described below, which we call the \textit{One-Sided Halving algorithm}. The policy begins by weighing the alternative against some threshold $t_2$; if the alternative is larger, then the policy commits to no more weighings and if it is smaller, it halves its threshold from $t_2$ to $t_2/2$ and repeats, until the threshold reaches some fixed lower bound $t_1$. 

\begin{center}
\begin{algorithm}[H]
  \SetAlgoLined
  \SetKwInOut{Input}{Input}

    \Input{An MDP $\mc$ corresponding to a WS alternative $(X,c)$ and two thresholds $0<t_1<t_2$.}

    \smallskip 
    Set $t=t_2$.\\
    \smallskip
    \While{$t\geq t_1$}{
    \smallskip
    Weigh the random variable $X$ against $t$.\\
    \smallskip
    \textbf{if } $X\leq t$  \textbf{ then } $\mathrm{t=t/2}$ \textbf{ else } break.\\
    }
    Commit to no more weighings for the alternative.
  \caption{One-Sided Halving Algorithm.}
\end{algorithm}
\end{center}

Let $h$ be the solution to equation $c=\expect{(X-h)^+}$. The following lemma, combined with our reduction from pointwise to local approximation (\Cref{lemma:pw-to-local}) and our local approximation composition result (\Cref{thm:la-comp}) directly implies~\Cref{thm:ws-approximation}.

\begin{restatable}{lemma}{wsupptag}\label{lem:osh-policy}
    Committing to no weighings if $g>\min\{\mu, M\}$, and otherwise
    committing to the One-Sided Halving algorithm with $t_1=g$ and $t_2 = \min(M,h)$  achieves an $O(\kappa)$-pointwise approximation.
\end{restatable}

The proof of~\Cref{lem:osh-policy} is deferred to~\Cref{app:ws}. Here, we provide some intuition on how this result is obtained. A commitment $\pi\in\comset{\mc}$ for WS will declare in advance a decision tree over pre-specified thresholds against which it will weigh $X$, resulting in a Markov chain $\mc^\pi$. Importantly, $\pi$ will end up partitioning the support $\X$ of distribution $\dist$ into a set of intervals, corresponding to its terminal states. Furthermore, the probability of running the Markov chain $\mc^\pi$ and ending up in a terminal state $t$ corresponding to some interval $I$ will be precisely $\prob{X\in I}$. A pictorial representation of such Markov chains is given in~\Cref{fig:ws_plots}. By directly relating the water filling amortized cost on these intervals for the One-Sided Halving algorithm to the surrogate cost of the underlying MDP, we are able to obtain our pointwise approximation guarantees.

\begin{figure}[h!]
    \centering
    \includegraphics[width=0.7\textwidth]{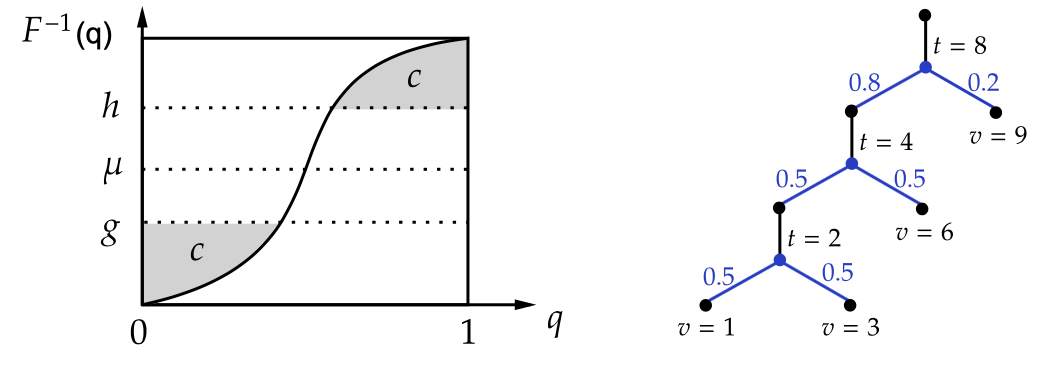}
    \caption{The indices $g$ and $h$ of the alternative can be obtained by the inverse CDF of the random value; the highlighted ares equal the weighing cost $c$. 
    The figure on the right corresponds to the Markov chain for One-Sided Halving instantiated with threshold $t_1=2$, $t_2=8$ for an instance with value $X\sim\mathrm{Unif}[0,10]$. The terminal states partition the support of $X$; the probability and value at each terminal equals the probability and the conditional expectation of the corresponding interval.}
    \label{fig:ws_plots}
\end{figure}

\newpage
\section{Pandora's Box with Optional Inspection (Maximization)}
\label{sec:gvo}
Finally, we consider the maximization version of \emph{Pandora's Box with Optional Inspection} (henceforth, PBOI), arguably the most studied Pandora's Box variant. In PBOI, the decision-maker may select a box without opening it; we refer to this action as ``grabbing'' the box. There is a rich line of research addressing PBOI; we provide a detailed discussion of these results in~\Cref{sec:related}. In the single-selection setting, PBOI is NP-hard but admits a PTAS \citep{FLL22,BC22}. For matroid constraints, \cite{BK19} show that the commitment gap is at least $0.63$; however, their proof is nonconstructive, so the only known way to achieve this ratio is to enumerate and evaluate the exponentially many committing policies.

Currently, there is no known policy that achieves better than a $0.5$-approximation for matroid-PBOI. Additionally, this $0.5$-approximation is somewhat uninteresting. As \citet{BK19} observe, the optimal policy that does not grab any box will attain half of the utility from boxes that the optimal doesn't grab, and the optimal policy that only grabs boxes will attain half of the utility from boxes that the optimal grabs. Thus, uniformly randomizing between these policies will immediately give us a $0.5$-approximation.\footnote{%
\Citet{BK19} make this observation in the single-item setting, but combining it with results of \citet{S17} generalizes it to the matroid setting.} In this section, we present the first improvement over the trivial approximation factor of $0.5$ for matroid-max-PBOI:

\begin{theorem}\label{thm:gvo-approximation}
    There exists an efficient randomized committing policy that achieves a $0.582$-approximation to the optimal adaptive policy for any instance of matroid-max-PBOI.
\end{theorem}

In~\Cref{sec:pboi_prelims} we cast PBOI into our CICS framework by considering each optional-inspection box as an MDP. A reasonable first approach would be to establish a local approximation condition for these MDPs in order to prove~\Cref{thm:gvo-approximation}; however, we show that for any $\epsilon>0$, there exist boxes for which we cannot do better than a $(0.5 + \epsilon)$-local approximation (\Cref{sec:pboi_no_local}). We use the intuition from these hard instances to develop a refinement of local approximation, which we call \emph{semilocal approximation}; we then show that, similarly to local approximation, semilocal approximation guarantees can be composed into a (global) approximation ratio (\Cref{sec:pboi_semilocal}). Finally, in \Cref{sec:pboi_approx} we show that this new notion of approximation suffices to establish the bound of~\Cref{thm:gvo-approximation}. Omitted proofs are presented in~\Cref{app:gvo}.

\paragraph{Limitations of Semilocal Approximation.} There are two prices we pay in refining local approximation to semilocal approximation. First, the definition as we state it currently is specific to PBOI, though it could perhaps be generalized to other MDP families in the future. Second, our composition result for semilocal approximations only holds for matroid feasibility constraints, whereas past works showed results compatible with any frugal (roughly, ``greedy'') algorithm \citep{S17, GJSS19, DS24}. We are not sure whether this second limitation is fundamental or might be overcome using a different proof technique.

\subsection{PBOI Preliminaries}\label{sec:pboi_prelims}
We begin by establishing the necessary notation for PBOI. All the definitions mentioned in this section were previously established for the minimization setting by \cite{DS24}; we simply restate them here for the maximization setting.

Fix any optional-inspection PB $\pbox=(\dist, c)$. We use $X\sim\dist$ to denote the value of the box and $\mu=\expect{X}$ to denote its expectation. We can model $\pbox$ as a Costly Information MDP $\mc$ with two actions: opening the box at a cost of $c$ and transitioning to a terminal state of reward $X\sim\dist$, and grabbing the box at a cost of $0$ and transitioning to a terminal state of reward $\mu$ (\Cref{fig:optional}). Therefore, the set of (deterministic) local commitments for $\mc$ consists of just two options $\{o,g\}$ and the optimality curve of $\mc$ can be written as $\optf_\mc (y) := \max\{\optf_{\mc^o}(y), \optf_{\mc^g}(y)\}$, where
\[\optf_{\mc^o}(y) := \max\{y, \expect{\max(X,y)} -c\} \quad \text{,} \quad \optf_{\mc^g}(y) := \max\{y, \mu\}\]
are (by definition) the optimality curves under the two commitments. Using the inherent connection between optimality curves and surrogate costs, we can use these expressions to recover the following \citep{D18}:

\begin{definition}[\textbf{Surrogate costs for PBOI}] \label{def:pboi_sur_cost} 
Let $\pbox = (\dist, c)$ be any optional-inspection box. We define its Gittins index $g$ as the solution to equation $c = \expectt{X\sim\dist}{(X-g)^+}$ and its backup index as $h=\max(\mu,h')$, where $h'$ is the solution to equation $c = \expectt{X\sim\dist}{(h'-X)^+}$. Then:
\begin{enumerate}
    \item $W^*_o := W^*_{\mc^o} = \min\{X,g\}$ satisfies $\optf_{\mc}^o(y) = \expect{\max\{W^*_o, y\}}$.

    \item $W^*_g := W^*_{\mc^g} = \mu$ satisfies $\optf_{\mc}^g(y) = \expect{\max\{W^*_g, y\}}$.

    \item $W^* := W^*_{\mc} = \max\{W^*_o, h\}$ satisfies $\optf_{\mc}(y) = \expect{\max\{W^*, y\}}$.
    
\end{enumerate}
\end{definition}

These indices have an intuitive interpretation: consider the local game $(\mc,y)$ for some value of the outside option $y$. For small values of $y$, the optimal policy will immediately grab the box, and for large values of $y$ it will accept the outside option. The Gittins index $g$ corresponds to the maximum value of the outside option for which opening is preferable to stopping, and $h'$ corresponds to the minimum value of the outside option for which opening is preferable to grabbing. In other words, the optimal policy for the local game $(\mc,y)$ will grab the box if $y\in [0,h']$, it will open the box if $y\in [h', g]$ and it will accept the outside option if $y\geq g$. Note that if $g<h'$, opening the box is never the optimal action. We refer to such boxes $\pbox = (\dist, c)$ as degenerate, and since they are never opened by any optimal policy, we can substitute them by normalized boxes $\pbox = (\mu, 0)$ without loss of any generality. From now on, we will be assuming that all degenerate boxes have been normalized; this in turn implies that $h=h'\leq \mu$ and $c\leq \mu\leq g$ (\Cref{fig:optional}).

\begin{figure}[h!]
    \centering
    \includegraphics[width=0.9\textwidth]{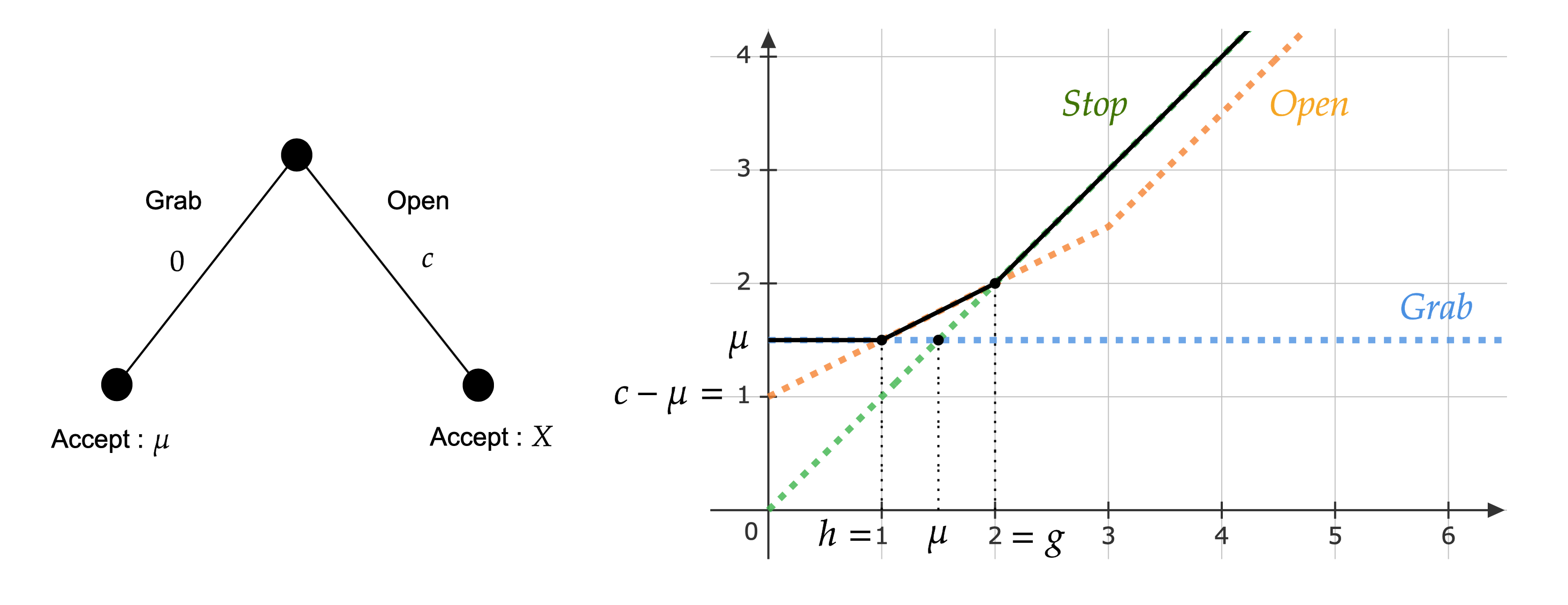} 
    
    \caption{An optional inspection box $\pbox = (\dist, c)$ can me modeled as an MDP $\mc$ with two actions, grab and open, incuring costs $0$ and $c$ and resulting to rewards $\expect{X}$ and $X$ respectively. On the right side, we demonstrate the optimality curves for the local game $(\mc,y)$ for a box with cost $c=0.5$ and reward $X=3\cdot \mathrm{Be}(1/2)$.}%
    \label{fig:optional}%
\end{figure}

\subsection{Insufficiency of Local Approximation}\label{sec:pboi_no_local}
Since any MDP $\mc$ describing an optional-inspection box only has two actions, any (randomized) commitment $\pi\in\comset{\mc}$ is fully described by the probability $p\in [0,1]$ of committing to the grab action. Therefore, by \cref{def:pboi_sur_cost}, the $\alpha$-local approximation condition for max-PBOI translates to
\[\exists p\in [0,1]: \;\; (1-p)\cdot \expect{\max\{W^*_o, y\}}+ p\cdot\expect{\max\{W^*_g, y\}} \geq \expect{\max\{\alpha W^*, y\}}.\]

We will now demonstrate a specific box where no commitment $p\in [0,1]$ can achieve a local approximation factor that is better than $1/2$.

\begin{example}\label{ex:pboi}
    For any $n \geq 1$, consider the optional-inspection box $\pbox^n = (\dist, c)$ where $c=n-1$ and
        \[
            \dist = \begin{cases}
                1 & \text{w.p. } 1-\frac{1}{n^2} \\
                n^3 & \text{w.p. } \frac{1}{n^2} \\
            \end{cases}.
        \]

\noindent Straightforward computations show that the parameters of this box are $\mu = n+1-n^{-2}$, $g=n^2$ and $h=1 + n^2/(n+1)$. Using these expressions and the definition of surrogate costs (\Cref{def:pboi_sur_cost}), we can proceed to write-down the local approximation condition at points $y=0$ and $y=\mu$ as follows:
\[\text{Local approximation at $y=0$: }\;\;\alpha \leq 1 - (1-p)\cdot \frac{n^3-n^2}{n^3+n^2-1}\]
\[\text{Local approximation at $y=\mu$: }\;\;\alpha \leq 1 - p\cdot \frac{n^2-n-1+n^{-2}}{n^2}\]
Just from these two values of $y$, we obtain that as $n\rightarrow \infty$, we have that $\alpha\leq 1 - \max(p,1-p)\leq 0.5$.
\end{example}

To better understand this ``worst-case'' behavior that \Cref{ex:pboi} demonstrates, it will be illustrative to visualize the optimality curves of the box $\pbox^n$ for some large value of $n$. \Cref{fig:vanilla-local-hedging-bad-case-box} shows the curves for $n = 20$. We see that when $y$ is small, we have a strong preference for grabbing the box, and when $y$ is larger, we have a very slight preference for opening the box. Intuitively, we might guess that committing to the grabbing action performs quite well, since in the worst case there is only a small \emph{additive} sub-optimality factor. However, local approximation requires a small \emph{multiplicative} sub-optimality factor, which, as \cref{ex:pboi} above demonstrates, we cannot achieve.

\begin{figure}[h!]
    \centering
    \includegraphics[width=0.4\textwidth]{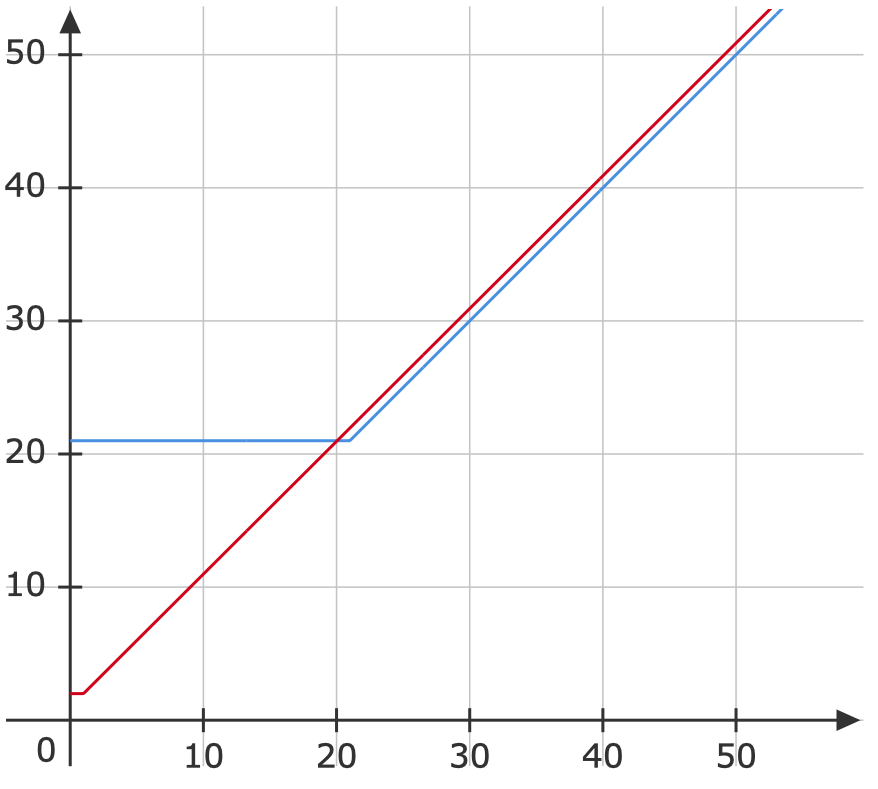}
    \caption{The red and blue lines represent the optimality curves for grabbing $\pbox^{20}$ and opening $\pbox^{20}$, respectively. This best possible local approximation achieved for this box is $0.537$.}%
    \label{fig:vanilla-local-hedging-bad-case-box}%
\end{figure}

\subsection{Semilocal Approximation}\label{sec:pboi_semilocal}

The above discussion suggests that we need to reason about additive suboptimality, not just multiplicative. We do so by introducing $(\alpha, \beta)$-\emph{semilocal approximation}, defined below. Roughly speaking, the $\alpha$ captures multiplicative suboptimality, while $\beta$ captures the additive suboptimality. Formally:

\begin{definition}[\textbf{Semilocal Approximation}]
    \label{def:semilocal_approx}
    We say that an optional-inspection box $\pbox = (\dist, c)$ admits an \emph{$(\alpha, \beta)$-semilocal approximation} if there exists a probability $p\in [0,1]$ such that for all $y\in\mathbb{R}$:
    \begin{equation*}
        (1 - p) \cdot \E{\max\Bgp{W^*_o, y}} + p \cdot \E{\max\Bgp{W^*_g, y}}
        \geq \E{\max\Bgp{\alpha W^*, y}} - p \beta \mu.
    \end{equation*}
\end{definition}

Re-writing the above expression as
\[
(1 - p) \cdot \E{\max\Bgp{W^*_o, y}} + p \cdot (\E{\max\Bgp{W^*_g, y}} + \beta\mu)
        \geq \E{\max\Bgp{\alpha W^*, y}} ,
\]
it becomes clear that $(\alpha,\beta)$-semilocal approximation is basically an $\alpha$-local approximation where we have boosted the value attained from grabbing by an additive factor proportional to the mean. It is important that we do this only for grabbing—boosting the value of inspecting by an additive factor would lead to poor guarantees when composing semilocal approximations in the matroid selection setting.

Equipped with the semilocal approximation definition, we now show that semilocal approximations can be composed under matroid feasibility constraints:

\begin{restatable}{theorem}{semilocalcompo}\label{thm:semilocal_composition}
    Let $\alpha>\beta\geq 0$, and let $\instance = (\pbox_1, \dots, \pbox_n, \feas)$ be any instance of matroid-max-PBOI where each constituent box $\pbox_i$ admits an $(\alpha, \beta)$-semilocal approximation under some probability $p_i\in [0,1]$. Then, $\cg(\instance)\geq \alpha - \beta$.
\end{restatable}

\begin{algorithm}
    \SetAlgoLined
    \SetKwInOut{Input}{Input}
    \SetKwFor{For}{For }{\unskip:}{}{}
    \SetKwIF{If}{ElseIf}{Else}{If}{\unskip:}{Else if}{Else}{}
    \SetKwComment{Comment}{}{}
    \DontPrintSemicolon
    \caption{Semilocal Approximation Composition Algorithm}
    \label{alg:semilocal_composition}
    \Input{%
        A normalized matroid-max-PBOI instance $\instance = (\pbox_1, \dots, \pbox_n, \feas)$. \\
        A vector of probabilities $(p_1, \dots, p_n)$.
    }
    \smallskip
    Relabel the boxes such that $\mu_1 \geq \dots \geq \mu_n$
    \;
    \smallskip
    $S^\grab \gets \{\}$
    \Comment*{\color{gray}\normalfont%
        set of boxes marked as ``grab''}
    \smallskip
    \For{\upshape$i \gets 1, \dots, n$}{%
        \smallskip
        
        $L_i \gets \textbf{if } S^\grab \cup \{i\} \in \feas \textbf{ then } 1 \textbf{ else } 0$
        \Comment*{\color{gray}\normalfont%
            $L_i = 0$ means we never want to grab box~$i$}
        \smallskip
        Sample $K_i \gets \Bernoulli(p_i)$
        \Comment*{\color{gray}\normalfont%
            $K_i = 1$ means \emph{provisionally} mark box~$i$ ``grab''}
        \smallskip
        \If{\upshape$K_i = 1 \textbf{ and } L_i = 1$}{
            \smallskip
            $S^\grab \gets S^\grab \cup \{i\}$
            \Comment*{\color{gray}\normalfont%
                fully mark box~$i$ as ``grab''}
        }
    }
    \smallskip

    Commit to grabbing the boxes $S^\grab$ and opening the boxes in $[n]\setminus S^\grab$.
    \;
    \smallskip
    Run the optimal (index) policy under the resulting commitment.
    \;
\end{algorithm}

The committing policy achieving the bound of~\Cref{thm:semilocal_composition} is presented in~\Cref{alg:semilocal_composition}. Importantly, this policy does not commit to grabbing each box $\pbox_i$ with probability $p_i$ independently (as would be the case when composing local approximation guarantees). Instead, it examines the boxes in decreasing mean order, flips a coin with probability $p_i$ for each box $\pbox_i$ and only commits to grabbing it if (i) the coin flip succeeds and (ii) the set of boxes that are committed to the grab action after inserting $\pbox_i$ remains an independent set of the underlying matroid. Once the commitments for each box have been specified, we simply run the optimal policy for the resulting matroid-max-CICS instance over Markov chains (equivalently, we substitute each grab box $(\dist, c)$ with a box $(\mu, 0)$ and run Weitzman's algorithm for the mandatory inspection setting). Therefore, as long as the probabilities $p_i$ achieving the semilocal approximation can be efficiently computed, the entire algorithm runs in polynomial time. The proof of~\Cref{thm:semilocal_composition} is deferred to~\Cref{app:gvo}.

We note that both the definition and composition algorithm for semilocal approximation are specifically tailored to the PBOI model, though we believe modest extensions beyond PBOI may be possible. The main limitaiton is that the proof of \cref{thm:semilocal_composition} relies crucially on the fact that if we choose to grab a box, then the box's terminal value is deterministic, namely~$\mu$. It is thus natural to try extending semilocal approximation to other Pandora's box variants that have a grab-like action that yields a deterministic terminal value, but we leave this to future work.

\subsection{Breaking the 0.5 Barrier}\label{sec:pboi_approx}
We are finally ready to prove~\Cref{thm:gvo-approximation}. From~\Cref{thm:semilocal_composition}, it suffices to prove that all optional inspection boxes admit an $(\alpha,\beta)$-semilocal approximation for some $\alpha,\beta$ such that $\alpha-\beta \geq 0.582$. This is achieved by the following lemma, proven in~\Cref{app:gvo}.

\begin{restatable}{lemma}{semilocalconstant}\label{lem:semilocal-constant}
    Let $\pbox$ be any (non-degenerate) box with opening cost $c$ and mean value $\mu$. For any $\beta\geq 0$, let
    \[
        \alpha(\beta) := \begin{cases}
            \frac{1}{1+\frac{c}{\mu}-\beta\frac{c}{\mu-c}} &\;\;, \text{if } \frac{1}{1+\frac{c}{\mu}-\beta\frac{c}{\mu-c}} \in (0,1]. \\
            1 &\;\; ,\textnormal{otherwise.}
        \end{cases}
    \]
Then, the probability $p=\frac{c}{\mu}\alpha(\beta)$ achieves an $(\alpha(\beta),\beta)$-semilocal approximation for $\pbox$.
\end{restatable}

By optimizing over $\alpha(\beta)-\beta$, we instantiate \Cref{lem:semilocal-constant} with $\beta = 0.1$, to obtain the corresponding probability. It is straightforward to verify (using a computer algebra system or numerical solver) that since $\frac{c}{\mu}\in[0,1]$ (recall that we have normalized all degenerate boxes), we have \[\gp*{1+\frac{c}{\mu}-\frac{c}{10(\mu-c)}}^{-1} \geq 0.682\] whenever $\gp[\big]{1+\frac{c}{\mu}-\frac{c}{10(\mu-c)}}^{-1} \in (0,1]$. This means $\alpha(0.1)\geq 0.682$, and thus all boxes admit a $(0.682, 0.1)$-semilocal approximation, from which the proof of \Cref{thm:gvo-approximation} follows.

Critical to the above argument is that \cref{lem:semilocal-constant} implies there exists a \emph{universal} pair $(\alpha, \beta)$, namely $(0.682, 0.1)$, such that all boxes admit an $(\alpha, \beta)$-semilocal approximation. It would \emph{not} suffice to show only that for each box, there exists a \emph{box-specific} pair $(\alpha', \beta')$ such that the box admits a $(\alpha', \beta')$-semilocal approximation, even if we always had $\alpha' - \beta' \geq 0.582$. This is because \cref{thm:semilocal_composition} requires all boxes to admit a $(\alpha, \beta)$-approximation for the same pair $(\alpha, \beta)$.


\newpage
\appendix
\section*{Appendix}
\section{The Maximization Setting}\label{app:maxim}

In this section, we describe how our entire framework extends to the maximization setting under matroid feasibility constraints. Our objective will be to restate our claims from Sections~\ref{sec:amortization} and~\ref{sec:local_approx}, as they were only established with respect to the matroid-min-CICS problem. Since the proofs follow exactly the same steps, rather than re-proving all our results, we simply discuss the differences, where there are any.

\subsection{Amortization for Markov Chains}
The cost amortization of a Markov chain is defined in the same manner, with the difference that instead of increasing the terminal values of the trajectories to obtain surrogate costs, we now \textit{decrease} them to obtain \textit{surrogate values}. Furthermore, the index of a state now corresponds to the \textit{maximum} surrogate value among all downwards trajectories instead of the minimum. Formally:

\begin{definition}[\bf Cost Amortization for Maximization]
A cost amortization for any Markov chain $\mc=(S, \sigma, A, c, \dist, V, T)$ is a non-negative vector $b=\{b_{s\tau}\}_{s\in S, \tau\in \traj(s)}$ satisfying $\sum_{\tau\in \traj(s)} p(\tau)b_{s\tau}= p(s)c(s)$ for all states $s\in S$. Based on this amortization, we define:
\begin{itemize}
    \item The amortized value of a trajectory $\tau\in \traj$ as $\surr_b(\tau) := v(\tau) - \sum_{s\in\tau} b_{s\tau_s}$. 
    
    \item The surrogate value of the Markov chain $\mc$ as the random variable $\surr_{\mc,b}$ that takes on value $\surr_b(\tau)$ for $\tau\in \traj$ with probability $p(\tau)$.
    
    \item The index of a state $s\in S$ of the Markov chain $\mc$ as $\ind_{\mc, b}(s):=\max_{\tau\in \traj: s\in\tau}\surr_b(\tau)$.
\end{itemize}
\end{definition}

Intuitively, we postpone the payment of the action costs until a terminal state is accepted, in which case a smaller (compared to the original value) surrogate value is collected. From this, we follow precisely the same steps as in the proof of~\Cref{lem:MC-lb} to upper bound the the utility of any algorithm for matroid-max-CICS through the surrogate values.

\begin{lemma}[Markov chain upper upper for maximization settings]
\label{lem:MC-lb-maxim}
    Consider any instance $\instance = (\vecM ,\feas)$ of matroid-max-CICS over Markov chains and let $b_i$ be any cost amortization of $\mc_i$ with surrogate value $\surr_i := \surr_{\mc_i, b_i}$ for all $i\in [n]$. Then, $\opt(\instance)\leq\expect{\max_{S\in\feas} \sum_{i\in S}\surr_i}$.
\end{lemma}

Up next, we extend our definition of water filling amortization to the maximization setting. Like before, the {\bf water draining cost amortization} is described algorithmically in a bottom-up fashion where states are considered in the chain in reverse topological order, starting from the terminals up towards the root $\sigma$. Trajectories $\tau\in\traj(t)$ for terminal states $t\in T$ are singletons and do not carry cost shares. Consider a state $s$ such that the cost shares for all states reachable from it have been determined. The state distributes its total cost $c(s)$ across its downstream trajectories $\tau\in\traj(s)$, starting from those with the \textbf{maximum current value}, until the equation $\sum_{\tau\in \traj(s)} p(\tau)b_{s\tau}= p(s)c(s)$ is satisfied. In other words, instead of increasing the cost of the minimum-cost trajectory, we now decrease the value of the maximum-value trajectory. We use $W^*_\mc$ to denote the water draining surrogate value of a Markov chain $\mc$, and $\ind^*_\mc(s)$ to denote the water draining index of a state $s$ in $\mc$.

Through this amortization, we once again define the corresponding index based policy; naturally, the policy will now select to advance the feasible Markov chain of \textit{maximum} index. 

\begin{definition}[Water Draining Index policy]
The water draining index policy for an instance $\instance = (\vecM ,\feas)$ of matroid-max-CICS over Markov chains selects at every step the Markov chain $i^*=\argmax_{i\in\feas_S} \ind^*_i(s_i)$, where $s_i$ is the current state of Markov chain $\mc_i$; $S$ is the set of terminated (selected) Markov chains; and $\feas_S=\{i: S\cup\{i\} \in \feas\}$.
\end{definition}

From the same observations as in the proof of~\Cref{lem:MC-lb}, it immediately follows that this algorithm achieves two desired properties: it always selects the feasible set of Markov chains with maximum total surrogate value, and whenever it advances a state $s\in S$ that contributes to the surrogate value of a downwards trajectory $\tau\in \traj(s)$ (i.e. $b_{s\tau}>0$) and nature realizes $\tau$, the algorithm will select it. From this, and the fact that the algorithm that greedily adds the maximum weight feasible element is optimal for the matroid independent set problem, the following counterpart to~\Cref{thm:MC-opt} follows, establishing the optimality of the Water Draining Index policy for matroid-max-CICS.

\begin{theorem}
\label{thm:MC-opt-maxim}
   For any matroid-max-CICS instance $\instance = (\vecM ,\feas)$ over Markov chains, the expected utility of the water draining index policy is equal 
    to $\expect{\max_{S\in\feas} \sum_{i\in S}W^*_{\mc_i}}$. The policy is therefore optimal for $\instance$.
\end{theorem}

\medskip
\subsection{Optimality Curves} 
The definition of a local game $(\mc,y)$ naturally extends to the maximization setting; at any step, the decision maker can either advance the Markov chain $\mc$ at a cost (until it reaches a terminal state in which case it may collect its value and terminate), or collect the reward $y$ of the outside option and terminate. Like before, we use $f_\mc(y)$ to denote the utility of the optimal adaptive policy in this game. From~\Cref{thm:MC-opt-maxim}, we immediately have that
\[f_\mc(y) = \expect{\max\{ y, W^*_{\mc}\}}.\]
From this expression, we obtain that the CDF of the water draining surrogate value $W^*_\mc$ can be derived from the
optimality curve as $\dd{y}{f_\mc(y)}$. This in turn allows us to define water draining surrogate values for arbitrary MDPs.

\begin{definition}
[Water draining surrogate values for MDPs]
 Let $\mc$ be an MDP with optimality curve $\optf_\mc$. The water draining surrogate value for $\mc$ is the random variable $W^*_\mc$ generated from the CDF $\dd{y}{\optf_\mc(y)}$. That is, $W^*_\mc$ is the random variable satisfying $\optf_\mc(y) = \expect{\max(y,W^*_\mc)}$ for all $y\in\real$.
\end{definition}

\medskip
\subsection{Amortization for MDPs} 
We will now use the definition of the water draining surrogate values to prove the following counterpart of~\Cref{lem:mdp-wf-char-new} that characterizes the water draining surrogate value of an MDP. Since this is the most technical proof of the framework and there are a few arguments that change in the maximization setting, we provide a proof sketch for the result.

\begin{lemma}
\label{lem:wf-char-maxim}
    For any MDP $\mc=(S, \sigma, A, c, \dist, V, T)$ and any deterministic commitment $\pi\in\comset{\mc}$, generating a Markov chain $\mc^\pi$ with states $S_\pi\subseteq S$, realizable trajectories $\traj_\pi\subseteq\traj$ and a distribution $p_\pi$ over trajectories and reachable states, there exists an amortized cost function $\surr^\pi:\traj_\pi \mapsto \Delta(\real)$ mapping trajectories to distributions over costs and a non-negative cost sharing vector $b^\pi = \{b^\pi_{s\tau}\}_{s\in S_\pi, \tau\in \traj_\pi(s)}$, such that the following properties hold:
    \begin{enumerate}
    \smallskip
    \item \textbf{\emph{Cost Sharing:}} the amortized cost of a trajectory pays for its own acceptance value and the cost shares it sends to upstream states. For all $\tau\in \traj_\pi$, it holds that $\expect{\surr^\pi(\tau)} = v(\tau) - \sum_{s\in \tau} b^\pi_{s\tau_s}$.
    \smallskip

    \item \textbf{\emph{Cost Dominance:}} the cost shares received by any state pay towards its action cost, but do not overpay. For all $s\in S_\pi$, it holds that $\sum_{\tau\in \traj_\pi(s)} p_\pi(\tau) b^\pi_{s\tau}\le p_\pi(s)c(\pi(s))$.
    \smallskip

    \item \textbf{\emph{Action Independence:}} sampling a trajectory $\tau\sim p_\pi$ and then sampling from the amortized cost distribution $\surr^\pi(\tau)$ generates a random surrogate cost for the MDP. This random variable is distributed identically to the water draining surrogate cost $W^*_\mc$. 
\end{enumerate}
\end{lemma}

\begin{proof}
    Notice that the only change with respect to the statement of~\Cref{lem:mdp-wf-char-new} for the minimization setting is the negative sign in the cost sharing property. Our proof will mirror the proof of~\Cref{lem:mdp-wf-char-new}, and we use the same notation throughout. Once again, the proof proceeds by induction on the horizon of the MDP; the $H=1$ case remains trivial.

    For the amortization of the action costs $c(a_j)$, we will now define $g_j$ to be the solution to equation
    \[\expect{Z_j} - c(a_j) = \expect{\min\{g_j,Z_j\}}\]
    and by defining $\hat{Z}_j:= \min\{g_j, Z_j\}$, we have
    \[\max\{y , \expect{\max\{y, Z_j\}} - c(a_j) \} = \expect{\max\{y , \hat{Z}_j\}} \]
    which in turn allows us to argue that
    \[\expect{\max\{y, W^*_\mc\}} \geq \expect{\max\{y,\hat{Z}_j\}}\]
    for all $y\in \real$ and $j\in [k]$.
    
    Using the identity $\max(a,b) = - \min(-a,-b)$, this implies that for all $j\in [k]$, the random variable $(-\hat{Z}_j)$ second-order stochastically dominates the random variable $(-W^*_\mc)$. From~\Cref{lemma:sdom}, this allows us to obtain mappings $m_j:\mathrm{supp}(\hat{Z}_j)\mapsto \Delta(\mathrm{supp}(W^*_\mc))$ such that:
    \begin{enumerate}
        \smallskip
        \item For each $j\in [k]$, $W^*_\mc$ can be obtain by sampling a $z\sim \hat{Z}_j$ and then sampling from $m_j(z)$.
        \smallskip
        \item For each $z\in\mathrm{supp}(\hat{Z}_j)$, we have $\expect{m_j(z)}\geq z$.
    \end{enumerate}
    \smallskip\noindent
    
    We can now define the amortized cost functions $\surr^\pi$ and the cost sharing vectors $b^\pi$. Fix any deterministic commitment $\pi\in\comset{\mc}$ and let $j = \pi(\sigma)\in [k]$ be the fixed action that $\pi$ takes at state $\sigma$. For each state $s\in R_j$, we use $\pi|s\in\comset{\mc_s}$ to denote the commitment $\pi$ after we transition to state $s$; note that this is also a deterministic commitment. By definition, each MDP $\mc_s$ has horizon up to $H$ and thus by the induction hypothesis, it admits a mapping $\surr^{\pi|\sigma}(\cdot)$ and a non-negative cost sharing vector $b^{\pi|\sigma}$, satisfying the properties of the lemma. Now, consider any trajectory $\tau\in\traj_\pi$ and observe that $\tau = \{\sigma, \tau_s\}$ for some $s\in R_j$ and some $\tau_s\in\traj_{\pi}(s)$. We define
    \[\surr^\pi(\tau) := m_j(\min\{g_j, \surr^{\pi|s}(\tau_s)\})\] 
    and
    \[b^\pi_{\sigma \tau} := \expect{\surr(\tau)}-\expect{\surr^{\pi|s}(\tau_s)}\]
    and for all other $s\in S_\pi\setminus\{\sigma\}$ and $\tau\in\traj_\pi(s)$, we use the same cost share $b^\pi_{s\tau} = b^{\pi|s}_{s\tau}$ that was used in $\mc_s$. The proof of the three properties follows precisely the same steps as~\Cref{lem:mdp-wf-char-new}.
\end{proof}

From~\Cref{lem:wf-char-maxim}, we immediately obtain the following counterpart of~\Cref{thm:MDP-lb}; the proof follows exactly the same steps as the proof of~\Cref{thm:MDP-lb} that we presented in~\Cref{app:amortization}, and is thus omitted.
\begin{theorem}
\label{thm:MDP-lb-max}
 In any instance $\instance = (\vecM ,\feas)$  of matroid-max-CICS, $\opt(\instance)\leq \expect{\max_{S\in\feas}\sum_{i\in S} W^*_{\mc_i}}$.
\end{theorem}

\medskip
\subsection{Local Approximation}
Finally, we note that our notion of local approximation seamlessly extends to the maximization setting by simply changing the inequality order. In particular:
\begin{definition}[Local approximation for maximization settings]\label{def:local-approx-maxim}
Let $\mc$ be any MDP and let $\alpha\in (0,1]$. We say that a commitment $\pi\in\comset{\mc}$ is an $\alpha$-local approximation for $\mc$ if 
    $\optf_{\mc^\pi}(\alpha y)\geq \alpha\cdot \optf_\mc(y)$ for all $y\in\real$.
\end{definition}

\smallskip
\noindent Notice that in the maximization setting, we have an $\alpha$-local approximations for $\alpha\in (0,1]$. By following exactly the same steps as in the proof of~\Cref{thm:la-comp}, the following composition theorem is immediate:
\begin{theorem}[Composition theorem for maximization settings]\label{thm:la-comp-max}
Let $\instance = (\vecM , \feas)$ be any instance of matroid-max-CICS, where each constituent MDP $\mc_i$ admits an $\alpha$-local approximation under some commitment $\pi_i\in\comset{\mc_i}$. Then, $\cg (\instance)\geq \alpha$.
\end{theorem}

\medskip
\noindent By combining~\Cref{thm:la-comp-max} with the optimality of the Water Draining Index policy (\Cref{thm:MC-opt-maxim}) and the upper bound of~\Cref{thm:MDP-lb-max}, we obtain a way to efficiently approximate the optimal solution for matroid-max-CICS assuming that the underlying MDPs achieve good local approximation guarantees.
\newpage
\section{The Combinatorial Setting}\label{app:combinatorial}

In this section, we describe how our entire framework extends to combinatorial settings beyond the case of matroids. We will state our results in full generality and distinguish between minimization and maximization whenever needed. The framework we consider is based on~\citep{S17} and its followup~\citep{GJSS19}, where the authors develop a technique for combinatorial selection over Markov chains. In this section, we show that under local approximation, their results can be seamlessly extended to arbitrary MDPs.

\begin{definition}[CICS] A Costly Information with Combinatorial Inspection (CICS) instance $\instance$ is defined with respect to a set of Costly Information MDPs $\{\mc_i\}_{i=1}^n$, a feasibility constraint $\feas\subseteq 2^{[n]}$ and a function $h:\feas\mapsto \real$. At each step, an algorithm chooses one of the MDPs and advances it through one of its actions. The game terminates once the algorithm accepts a feasible\footnote{We assume that the algorithm will always ensure that this happens, i.e. when having already accepted a set of MDPs $S$, it will never accept another MDP $i$ for which $S\cup \{i\} \cup A\notin \feas$ for all  $A\subseteq [n]$.} set $S\in\feas$ of MDPs. 

For a specific run of the algorithm, let $A$ denote the set of all actions the algorithm took, $S\in\feas$ denote the set of terminated MDPs and $T$ be the corresponding set of terminals it accepted. Then, the total cost of the algorithm for this run in the minimization setting is given by
\[\sum_{t\in T} v(t) + h(S) + \sum_{a\in A} c(a) \]
and the total utility of the algorithm for the run in the maximization setting is given by 
\[ \sum_{t\in T} v(t) + h(S) - \sum_{a\in A} c(a).\]
The optimal adaptive policy in the minimization (resp. maximization) setting is the policy of minimum (resp. maximum) expected cost (resp. utility).
\end{definition}

We note that there are two differences with respect to the matroid setting. The first is that $\feas$ can now be any arbitrary set of constraints. In fact, we won't even have to assume that the set is downwards or upwards close, as long as the algorithms always ensure that they accept a feasible set of MDPs. The second is the addition of the function $h(\cdot)$; this is a function that encodes an extra cost (or reward, depending on the setting) that does not depend on the precise terminals of the MDPs that were selected, but only on the set of terminated MDPs.

The notions of amortization, water filling/draining, optimality curves, surrogate costs and local approximation apply to each MDP separately, and thus are independent of the underlying combinatorial setting. Therefore, all these definitions extend to the combinatorial setting without change. Our \textbf{first contribution} is to extend~\Cref{thm:MDP-lb} beyond the matroid setting and prove that the performance of the optimal adaptive policy in any combinatorial setting is bounded by the surrogate costs of the underlying MDPs. We note that while this result was known in the single-selection setting, we are the first to prove it for the general combinatorial setting. 

\begin{theorem}\label{thm:comb-bound}
    Let $\instance = (\mc_1, \cdots , \mc_n , \feas , h)$ be any instance of CICS. For each $i\in [n]$, let $W^*_{\mc_i}$ be the water filling (resp. water draining) surrogate costs in the minimization (resp. maximization) setting. Then:
    \begin{enumerate}
        \item For the minimization setting, $\mathrm{OPT}(\instance)\geq  \expect{\min_{S\in\feas}(\sum_{i\in S} W^*_{\mc_i} + h(S))}$.
        \item For the maximization setting, $\mathrm{OPT}(\instance)\leq  \expect{\max_{S\in\feas}(\sum_{i\in S} W^*_{\mc_i} + h(S))}$.
    \end{enumerate}
\end{theorem}

Our \textbf{second contribution} is to show that local approximation continues to imply composition results even in the combinatorial setting. In particular, we prove the following extension of~\Cref{thm:la-comp}.

\begin{theorem}\label{thm:comb-la}
    Let $\instance = (\mc_1, \cdots , \mc_n , \feas , h)$ be any instance of CICS. For each $i\in [n]$, let $W^*_{\mc_i}$ be the water filling (resp. water draining) surrogate costs in the minimization (resp. maximization) setting. Finally, for each $i\in [n]$, let $\pi_i\in\comset{\mc_i}$ be some commitment that $\alpha$-locally approximates $\mc_i$. Then:
    \begin{enumerate}
        \item For the minimization setting, $\expect{\min_{S\in\feas}(\sum_{i\in S} W^*_{\mc^{\pi_i}_i} + h(S))}\leq  \alpha\cdot \expect{\min_{S\in\feas}(\sum_{i\in S} W^*_{\mc_i} + h(S))}$.
        \item For the maximization setting, $\expect{\max_{S\in\feas}(\sum_{i\in S} W^*_{\mc^{\pi_i}_i} + h(S))}\geq  \alpha\cdot \expect{\max_{S\in\feas}(\sum_{i\in S} W^*_{\mc_i} + h(S))}$.
    \end{enumerate}
\end{theorem}

By combining~\Cref{thm:comb-bound} with~\Cref{thm:comb-la}, and assuming that our MDPs admit local approximation guarantees, we are left with the task of optimizing over the CICS instance that is generated by the commitments; observe that this is now an instance over Markov chains. In the case of matroid constraints, we showed that we can always efficiently achieve this via the water filling/draining index policy. However, depending on the combinatorial constraint, efficient optimization might not be possible -- consider for example the case where each MDP is a single terminal state corresponding to a set of elements and we need to select a minimum cost set cover.

The final key to the puzzle will be a way of efficiently approximating the optimal policy for a CICS instance $\instance = (\mc_1, \cdots , \mc_n , \feas , h)$ over Markov chains. This is precisely the setting that is considered by~\cite{GJSS19}. In this work, it is shown that a sufficient condition to get an efficient $\beta$-approximation to the optimal policy for a CICS instance over Markov chains is for the underlying pair $(\feas, h)$ to admit a $\beta$-approximate frugal algorithm. We defer the reader to~\cite{GJSS19} for the full details. Here, we will just mention some examples of such combinatorial settings:
\begin{itemize}
    \item Matching constraints in the maximization setting admit a $0.5$-approximate frugal algorithm.

    \item Facility location constraints in the minimization setting admit a $1.861$-approximate frugal algorithm.

    \item Set cover constraints in the minimization setting admit a $\min(f, \log n)$-approximate frugal algorithm where $f$ is the maximum number of sets in which a ground element is present.
\end{itemize}

\bigskip   
\noindent Combining everything, we obtain the following result for the combinatorial setting:
\begin{corollary}
    Let $\instance = (\mc_1, \cdots , \mc_n , \feas , h)$ be any instance of CICS such that:
    \begin{enumerate}
        \item Each MDP $\mc_i$ admits an $\alpha$-local approximation.
        \item The combinatorial setting $(\feas, h)$ admits a $\beta$-frugal algorithm.
    \end{enumerate}
    Then, $\cg(\instance)\leq\alpha\cdot\beta$ for minimization, and $\cg(\instance)\geq\alpha\cdot\beta$ for maximization.
\end{corollary}

In other words, if (i) we can efficiently compute an $\alpha$-local approximate commitment for each MDP and (ii) the combinatorial setting admits a frugal algorithm, then we can efficiently construct a committing policy that $(\alpha\beta)$-approximates the optimal (non-committing) policy. We are left with the task of proving our extended Theorems~\ref{thm:comb-bound} and~\ref{thm:comb-la}. It is not hard to see that both proofs are identical to their matroid counterparts and are thus omitted. In particular:
\begin{enumerate}
    \item The proof of~\Cref{thm:comb-bound} follows exactly the same steps as the proofs of~\Cref{thm:MDP-lb} (for the minimization setting) and~\Cref{thm:MDP-lb-max} (for the maximization setting) that were presented in~\Cref{app:amortization} and~\Cref{app:maxim} respectively. In these proofs, we separately bounded the cost/utility contribution of each MDP by the corresponding surrogate costs/values and simply invoked the feasibility of the optimal adaptive policy at the end. Thus, the exact same proofs imply~\Cref{thm:comb-bound} for any feasibility constraint $\feas$ and any cost/reward function $h(\cdot)$.

    \item The proof of~\Cref{thm:comb-la} follows exactly the same steps as the proofs of~\Cref{thm:la-comp} (for the minimization setting) and~\Cref{thm:la-comp-max} that were presented in~\Cref{app:local_approx} and~\Cref{app:maxim} respectively. In particular, none of these proofs used the fact that $\feas$ is a matroid constraint at any point, and it is also straightforward to see that they immediately extend for any cost/reward function $h(\cdot)$.
\end{enumerate}

\newpage
\section{Challenges of Composing Local Conditions}\label{app:example}

In this section, we demonstrate the challenges of extending Whittle's condition to obtain bounds on the commitment gap. Fix any instance $\instance = (\vecM,\feas)$ of CICS. Whittle's condition states that if the (local) commitments $\pi_i$ are optimal for \textit{all} local games $(\mc_i,y)$, then $\Pi = (\pi_1,\cdots, \pi_n)$ is optimal for $\instance$ and thus the commitment gap is $1$. Naturally, this is a very strong local condition and it is desirable to relax it into an approximation guarantee, while still being able to compose it into a global bound.

A natural approach would be to try and extend this result to approximation guarantees over the local games. However, this does not work. Consider, for example, a minimization MDP $\mc$ with two $0$-cost actions: the first action leads to a deterministic value of $1$, and the second uniformly leads to a stochastic value of $0$ or $1$. Committing to each of these actions leads to an expected cost of $\min(y,1)$ and $1/2\cdot \min(y,1)$ in the local game respectively; the deterministic action achieves a $2$-approximation for all local problems $(\mc, y)$. Now consider a single-selection CICS with $n$ copies of the above MDP. If we commit to the deterministic action in each MDP, our cost is deterministically $1$, whereas taking the stochastic action in each costs us $1/2^n$ -- an exponential gap! 

A different approach would be to consider the negation of Whittle's condition: if a (local) commitment $\pi_i$ is strictly suboptimal for all local game $(\mc_i, y)$, then no optimal policy for the CICS instance $\instance$ will commit to $\pi_i$. This could potentially allow us to rule out some of the commitment policies and simplify our instance. However, we show that this isn't true: there are instances of CICS where the commitment gap is $1$ (and thus the optimal policy is a committing policy) and yet the actions it commits to are unambiguously suboptimal!


\paragraph{The Example.} We consider an instance of single-selection min-CICS consisting of two MDPs ($\mc_1,\mc_2)$. The MDP $\mc_1$ is a depth $1$ Markov chain, consisting of a single action of cost $3$ and resulting to a uniformly random terminal state of value $0$ or $50$. The MDP $\mc_2$ is of height $2$. On the first level, there are two actions $a_1,a_2$ of cost $c(a_1)=c(a_2)=0$. Taking action $a_1$ deterministically results to a terminal state of value $5$. Taking action $a_2$ results to a terminal state of value $0$ with probability $1/2$ and to a non-terminal state $s$ with probability $1/2$. On state $s$, there are two more actions $a_3,a_4$ of cost $c(a_3)=c(a_4)=0$. Taking $a_3$ results to a terminal of value $12$ and taking $a_4$ results to a uniformly random terminal of value $0$ or $50$. 

Notice that there are three committing policies in $\comset{\mc_2}$, corresponding to action sequences $\{a_1\}$, $\{a_2,a_3\}$ and $\{a_2,a_4\}$; we use $\pi_1$, $\pi_{23}$ and $\pi_{24}$ to denote these policies. As we have seen, each of these committing policies $\pi$ implies a Markov chain $\mc_2^\pi$ and an optimality curve for the local game $(\mc_2^\pi , y)$, denoting the cost of the optimal policy for the local game as a function of the outside option $y\geq 0$. For this specific example, it is not hard to see that:
\begin{itemize}
    \item $f_{1}(y) := f_{\mc^{\pi_1}_2}(y) = \min\{y,5\}$.
    
    \item $f_{23}(y) := f_{\mc^{\pi_{23}}_2}(y) =0.5\cdot \min\{y,12\}$.

    \item $f_{24}(y) := f_{\mc^{\pi_{24}}_2}(y) = 0.25\cdot \min\{y,50\}$.

\end{itemize}
\noindent A pictorial representation of the two MDPs and the corresponding optimality curves for the three committing policies on $\mc_2$ is shown in~\Cref{fig:example}.

\begin{figure}[h!]
    \centering
    \includegraphics[width=0.9\textwidth]{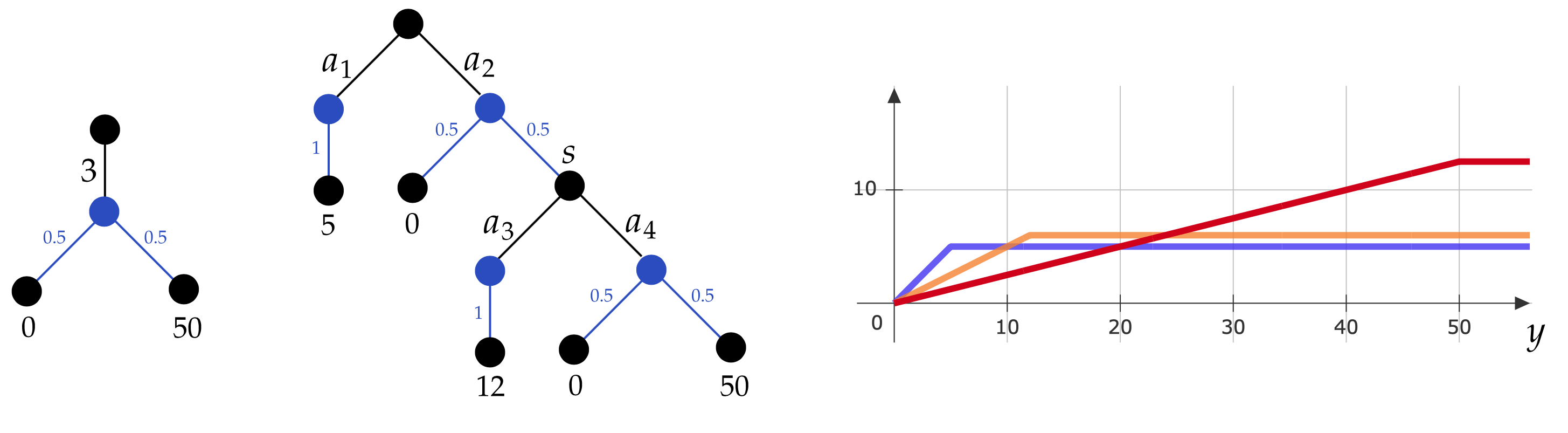}
    \caption{The MDPs $\mc_1$ (left) and $\mc_2$ (right). All the actions in $\mc_2$ have a cost of $0$. There are three committing policies for $\mc_2$ with optimality curves $f_1(y)$ (blue), $f_{23}(y)$ (orange) and $f_{24}(y)$ (red).}
    \label{fig:example}
\end{figure}

The optimality curves demonstrate that for any value of $y$, the committing policy $\pi_{23}$ is sub-optimal for the local game $(\mc_2, y)$; in other words, the optimal policy for the local game will \textit{never} take action $a_3$. One could expect that since $a_3$ is always suboptimal, the optimal policy would never take this action in the context of the bandit superprocess $(\mc_1,\mc_2)$. However, we can show that this isn't true. On the contrary, the optimal algorithm for $(\mc_1,\mc_2)$ will commit to $\pi_{23}$!
\begin{claim}
    The optimal policy for the single-selection min-CICS instance $(\mc_1,\mc_2)$ is the committing policy that commits to the unique commitment for Markov chain $\mc_1$ and to $\pi_{23}$ for MDP $\mc_2$.
\end{claim}
\begin{proof}
    We prove the claim by computing the expected cost of all possible policies for the single-selection min-CICS instance $(\mc_1,\mc_2)$ and showing that the proposed committing policy has the minimum expected cost. At the beginning, every policy needs to either take the costly action of $\mc_1$ or take either action $a_1$ or $a_2$ from $\mc_2$.

    \begin{itemize}
        \item  We first consider the set of policies that start by taking the costly action of $\mc_1$ at a cost of $3$; if the realized terminal has value $0$ then we should clearly terminate, otherwise we play the local game on $\mc_2$ for $y=50$. In other words, the optimal policy that starts by taking the costly action of $\mc_1$ has expected cost
        \[3 + \frac{1}{2}\cdot 0 + \frac{1}{2}\cdot \min\{f_1(50), f_{23}(50), f_{24}(50)\} = 5.5.\]

        \item Up next, we consider the set of policies that start by taking action $a_1$ in $\mc_2$; in that case, we should play the local game on $\mc_1$ with outside option $y=5$ optimally. In other words, the optimal policy that starts by taking action $a_1$ on $\mc_2$ has expected cost
         \[\min\{5, 3 + \frac{1}{2}\cdot 0 + \frac{1}{2}\cdot 5 \} = 5.\]

         \item Finally, we need to consider the set of policies that start by taking action $a_2$. If the realized outcome is the terminal state of value $0$ (this happens with probability $1/2$), then we should clearly accept it and terminate at a total cost of $0$. Otherwise, we are left with the task of computing the optimal policy for the single-selection min-CICS $(\mc_1,\mc_s)$ where $\mc_s$ is the sub-process of $\mc_2$ rooted at state $s$. In other words, the cost of the optimal policy that starts by taking action $a_2$ on $\mc_2$ will be precisely
         \[\frac{1}{2}\cdot\mathrm{OPT}(\mc_1,\mc_s).\]
         We will now compute the cost of the optimal policy for the sub-instance $(\mc_1,\mc_s)$. Observe that any policy will either start by taking $a_3$ or $a_4$ or by taking the costly action in $\mc_1$; in the latter case, if the outcome is $0$ the policy should terminate, otherwise there is once again an optimal choice between $a_3$ and $a_4$. Thus, we can assume without loss that the optimal policy for the sub-instance $(\mc_1,\mc_s)$ will commit in advance to either $a_3$ or $a_4$. Furthermore, since both $a_3$ and $a_4$ have a cost of $0$, we can assume that we start by taking the corresponding action. Thus:

         \begin{itemize}[$\bullet$]
             \item If the policy commits to $a_3$, the optimal cost for the sub-instance $(\mc_1,\mc_s)$ is
             \[\min\{12 , 3 + \frac{1}{2}\cdot 0 + \frac{1}{2}\cdot \min\{12,50\}\} = 9.\]

             \item If the policy commits to $a_4$, the optimal cost for the sub-instance $(\mc_1,\mc_s)$ is
             \[\frac{1}{2}\cdot 0 + \frac{1}{2}\cdot \min\{50, 3 + \frac{1}{2}\cdot \min\{0,50\} + \frac{1}{2}\cdot \min\{50,50\}\} = 14.\]
         \end{itemize}

    Thus, the cost of the optimal policy that starts by taking action $a_2$ on $\mc_2$ will be $(1/2)\cdot 9 = 4.5$.
    \end{itemize}
    From the above case analysis, we obtain that the optimal policy for $(\mc_1,\mc_2)$ will start by taking action $a_2$, then action $a_3$ and finally the costly action of $\mc_1$ - the policy terminates whenever a terminal state of value $0$ is reached or it runs out of actions to take (in which case it accepts the terminal of value $12$). Clearly, this is equivalent to the optimal policy that commits to $\pi_{23}$, concluding the proof.

\end{proof}

\newpage
\newpage
\section{Omitted Proofs from~\Cref{sec:amortization} (Amortization)}\label{app:amortization}

In this chapter of the appendix, we present all omitted proofs from~\Cref{sec:amortization}. We restate all our results for the reader's convenience.

\subsection{Amortized Surrogate Costs for Markov Chains}
In this section, we provide the formal proofs and extra details on all our results for matroid-min-CICS over Markov chains that we established in~\Cref{sec:mc-amortization}. 

We begin with the proof of the lower bound we stated in~\Cref{lem:MC-lb}, relating the cost of any algorithm to the surrogate costs of any amortization. We note that one way to prove~\Cref{lem:MC-lb} is to follow the approach developed by \cite{S17} and relate the performance of any algorithm on a ``costly information'' instance to its performance in a ``free information world'' where action costs are paid by an outside investor who is in turn paid back the extra amortized cost of any terminal that the algorithm follows and accepts. Instead, we provide an algorithmic proof that will allow us to directly argue about the optimality of the water filling index policy, as well as extend our setting to MDPs.

\mclbtag*
\begin{proof}
Let $\mathrm{ALG}$ be any algorithm for $\instance$. We will decompose the expected cost of $\mathrm{ALG}$ into the expected cost that it pays in every Markov chain $\mc_i$ for all $i\in [n]$. Let $\mc_i=(S_i, \sigma_i, A_i, c_i, \dist_i, V_i, T_i)$ be the $i$-th Markov chain and let $\surr_i := \surr_{\mc_i, b_i}$ be any cost amortization for this chain. We drop the $i$-index for notational convenience and define the following random events:
    \begin{itemize}
        \item $I(s) :=$ the event that $\mathrm{ALG}$ advances $\mc$ at the non-terminal state $s\in S$.
        \item $R(\tau):=$ the event that nature realizes a trajectory $\tau'\in\traj$ with suffix $\tau$.
        \item $R(s):=$ the event that nature realizes a trajectory $\tau\in\traj$ with $s\in\tau$.
        \item $A:=$ the event that $\mathrm{ALG}$ accepts a terminal state $t\in T$.
    \end{itemize}
Then, the expected cost paid by the algorithm on Markov chain $\mc$ can be written as:
\begin{equation}
\label{eq:amor1}
    \mathrm{ALG}(\mc) = \sum_{s\in S}\prob{I(s)}\cdot c(s) + \sum_{\tau\in\traj} \prob{R(\tau)}\cdot\prob{A|R(\tau)}\cdot v(\tau) 
\end{equation}

For the expected cost of taking costly actions, we have that
\begin{equation}\label{eq:amor2}
    \sum_{s\in S}\prob{I(s)}\cdot c(s) = \sum_{s\in S}\prob{I(s)}\cdot \sum_{\tau\in \traj(s)}\frac{\prob{R(\tau)}}{\prob{R(s)}}\cdot b_{s\tau} = \sum_{s\in S}\prob{I(s)}\cdot \sum_{\tau\in \traj(s)}\prob{R(\tau)|R(s)}\cdot b_{s\tau}
\end{equation}
where the first equality follows by the definition of cost shares and the second equality follows by $\tau'\in\traj (s)$.

For the expected cost of accepting a terminal state, we have that
\begin{equation}\label{eq:amor3}
    \sum_{\tau\in\traj} \prob{R(\tau)}\cdot\prob{A|R(\tau)}\cdot v(\tau)  = \sum_{\tau\in\traj} \prob{R(\tau)}\cdot\prob{A|R(\tau)}\cdot \big(\surr(\tau) - \sum_{s\in \tau}b_{s\tau_s}\big)
\end{equation}
by definition of the surrogate cost $\surr(\tau)$. We further decompose this expression into two terms:
\begin{itemize}
    \item For the first term, we have
    \begin{equation}\label{eq:amor4}
        \sum_{\tau\in\traj} \prob{R(\tau)}\cdot\prob{A|R(\tau)}\cdot \surr(\tau) 
    = \sum_{\tau\in\traj} \prob{R(\tau)}\cdot\expect{A\cdot \surr(\tau) |R(\tau)} = \expect{A\cdot \surr}
    \end{equation}
    with the first equality following from $A\in \{0,1\}$ and the second by definition of the surrogate cost.

    \item For the second term, we have
    \begin{align}\label{eq:amor5}
    \begin{split}
        \sum_{\tau\in\traj} \prob{R(\tau)}\cdot\prob{A|R(\tau)}\cdot \sum_{s\in \tau}b_{s\tau_s} &= \sum_{s\in S}\sum_{\tau\in \traj (s)}b_{s\tau}\cdot \sum_{\tau'\in\traj : \tau \subset \tau'}\prob{R(\tau')}\cdot \prob{A|R(\tau')} \\
        &= \sum_{s\in S}\sum_{\tau\in \traj (s)}b_{s\tau}\cdot \prob{A\cap R(\tau)} 
    \end{split}
    \end{align}
    with the first line following by exchanging the summation order and the second line following by definition of the event $R(\tau)$, specifically that $R(\tau) = \sum_{\tau'\in\traj : \tau \subset \tau'}R(\tau')$.
\end{itemize}
Combining Equations~\eqref{eq:amor1} through~\eqref{eq:amor5}, we obtain that the expected cost of the algorithm on Markov chain $\mc$ can be written as
\[\mathrm{ALG}(\mc) = \expect{A\cdot \surr} + \sum_{s\in S} \sum_{\tau\in \traj(s)} b_{s\tau}\cdot\bigg(\prob{I(s)}\cdot \prob{R(\tau)|R(s)} - \prob{A\cap R(\tau)} \bigg) .  \]

Observe that the second term is always non-negative: in order to accept a terminal from the chain \textbf{and} for nature to realize a trajectory $\tau\in \surr (s)$, it is necessary for the algorithm to advance state $s$ (this already requires that $s$ belongs in the realized trajectory) \textbf{and} for the suffix $\tau$ to be realized conditioned on having reached $s$; notice that the last two events are independent. Thus, we have concluded that 
\[\mathrm{ALG}(\mc) \geq \expect{A\cdot \surr}.\]

The proof is completed by noting that
\[\expect{\text{cost}(\mathrm{ALG})} = \sum_{i=1}^n \mathrm{ALG}(\mc_i) \geq \sum_{i=1}^n \expect{A_i\cdot \surr_i} = \expect{\sum_{i=1}^n A_i \surr_i} \geq \expect{\min_{S\in\feas}\sum_{i\in S}\surr_i}\]
with the last inequality obtained by the fact that the set $S=\{i\in [n] : A_i=1\}$ must be feasible (i.e. $S\in\feas$) with probability $1$ for $\mathrm{ALG}$ to be a valid policy.
\end{proof}

\medskip

Observe that in the proof of~\Cref{lem:MC-lb}, we only used inequalities at two points: (i) when relating the amortized cost of the trajectories that the algorithm selected to the (feasible) set of trajectories of minimum total surrogate cost and (ii) when we established that $\prob{I(s)}\cdot \prob{R(\tau)|R(s)}\geq \prob{A\cap R(\tau)}$ for all states $s$ and all trajectories $\tau\in\traj (s)$; notice that this was only required for pairs $(s,\tau)$ with $b_{s\tau}>0$. The first inequality can become tight if the algorithm somehow ensures that it will always accept the feasible set of realized trajectories whose total surrogate cost is minimum. The second inequality becomes tight if the algorithm can somehow ensure that whenever it advances a state $s$ sharing cost $b_{s\tau}>0$ with one of its downstream trajectories $\tau$, it will always accept the corresponding Markov chain if the suffix $\tau$ gets realized. This allows us to characterize the conditions under which the lower bound is actually met by an algorithm, which we formally state as the following corollary:
\begin{corollary}\label{cor:mc-opt-conditions}
    Consider any instance $\instance = (\vecM, \feas)$ of matroid-min-CICS over Markov chains and let $b_i$ be any cost amortization of $\mc_i$ with surrogate cost $\surr_i := \surr_{\mc_i, b_i}$ for all $i\in [n]$. Then, any algorithm that satisfies both
    \begin{enumerate}
        \item \textbf{Surrogate Optimality.} The algorithm always accepts a feasible set of Markov chains whose realized trajectories have the minimum total surrogate cost.
        \item \textbf{Promise of Payment.} Whenever one of the Markov chains $\mc_i$ gets advanced from a state $s_i\in S_i$ such that $b_{s_i\tau_i}>0$ for some $\tau_i \in \traj (s_i)$, the algorithm will keep advancing $\mc_i$ as long as it's trajectory evolves according to $\tau_i$.
    \end{enumerate}
    will have expected cost precisely $\expect{\min_{S\in\feas} \sum_{i\in S}\surr_i}$.
\end{corollary}

Promise of payment states that whenever $b_{s\tau}>0$ and the algorithm advances $s$, it will ensure that while it is possible for $\tau$ to be realized, the corresponding Markov chain will continue to be advanced; in other words, that $\prob{I(s)}\cdot \prob{R(\tau)|R(s)} = \prob{A\cap R(\tau)}$. We note that achieving any of these properties individually is trivial; fully advancing all Markov chains and picking the feasible set of trajectories of minimum total surrogate cost achieves surrogate optimality, and either advancing a Markov chain fully or not advancing it at all achieves promise of payment. The water filling amortization is specifically defined so that the corresponding index policy achieves both of these properties, simultaneously:
\mcopttag*
\begin{proof}
It suffices to argue that the water filling index policy, paired with the water filling amortization, satisfies both surrogate optimality and promise of payment. We begin with surrogate optimality. By definition, the water filling index policy advances the chain of minimum water filling index; recall that the index of a state corresponds to the minimum surrogate cost among all trajectories passing through it. This ensures that while there is potential for some chain $\mc_i$ to realize the trajectory of minimum surrogate cost, the algorithm will keep advancing it. In other words, this algorithms ends up \textbf{greedily} accepting Markov chains with respect to the surrogate cost of their realized trajectories. Paired with the fact that a minimum cost basis of a matroid is always obtained by greedily adding the cheaper feasible element, this establishes surrogate optimality of the water filling index policy.

Next, we show that promise of payment also holds. Say that at any point the algorithm advances Markov chain $\mc_i$ at state $s_i$ with index $I_i(s_i)$ and there exists some trajectory $\tau_i\in\traj(s_i)$ such that $b_{s_i\tau_i}>0$. We need to argue that as long as the state of $\mc_i$ evolves according to $\tau_i$, the algorithm will keep advancing it. Since $\mc_i$ was advanced by the water filling index policy, this implies that $I_i(s_i)$ was the minimum index among the current states of all Markov chains; thus, a sufficient condition is to show that $I_i(s'_i) \leq I_i(s_i)$ for all $s'_i\in \tau_i$, as then all the states in the trajectory will continue to have the minimum index and thus will be advanced if realized. Since $b_{s_i\tau_i}>0$, we know by definition of the water filling amortization that immediately after $s_i$ got amortized, $\tau_i$ had the minimum surrogate cost across all trajectories in $\traj (s_i)$. As states are considered in a bottom-up order, this implies that any subsequent steps of the amortization will maintain that the trajectory passing through $s_i$ that has the minimum surrogate cost ends with $\tau_i$. This immediately gives us the proof, by definition of the indices and $s'_i\in\tau_i$.
\end{proof}

\subsection{Amortized Surrogate Costs for MDPs}
In this section we formally prove our main result from~\Cref{sec:amortization}, namely~\Cref{thm:MDP-lb}. We note that the proof mirrors our approach for proving the same result in the special case of Markov chains (\Cref{lem:MC-lb}), and is enabled by the characterization of~\Cref{lem:mdp-wf-char-new}.
\mdplbtag*
\begin{proof}
Fix any algorithm (including the optimal adaptive policy) for instance $\instance$ and let $\mathrm{ALG}$ denote the (random) cost of this algorithm. Also, for all $i\in [n]$, let $\mathrm{ALG}(i)$ denote the (random) cost suffered by the algorithm from paying action costs and accepting terminal states in $\mc_i$. Notice that $\mathrm{ALG} = \sum_{i=1}^n \mathrm{ALG}(i)$ with probability $1$. Finally, let $X(i)$ be the random variable indicating whether the algorithm accepts a state from $\mc_i$ or not. To prove~\Cref{thm:MDP-lb}, we will argue that for all $i\in [n]$,
\begin{equation}\label{eq:to_show}
   \expect{\mathrm{ALG}(i)} \geq \expect{X(i)\cdot W^*_{\mc_i}}. 
\end{equation}
Notice that if this is true, then we immediately have that
\[\expect{\mathrm{ALG}} = \sum_{i=1}^n\expect{\mathrm{ALG}(i)} \geq \sum_{i=1}^n\expect{X(i)\cdot W^*_{\mc_i}} = \expect{\sum_{i=1}^nX(i)\cdot W^*_{\mc_i}} \geq \expect{\min_{S\in\feas}\sum_{i\in S}W^*_{\mc_i}} \]
with the last inequality following from the fact that for any feasible algorithm, the set of accepted terminals $S=\{i : X(i)=1\}$ must be feasible (i.e. $S\in\feas$) with probability $1$.

We now turn our attention to proving inequality~\eqref{eq:to_show}. We use $\A$ to denote the inner randomness of the algorithm and $\R_i$ to denote the randomness of each MDP $\mc_i$ for all $i\in [n]$. Notice that conditioned on $\A$ and $\R_i$ for all $i\in [n]$, the outcome of the algorithm is deterministic. Now fix any $i\in [n]$ and let \[\R^{-i} := \A \cup\big( \cup_{j\neq i}\R_j)\]
encode all the randomness in the algorithm's run \textbf{except} from the realizations of $\mc_i$. The key observation is that conditioned on $\R^{-i}$, all the actions that the algorithm takes in MDP $\mc_i$ are \textit{predetermined}; in other words, the algorithm's trajectory on $\mc_i$ is fully described by some deterministic commitment denoted $\pi_i = \pi_i(\R^{-i})\in\comset{\mc_i}$. Thus, we have that
\begin{equation}\label{eq:com-pol-cost}
\expect{\mathrm{ALG}(i)} = \expectt{\R^{-i}}{\expectt{\R_i}{\mathrm{ALG}(i) | R^{-i}}} = \expectt{\R^{-i}}{\mathrm{cost}(\pi_i)}
\end{equation}
where $cost(\pi_i)$ is the expected cost of the algorithm, following the commitment $\pi_i = \pi_i(\R^{-i})$, on MDP $\mc_i$. The proof is the completed by the following generalization of~\Cref{lem:MC-lb}, which is enabled from~\Cref{lem:mdp-wf-char-new}.
\begin{claim}\label{claim:mdp-lb-helper}
    Fix any MDP $\mc$ and any deterministic commitment $\pi\in\comset{\mc}$. Then, the expected cost of any algorithm following $\pi$ on $\mc$ will be at least
    \[\expect{X(\pi)\cdot W^*_{\mc}}\]
    where $X(\pi)$ is an indicator of whether the algorithm accepts a terminal state of $\mc$ or not.
\end{claim}
Notice that the above claim doesn't depend on the underlying CICS instance $\instance$; it simply states that conditioned on running a committing policy on some MDP and accepting a terminal state, the expected total cost spent on this MDP is lower bounded by the surrogate cost. Since this lower bound applies to all commitments $\pi\in\comset{\mc}$, coupled with equation~\eqref{eq:com-pol-cost}, it directly implies inequality~\eqref{eq:to_show}, as
\[\expect{\mathrm{ALG}(i)} =  \expectt{\R^{-i}}{\mathrm{cost}(\pi_i)} \geq \expectt{R^{-i}}{\expect{(X(i)|R^{-i})\cdot W^*_{\mc_i}}} = \expect{X(i)\cdot W^*_{\mc_i}},\]
completing the proof.

\paragraph{Proof of~\Cref{claim:mdp-lb-helper}}. Notice that since we are committing to running policy $\pi\in\comset{\mc}$ on $\mc$, we are essentially running some algorithm on the Markov chain $\mc^\pi$. From~\Cref{lem:MC-lb} and~\Cref{thm:MC-opt}, we know that the contribution of Markov chain $\mc^\pi$ to the total cost will be at least
\[\expect{X(\pi)\cdot W^*_{\mc^\pi}}\]
and thus to prove the claim, we will need to show that 
\[\expect{W^*_{\mc^\pi}} \geq \expect{W^*_\mc}\]
for all $\pi\in\comset{\mc}$.

Let $\mc=(S, \sigma, A, c, \dist, V, T)$ and recall that we use $S_\pi$ and $\traj_\pi$ to denote the state space and trajectory set of $\mc^\pi$ and $p_\pi$ to denote the implied distribution over $\traj_\pi$. Finally, we use $c(s)$ to denote the cost of the unique action that the deterministic commitment $\pi$ chooses at state $s\in S_\pi$. By definition of the water filling surrogate cost of a Markov chain, we have
\begin{align*}
    \expect{W^*_{\mc^\pi}} &= \sum_{\tau\in \traj_\pi} p_\pi(\tau)\cdot \surr^*_{\mc^\pi}(\tau) = \sum_{\tau\in \traj_\pi} p_\pi(\tau)\cdot \bigg( v(\tau) + \sum_{s\in\tau}b^*_{s\tau_s} \bigg)
\end{align*}
where $b^*_{s\tau}$ are non-negative cost shares satisfying $\sum_{\tau\in\traj_\pi(s)} p_\pi(\tau) b^*_{s\tau} = p_\pi(s)c(s)$ for all $s\in S_\pi$. Since $\pi$ is deterministic, we know from \Cref{lem:mdp-wf-char-new} that there exists an amortized cost function $\surr^\pi (\cdot)$ over the trajectories in $\traj_\pi$ and a set of cost shares $\{b^\pi_{s\tau}\}$ that will satisfy action independence, cost sharing and cost dominance. Using these, we have
\begin{align*}
    \expect{W^*_{\mc^\pi}}
    &= \sum_{\tau\in \traj_\pi} p_\pi(\tau)\cdot \bigg( v(\tau) + \sum_{s\in\tau}b^*_{s\tau_s} \bigg) \\
    &= \sum_{\tau\in \traj_\pi} p_\pi(\tau)\cdot \bigg( \expect{\surr^\pi (\tau)} -\sum_{s\in\tau}b^\pi_{s\tau_s}  + \sum_{s\in\tau}b^*_{s\tau_s} \bigg) &&\text{(Cost Sharing)}\\
    &= \sum_{\tau\in \traj_\pi} p_\pi(\tau)\cdot  \expect{\surr^\pi (\tau)} + \sum_{\tau\in \traj_\pi}\sum_{s\in \tau}p_\pi(\tau)(b^*_{s\tau_s} - b^\pi_{s\tau_s})\\
    &= \expect{W^*_\mc} + \sum_{s\in S_\pi}\sum_{\tau\in \traj_\pi(s)}p_\pi(\tau)(b^*_{s\tau} - b^\pi_{s\tau}) &&\text{(Action Independence)} \\
    &= \expect{W^*_\mc} + \sum_{s\in S_\pi} p_\pi(s)c(s) - \sum_{\tau\in \traj_\pi(s)} p_\pi(\tau) b^\pi_{s\tau} \\
    &\geq \expect{W^*_\mc}. &&\text{(Cost Dominance)}
\end{align*}
\end{proof}

\medskip
\subsection{Second Order Stochastic Dominance}

Finally, in this section we provide a proof for~\Cref{lemma:sdom}. Once again, we note that this result is standard; here, we provide our own constructive proof for the sake of completeness and building intuition on how the water filling surrogate cost of an MDP is obtained. For simplicity, we refer to~\Cref{lemma:sdom} as the Stochastic Dominance Lemma, henceforth SDL.
\stdomtag*
\begin{proof}
We will say that $X\preceq Z$ if there exists a mapping $m:\mathrm{supp}(Z)\mapsto \Delta(\mathrm{supp}(X))$ such that the two conditions of the SDL hold. We associate each random variable $W$ with a function $f_W(y):=\expect{\min\{y,W\}}$ over $y\in\real$. In other words, we want to prove that \[f_X(y)\leq f_Z(y)\;\;\forall y\in\real \Longrightarrow X\preceq Z.\]

We first prove \textbf{transitivity} of our condition. In other words, if $X,Y,Z$ are discrete random variables such that $X\preceq Y$ and $Y\preceq Z$, we also have that $X\preceq Z$. The proof is immediate; let $m_1: \mathrm{supp}(Z) \mapsto \Delta(\mathrm{supp}(Y))$ be the corresponding mapping for $Y\preceq Z$ and $m_2: \mathrm{supp}(Y) \mapsto \Delta(\mathrm{supp}(X))$ be the corresponding mapping for $X\preceq Y$. Then, the composition mapping $m:=m_2\circ m_1$ that maps each $z\in\mathrm{supp}(Z)$ to a random realizations of $m_2(y)$ for a randomly sampled $y\sim m_1(z)$ immediately satisfies both conditions and yields $X\preceq Z$. Formally, for each $x\in\mathrm{supp}(X)$ we have
\[\prob{X=x}  =  \sum_{y,z} \prob{Z=z}\cdot\prob{m_1(z)=y}\cdot \prob{m_2(y)=x} =\sum_z \prob{Z=z}\cdot \prob{m(z)=x}\] 
and for each $z\in\mathrm{supp}(Z)$ we have
\[\expect{m(z)} = \sum_{y} \prob{m_1(z) = y}\cdot \expect{m_2(y)} \leq \sum_{y} \prob{m_1(z) = y} \cdot  y = \expect{m_1(z)}\leq z\]
and thus transitivity of the $\preceq$ operator is established.

We will now proceed to the \textbf{main proof}. Fix the random variables $X$ and $Z$ such that $f_X(y)\leq f_Z(y)$ for all $y\in\real$ and order their support sets so that \[\X:=\mathrm{support}(X) = \{x_1, \cdots , x_M\}\] and  \[\Z:=\mathrm{support}(Z) = \{z_1, \cdots ,z_N\}\]
with $x_1 < x_2 < \cdots < x_M$ and $z_1 < z_2 <\cdots < z_N$. We also use $p^X_i:= \prob{X=x_i}$ for all $i\in [N]$ and $p^Z_i:= \prob{Z=Z_i}$ for all $i\in [M]$ to denote the corresponding probabilities, with \[\sum_{i=1}^M p^X_i = \sum_{i=1}^N p^Z_i = 1.\]

We will say that $X$ and $Z$ \textbf{agree up to index $\mathbf{i}$}, if $x_j = z_j$ and $p^X_j=p^Z_j$ for all $j<i$. Conventionally, we say that any two random variables will agree up to index $1$ according to this definition. We will structure our proof of the SDL as an induction on the maximum index that $X$ and $Z$ agree up to. In particular, we break-down our proof in the following two steps:
\begin{enumerate}
    \item (Induction Base). If $X$ and $Z$ agree up to index $M$, the SDL holds.

    \item (Induction Step). If $X$ and $Z$ agree up to index $i\in [M-1]$, there exists a random variable $Z'$ such that
    \begin{enumerate}
        \item $Z'$ and $X$ agree up to index $(i+1)$.
        \item $Z'\preceq Z$.
        \item For all $y\in\real$, $f_X(y) \leq f_{Z'}(y)$.
    \end{enumerate}
\end{enumerate}

Before proving each of these two claims, let's see how they naturally construct an inductive proof for the SDL. Initially, we have that by definition, $X$ and $Z$ agree up to index $1$. We can then apply our second claim (i.e. the induction step) to obtain a random variable $Z_1\preceq Z$ that agrees with $X$ up to index $2$ and satisfies $f_X(y)\leq f_{Z_1}(y)$ for all $y\in\real$. We can then re-apply the induction step to obtain a random variable $Z_2\preceq Z_1$ that agrees with $X$ up to index $3$ and satisfies $f_X(y)\leq f_{Z_2}(y)$ for all $y\in\real$. We keep applying the induction step for as long as we can, until we obtain a random variable $Z_{M-1}\preceq Z_{M-2}\preceq \cdots \preceq Z_1 \preceq Z$ that agrees with $X$ up to index $M$ and satisfies $f_X(y)\leq f_{Z_{M-1}}(y)$ for all $y\in\real$. We then proceed to use our first claim (i.e. the induction base) to show that $X\preceq Z_{M-1}$. Finally, we use the transitivity of operator $\preceq$ in order to obtain $X\preceq Z$, which concludes the proof.

\paragraph{Proof of Induction Base.} Assume that $X$ and $Z$ agree up to index $M$; this means that $x_i=z_i$ and $p^X_i = p^Z_i$ for all $i<M$. Then, consider the deterministic mapping $m(z_i) = x_i$ if $i< M$ and $m(z_i) = x_M$ if $i\geq M$. Since $x_M$ is the last point in the support $\X$, sampling $z\sim Z$ and outputting $m(z)$ is clearly equivalent to sampling $x\sim X$. For the expectation condition, we need to show that $m(z_i) = x_M \leq z_i$ for all $i\geq M$. For $y=x_M$, we have that $f_X(x_M)\leq f_Z(x_M)$. By definition:
\[f_X(x_M) = \sum_{i=1}^{M} p_i^X \cdot x_i = \sum_{i=1}^{M-1} p_i^X \cdot x_i + p_M^X \cdot x_M\]
and
\[f_Z(x_M) = \sum_{i=1}^{M-1} p_i^Z\cdot z_i + \sum_{i=M}^N p_i^Z\cdot \min(z_i, x_M) .\]
Observe that the sums for $i\in [M-1]$ are equal by the agreement assumption and also $\sum_{i=M}^Np_i^Z = p_M^X$. Thus, to satisfy $f_X(x_M)\leq f_Z(x_M)$ we would need $x_M \leq \min(x_M, z_i)$ for all $i\geq M$ or equivalently that $x_M \leq z_M < z_{M+1} <\cdots < z_N$ and the proof follows.

\paragraph{Proof of Induction Step.} We will now prove the induction step. Assume that the random variables $X$ and $Z$ agree up to index $i$ for some $i\in [M-1]$; if they also agree up to index $(i+1)$ the step follows for $Z'=Z$ since clearly $Z\preceq Z$. If they don't agree up to index $(i+1)$, then by $f_X(y)\leq f_Z(y)$ for all $y\in\real$ it must necessarily be the case that either $x_i < z_i$ or ($x_i = z_i$ and $p_i^X > p_i^Z$). In any case, we deduce that $f_X(y) = f_Z(y)$ for all $y\leq x_i$ and $\lim_{y\rightarrow x^+_i}f_X(y) < \lim_{y\rightarrow x^+_i}f_Z(y)$.

Now, let $\ell(y)$ denote the unique straight line that passes through points $(x_i, f_X(x_i))$ and $(x_{i+1}, f_X(x_{i+1}))$, that is:
\[\ell(y) = \frac{f_X(x_{i+1})-f_X(x_{i})}{x_{i+1}-x_i}\cdot (y - x_i) + f_X(x_i). \]
Since $f_Z(y)$ is clearly concave, $\ell(y)$ will intersect it in at most two points; one of them is the point $(x_i, f_Z(x_i)) = (x_i, f_X(x_i))$ and the other point will necessarily be $(s,f_Z(s))$ for some $s\geq x_{i+1}$. A pictorial representation is given in~\Cref{fig:stochastic-dominance-induction-step}.

\begin{figure}[h!]
    \centering
    \includegraphics[width=0.65\textwidth]{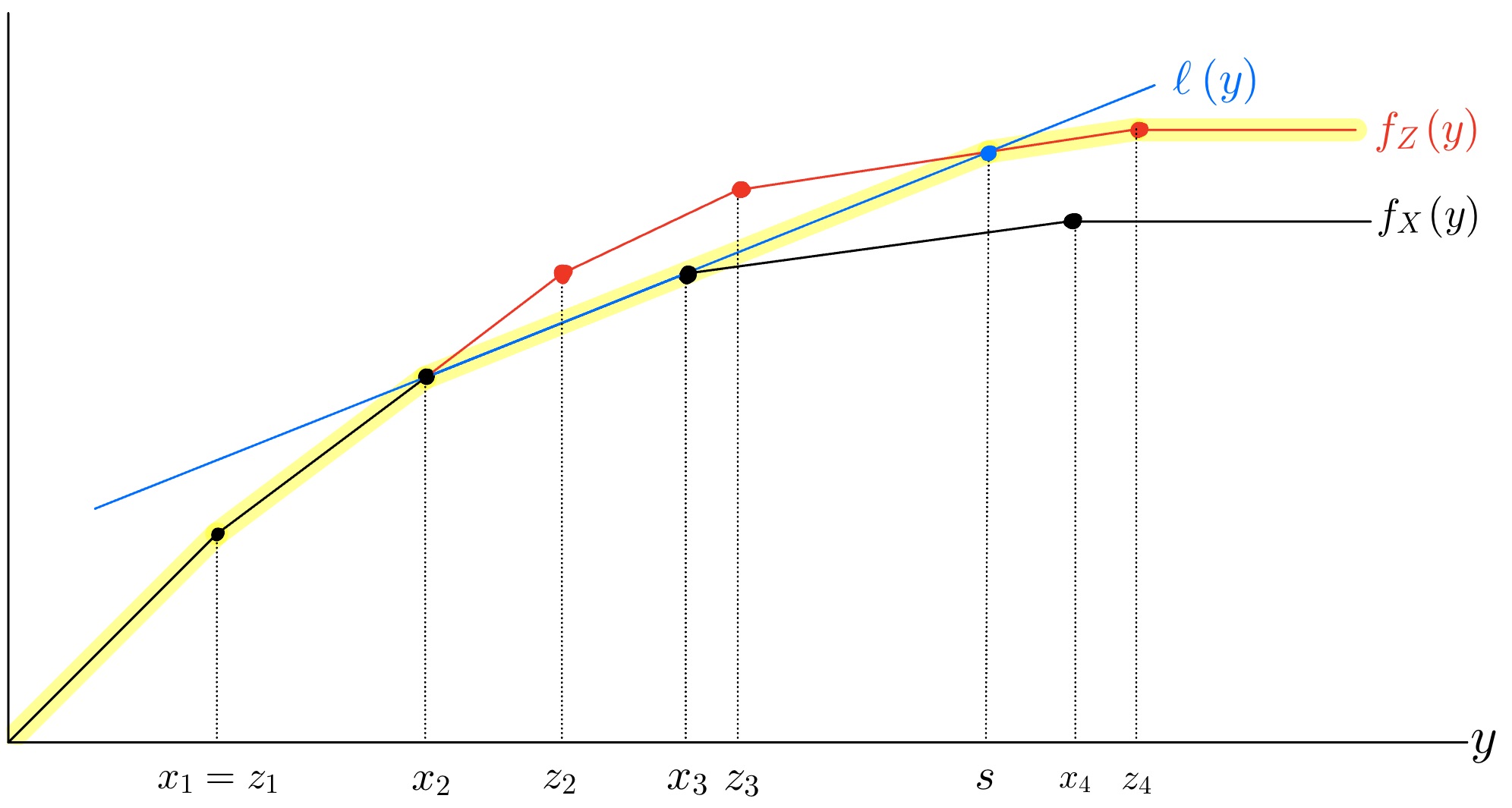}
    \caption{The induction step. Random variables $X$ and $Z$ agree up to index $2$. The line $\ell(y)$ extends the line-segment of $f_{X}(y)$ from $x_2$ to $x_3$ and cuts $f_Z(y)$ on some $y=s\geq x_3$. Let $h(y)$ denote the highlighted curve. The random variable $Z'$ that has curve $f_{Z'}(y) = h(y)$ agrees with $X$ up to index $3$ and satisfies $f_{Z'}(y)\geq  f_{X}(y)$ for all $y\in\real$. Here, $J=\{2,3\}$ denotes the set of indices of $f_Z$'s break-points that lie in $(x_2,s)$.}
    \label{fig:stochastic-dominance-induction-step}
\end{figure}

Now, consider the curve $h(y) = \min\{\ell(y), f_Z(y)\}$. By construction, this curve satisfies $f_X(y) \leq h(y)$ for all $y\in\real$. Furthermore, if it is the case that there exists some random variable $Z'$ such that $f_Z'(y) = h(y)$, then it would be the case that $Z'$ and $X$ agree up to index $(i+1)$. Thus, to complete the proof, we need to show that such a random variable $Z'$ not only exists, but also satisfies $Z'\preceq Z$. This will require it's own type of induction, so we will state it an prove it as a separate claim (\Cref{claim:sdom-helper}) to ease the presentation. Notice that once this claim is proven, the proof of the SDL is completed.
\end{proof}

\begin{claim}\label{claim:sdom-helper}
    For any discrete random variable $Z$ with curve $f_Z(y) = \expect{\min\{y,Z\}}$ and any line $\ell(y)$ intersecting $f_Z(y)$ at exactly two points $a<b$, there always exists discrete random variable $Z'$ such that $Z'\preceq Z$ and $f_{Z'}(y) = \min\{f_Z(y), \ell(y)\}$.
\end{claim}

\begin{proof}
Once again, we denote $\Z:=\mathrm{support}(Z) = \{z_1, \cdots ,z_N\}$ and assume $z_1 < z_2 <\cdots < z_N$. We also use $p_i^Z := \prob{Z=z_i}$ for $i\in [N]$. Notice that the function $f_Z(y)$ is piece-wise linear, with breakpoints precisely at $z_i\in\Z$. Furthermore, the slope at the interval $(-\infty , z_i]$ is $1$, the slope at any interval $[z_i, z_{i+1}]$ for $i\in [N-1]$ is $1-\sum_{j\leq i}p^Z_j $ and the slope at the interval $[z_N,\infty)$ is $0$. Finally, for each point $z_i\in\Z$, the difference of the slope of $f_Z(y)$ on the segment to its left minus the slope of $f_Z(y)$ on the segment to its right equals the probability $p_i^Z$.

We will begin by designing a useful \textbf{gadget}. This gadget $G(Z,a,b)$ takes as input a discrete random variable $Z$ and two parameters $a<b$ such that there exists index $i\in [N]$ with $z_{i-1}\leq a < z_i < b \leq z_{i+1}$ (we denote $z_0=-\infty$ and $z_{N+1}=+\infty$) and returns a random variable $Z'$ that is obtained by mapping each point $z_j\in \Z$ with $j\neq i$ to itself, and mapping point $z_i$ to point $a$ with some probability $\lambda$ and to point $b$ with probability $1-\lambda$. Clearly, $\mathrm{support}(Z')= \{a,b\}\cup \Z\setminus\{z_i\}$. Let $\ell_{ab}(y)$ be the line passing through points $(a,f_Z(a))$ and $(b,f_Z(b))$ and let $s$ be the slope of this line. Furthermore, let $s_1$ be the slope of $f_Z(y)$ at the interval $[a,z_i]$ and $s_2$ be its slope at the interval $[z_i, b]$; by definition, $s_1-s_2 = p^Z_i$. Also, note that $s_1>s>s_2$ by concavity. Then, by using mapping probability
\[\lambda:= \frac{s_1-s}{s_1-s_2}\]
it is not hard to see that $f_{Z'}(y) = \min\{f_Z(y),\ell_{ab}(y)\}$, since we maintain the probability mass at all points $z_j\neq z_i$ and furthermore we have $\prob{Z'=a} = p_i^Z\cdot \lambda = s_1-s$ and $\prob{Z'=b} = p_i^Z \cdot (1-\lambda) = s-s_2$. Furthermore, it is also not hard to see that $Z'\preceq Z$; we only need to verify that
\[\lambda\cdot a + (1-\lambda)\cdot b \leq z_i\]
or equivalently (by substituting $\lambda$'s definition) that $s(b-a) \leq s_1(z_i-a) + s_2(b-z_i)$; this always holds with equality due to the definition of $s_1,s_2$ and $s$. 

We are now ready to prove the claim. Let $J=\{z_i\in (a,b)\}$; this is the set of points $z_i\in\Z$ for which $f_Z(z_i)> \ell(z_i)$. Notice that the gadget $G(Z,a,b)$ already proves the claim for the special case of $|J|=1$. We will now inductively apply the gadget to prove the general case. Let $z_i$ be the minimum point in $J$; we begin by applying the gadget $G(z_i , a , z_{i+1})$ to obtain a new random variable $Z'$; this is allowed since $z_i$ is the unique point of $\Z$ in the interval $(a,z_{i+1})$. By our construction of the gadget and concavity of $f_Z(y)$, we have that $Z'\preceq Z$ and also that $f_{Z'}(y) \geq \min\{f_Z(y),\ell(y)\}$ for all $y\in\real$. Furthermore, the corresponding $J$-interval for $Z'$ will now have one less point; thus, we can inductively keep applying our gadget and by transitivity of the $\preceq$ operator the claim follows. A pictorial proof is shown in~\Cref{fig:sdom_induction}

\begin{figure}[h!]
    \centering
    \includegraphics[width=0.9\textwidth]{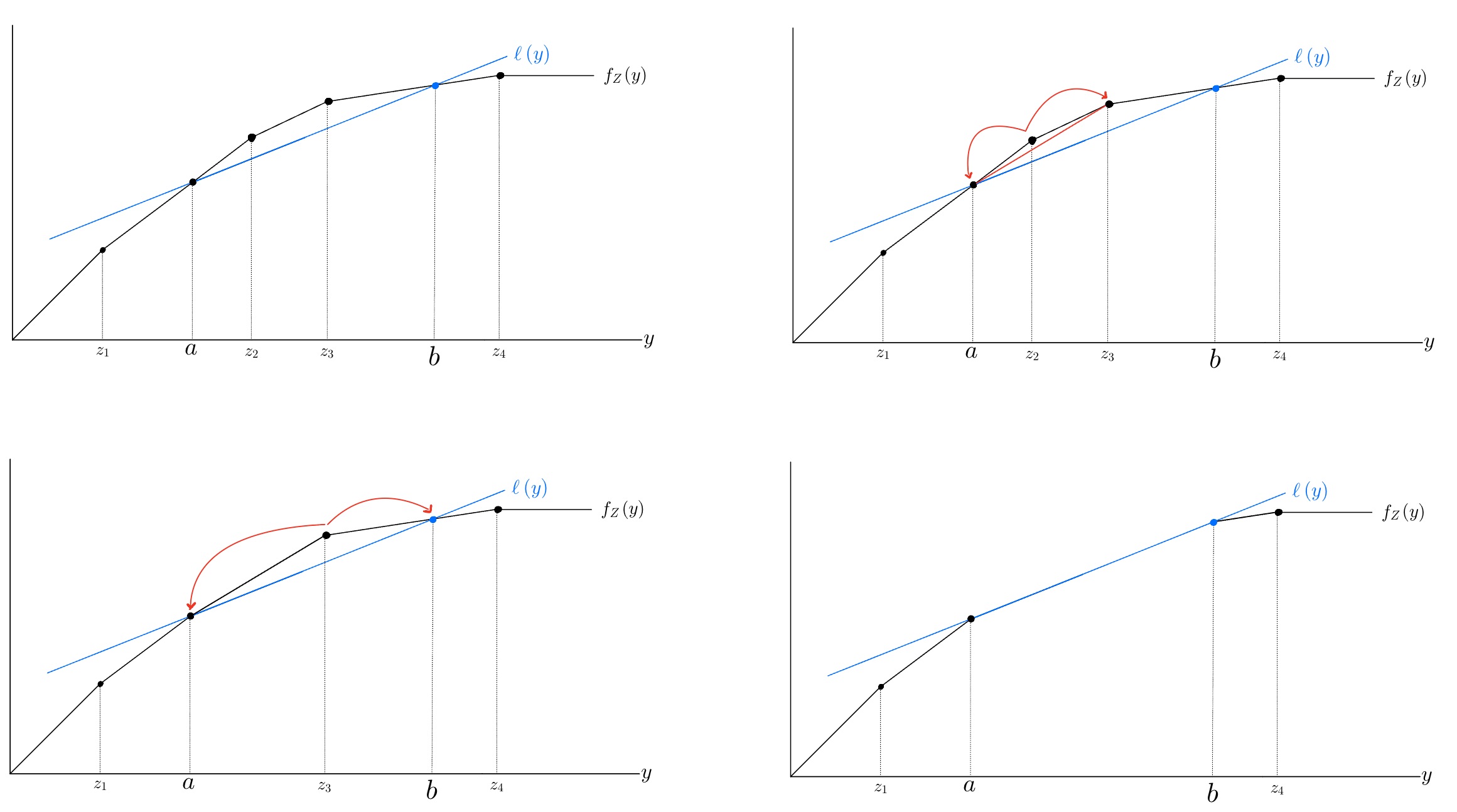}
    \caption{Here, $J=\{z_2,z_3\}$. At each step, we isolate the leftmost point of $f_Z(y)$ that is dominated by $\ell(y)$ and use our gadget to distribute it between $a$ and the next point. Eventually, we recover a random variable with optimality curve $\min\{f_Z(y), \ell(y)\}$.}
    \label{fig:sdom_induction}
\end{figure}

\end{proof}

\newpage
\newpage
\section{Omitted Proofs from~\Cref{sec:local_approx} (Local Approximation)}\label{app:local_approx}

In this section, we present the omitted proof of~\Cref{thm:la-comp} from~\Cref{sec:local_approx}.
\loccomptag*
\begin{proof}
Let $z := [W^*_{\mc_2^\pi}, \cdots W^*_{\mc_n^\pi}]$ encode the surrogate costs of all Markov chains $\mc_i^{\pi_i}$ with the exception of $\mc_1^{\pi_1}$. Then, we have

\begin{align*}
    \expect{\min_{S\in \feas} \sum_{i\in S} W^*_{\mc_i^{\pi_i}}} &= \expectt{z}{\expectt{W^*_{\mc_1^{\pi_1}}}{\min_{S\in \feas} \sum_{i\in S} W^*_{\mc_i^{\pi_i}}}}\\
    &= \expectt{z}{\expectt{W^*_{\mc_1^{\pi_1}}}{\min_{S\in \feas}\bigg( W^*_{\mc_1^{\pi_1}}\cdot \mathbbm{1}[1\in S] + \sum_{i\in S\setminus\{1\}} W^*_{\mc_i^{\pi_i}}\bigg)}} \\
    &= \expectt{z}{\expectt{W^*_{\mc_1^{\pi_1}}}{ \min\bigg(
    \min_{S\in\feas : 1\notin S} \sum_{i\in S}W^*_{\mc_i^{\pi_i}}\;\;
    ,\;\;
    W^*_{\mc_1^{\pi_1}} + \min_{S: 1\notin S ; S\cup\{1\}\in\feas} \sum_{i\in S} W^*_{\mc_i^{\pi_i}} \bigg)
    }} \\
\end{align*}

Now, let $f_1(z):= \min_{S\in\feas : 1\notin S} \sum_{i\in S}W^*_{\mc_i^{\pi_i}}$ and $f_2(z):=\min_{S: 1\notin S ; S\cup\{1\}\in\feas} \sum_{i\in S} W^*_{\mc_i^{\pi_i}}$. Note that both quantities depend only on $z$ and not on $W^*_{\mc_1^{\pi_1}}$. Then, we have

\begin{align*}
    \expect{\min_{S\in \feas} \sum_{i\in S} W^*_{\mc_i^{\pi_i}}} &=  \expectt{z}{\expectt{W^*_{\mc_1^{\pi_1}}}{ \min\bigg(
    f_1(z)\;\;
    ,\;\;
    W^*_{\mc_1^{\pi_1}} + f_2(z)\bigg)
    }} \\
    &= \expectt{z}{f_2(z) + \expectt{W^*_{\mc_1^{\pi_1}}}{ \min\bigg(
    f_1(z)-f_2(z)\;\;
    ,\;\;
    W^*_{\mc_1^{\pi_1}}\bigg)
    }} \\
    &\leq \expectt{z}{f_2(z) + \expectt{W^*_{\mc_1}}{ \min\bigg(
    f_1(z)-f_2(z)\;\;
    ,\;\;
    \alpha\cdot W^*_{\mc_1}\bigg)
    }} && \text{(Local Approximation)}\\
    &= \expectt{z}{\expectt{W^*_{\mc_1}}{ \min\bigg(
    f_1(z)\;\;
    ,\;\;
    \alpha\cdot W^*_{\mc_1} + f_2(z)\bigg)
    }} \\
    &= \expect{\min_{S\in \feas} \sum_{i\in S} \Tilde{W}_{i}}
\end{align*}
where $\Tilde{W}_1 = \alpha\cdot W^*_{\mc_1}$ and $\Tilde{W}_i := W^*_{\mc_i^{\pi_i}}$ for all $i\neq 1$. Thus, we have substituted $W^*_{\mc_1^{\pi_1}}$ with $\alpha\cdot W^*_{\mc_1}$. By repeating the same process for $i=2,3,\cdots, n$, we obtain that
\[\expect{\min_{S\in \feas} \sum_{i\in S} W^*_{\mc_i^{\pi_i}}} \leq \expect{\min_{S\in \feas} \sum_{i\in S} \alpha\cdot W^*_{\mc_i}}\]
and thus
\[\frac{\expect{\min_{S\in \feas} \sum_{i\in S} W^*_{\mc_i^{\pi_i}}}}{\expect{\min_{S\in \feas} \sum_{i\in S} W^*_{\mc_i}}} \leq \alpha .\]
From~\Cref{thm:MC-opt} the nominator is precisely the cost of the optimal committing policy for $\instance$ under commitments $\comms=(\pi_1, \cdots, \pi_n)$ and from~\Cref{thm:MDP-lb}, the denominator is a lower bound on $\opt(\instance)$. Thus, we have $\cg(\instance)\leq \alpha$ as desired.

\end{proof}

\newpage
\section{Omitted Proofs from~\Cref{sec:pvo} (Partial Inspection PB)}\label{app:pvo}

In this chapter of the appendix, we present all omitted proofs from~\Cref{sec:pvo}. We restate all our results for the reader's convenience.

\subsection{Proof of~\Cref{claim:pvo-global-approx}}

\lemmaphiapprox*
\begin{proof}
    The proof relies on the crucial fact that opening and peeking into a PI box reveals precisely the same information to the decision maker; in particular, the value of the box. We construct a policy that commits to directly opening boxes $i\in O$ and peeking before opening boxes $i\in P$, while simultaneously mimicking the optimal adaptive policy. In particular, fix any box $i$. 
    
    \begin{itemize}
        \item If the optimal never interacts with box $i$, neither does our policy.
        
        \item If the optimal directly opens box $i$, then (i) if $i\in O$ our policy also opens it and never peeks into it, and (ii) if $i\in P$ our policy first peeks into the box and then immediately opens it.

        \item If the optimal first peeks into box $i$, then (i) if $i\in P$ our policy also peeks into it (and then opens it whenever the optimal decides to open it), and (ii) if $i\in O$ our policy directly opens it and never peeks into it.

        \item If the optimal selects box $i$, so does our policy.

    \end{itemize}

    Observe that at any point the optimal and our policy have the exact same information, so we can keep mimicking the optimal decision tree. Furthermore, our policy clearly respects the given commitment. Finally, whenever the optimal selects a box, we can do the same as the set of our policy's opened boxes is always a superset of the optimal's opened boxes. This ensures feasibility of our algorithm under any combinatorial constraint $\feas$.

    Thus, the only difference between the costs of our policy and the optimal is due to differences in the selected actions. In particular, if $i\in P$ then if the optimal peeks into or ignores the box then so does the algorithm, with the worst case being the optimal directly opening the box. In that case, the optimal pays $c_i^o$ whereas our policy pays $c_i^p+c_i^o$. On the other hand, if $i\in O$, then the worst case is if the optimal peeks into the box and decides not to open it; in that case the optimal pays $c_i^p$ whereas our algorithm pays $c_i^o$. Finally, our algorithm pays precisely the same cost as the optimal for accepting boxes.

    Combining everything, we conclude that our policy achieves an
    \[\alpha:= \max \bigg( \max_{i\in O} (\frac{c_i^o}{c_i^p}) \; , \; \max_{i\in P} (\frac{c_i^p + c_i^o}{c_i^o})\bigg)\]
    approximation to the optimal adaptive policy. For each box $i\in [n]$, let $\lambda_i = c^p_i/c^o_i \in (0,1)$ and observe that by definition of the partition sets we have that $1/\lambda_i \leq 1 + \lambda_i$ if and only if $i\in O$. Thus, we obtain that
    \[\alpha \leq \max_i \min(1+ \lambda_i, \frac{1}{\lambda_i}) \leq \max_{x\in (0,1)} \min(1+x, \frac{1}{x}) = \phi.\]
\end{proof}
\medskip

\subsection{Proof of~\Cref{lemma:pvo-local-approx}} 
Recall that we have already defined the opening index $g^o$ and peeking index $g^p$ of a PI $\pbox = (\dist , c^o, c^p)$ as the solutions to equations $c^o = \expectt{X\sim\dist}{(g^o-X)^+}$ and $c^p = \expectt{X\sim\dist}{(g^p-X-c^o)^+}$ respectively. Observe that if the optimal adaptive policy prefers opening over peeking the box for some outside option $y$, then the same will be true for any $y'>y$. Also, for $y=0$ peeking is clearly preferable to opening. Thus, the optimality curves $f^o_\pbox(y)$ and $f^p_\pbox(y)$ will have a unique intersection (up to an interval). We use $\tau$ to denote this intersection; formally, $\tau$ is the maximal solution to equation \[\expectt{X\sim\dist}{(\tau - X)^+} - \expectt{X\sim\dist}{(\tau - X-c^o)^+} = c^o - c^p.\]
This implies that unless $g^p < g^o < \tau$, the opening action will dominate the peeking action for all values of the outside option $y$ for which accepting it isn't optimal; as a consequence, such instances trivially admit a $1$-local approximation by committing to the opening action. From now on, we assume that $g^p < g^o < \tau$ and prove the two missing claims from~\Cref{sec:pvo} that complete the proof of~\Cref{lemma:pvo-local-approx}.

\pvoclone*
\begin{proof}
    By concavity, we have that $f^o_\pbox(y) \leq \alpha f^o_\pbox(\frac{y}{\alpha})$ for all $y\in \real$ and $\alpha\geq$ 1; thus, we only need to verify the condition for $y$ such that $f_{\pbox}(\frac{y}{\alpha}) = f^p_\pbox (\frac{y}{\alpha})$. In other words, we only need to verify
\[\min\{y, f^o_\pbox(y)\} \leq \alpha f^p_\pbox(\frac{y}{\alpha})\;\;\;\;\forall y\in [\alpha g^p, \alpha\tau].\]

For start, we will show that the condition holds for $\alpha = g^o/g^p$; notice that for this parameter, $f^o_\pbox(y)\leq y$ for all $y\in [\alpha g^p, \alpha\tau]$ and thus we can now write our condition as
   \[D_\alpha(y) :=  \alpha f^p_\pbox(\frac{y}{\alpha})-f^o_\pbox(y)\geq 0\;\;\;\;\forall y\in [\alpha g^p, \alpha\tau].\] 
Let $F(\cdot)$ denote the CDF of distribution $\dist$. By definition of $f^o_\pbox(y)$ and $f^p_\pbox(y)$, we have
\[\dd{y}{f^o_\pbox(y)} = 1 - F(y)\;\; , \;\; \dd{y}{f^p_\pbox(y)} = 1 - F(y-c^o) \]
and by taking the derivative of $D_\alpha(\cdot)$, we immediately obtain that $D_\alpha(\cdot)$ is a weakly increasing function of $y$ and thus the condition only needs to hold on $y=\alpha g^p$. Observe that for $\alpha = g^o/g^p$ and $y=\alpha g^p = g^o$, we have
\[D_{\alpha}(y) = \frac{g^o}{g^p}\cdot f^p_\pbox(g^p) - f^o_\pbox(g^o)  = g^o - g^o = 0\]
and thus we obtain that the opening action always achieves a $g^o/g^p$ local approximation.

The final step of the proof will be to show that
\[\frac{g^o}{g^p} \leq \frac{c^o}{c^p}\cdot \big(1-\frac{c^o}{g^p}\big)\]
For this purpose, we define the function $h(z):= \expectt{X\sim\dist}{(z-X)^+}$ and observe that $h'(z) = F(z)$; thus, $h(z)$ is an increasing and convex function of $z$ with $h(0)=0$. Furthermore, by definition of the indices and the curves we have that $h(g^o)=c^o$ and $h(g^p-c^o)=c^p$. Thus, by convexity, we immediately obtain that
\[\frac{h(g^o)}{g^o}\geq \frac{h(g^p-c^o)}{g^p-c^o} \Rightarrow \frac{c^o}{g^o}\geq \frac{c^p}{g^p-c^o}\]
from which the claim follows.

\end{proof}
\smallskip

\pvocltwo*
\begin{proof}
    By concavity, we have that $f^p_\pbox(y) \leq \alpha f^p_\pbox(\frac{y}{\alpha})$ for all $y\in \real$ and $\alpha\geq$ 1; thus, we only need to verify the condition for $y$ such that $f_{\pbox}(\frac{y}{\alpha}) = f^o_\pbox (\frac{y}{\alpha})$; notice that this corresponds to $y\geq \alpha\tau > g^p$ and thus $f^p_\pbox(y) \leq y$. In other words, we only need to verify
\[f^p_\pbox(y) \leq \alpha f^o_\pbox(\frac{y}{\alpha})\;\;\;\;\forall y\geq \alpha\tau.\]
   
Let $D_\alpha(y):=\alpha \cdot f^o_\pbox(y/\alpha)-f^p_\pbox(y)$. We need $D_\alpha(y)\geq 0$ for $y\geq \alpha\tau$. The condition immediately holds at $y=\alpha\tau$ since 
\[\alpha f^o_\pbox(\frac{\alpha\tau}{\alpha}) = \alpha f^o_\pbox(\tau) = \alpha f^p_\pbox(\tau)\geq f^p_\pbox(\alpha\tau)\]
by concavity of $f^p_\pbox(\cdot)$. Next, observe that

\[\dd{y}{D_\alpha(y)} = \dd{y}{f^o_\pbox(\frac{y}{\alpha})} - \dd{y}{f^p_\pbox(y)} = F(y-c^o) - F(\frac{y}{\alpha})\]
and thus $D_a(y)$ gets minimized at $y=\frac{\alpha\cdot c^o}{\alpha-1}$. If $\frac{\alpha\cdot c^o}{\alpha -1} \leq \alpha\tau$ or equivalently $\alpha \geq 1 + c^o/\tau$, then $\dd{y}{D_\alpha(y)}\geq 0$ in the area of interest; thus, we obtain that the peeking action always achieves a $1+\frac{c^o}{\tau} \leq 1+\frac{c^o}{g^p}$ local approximation. To complete the proof, we need to show that the peeking action also achieves a $1+\frac{c^p}{c^o}$ local approximation or equivalently that for $\alpha = 1 + \frac{c^p}{c^o}$ we have that $\min_{y\geq \alpha\tau}D_\alpha(y) = D(\frac{\alpha\cdot c^o}{\alpha-1}) \geq 0$.

Once again, we consider the function $h(z) = \expectt{X\sim\dist}{(z-X)^+}$; this time, we observe that by definition we have
\[f^o_\pbox(y) = c^o + y - h(y) \;\; , \;\; f^p_\pbox(y) = c^p + y - h(y-c^o).\] 
From this, the condition $D(\frac{\alpha\cdot c^o}{\alpha-1}) \geq 0$ translates to
\[h(\frac{c^o}{\alpha-1}) \leq \frac{\alpha\cdot c^o-c^p}{\alpha-1}\]
and for $\alpha = 1+ \frac{c^p}{c^o}$, this statement is equivalent to 
\[h(\frac{c^o\cdot c^o}{c^p})\leq \frac{c^o\cdot c^o}{c^p}\]
which is clearly true, by definition of $h(\cdot)$.

\end{proof}

\newpage
\section{Omitted Proofs from~\Cref{sec:additive_pb} (Additive PB)}\label{app:additive_pb}
In this chapter of the appendix, we present all omitted proofs from~\Cref{sec:additive_pb}. We restate all our results for the reader's convenience.

\subsection{Proof of~\Cref{thm:sum-mdp-static}}
\additivestatic*
\noindent We will separately prove the theorem for $k=2$ and $k\geq 3$.

\paragraph{The $k=2$ case.} We use $(X,c_x)$ and $(Z,c_z)$ to denote the $k=2$ components of the additive-MDP and $g_x,g_z$ to denote their respective water filling indices. Notice that there are only two (static) committing policies to consider: $\pi_{XZ}$ and $\pi_{ZX}$ depending on whether $X$ is probed first or not. We use $g_{xz}$ and $g_{zx}$ for their respective indices. Using the same structural properties of surrogate costs for additive-MDPes as in the proof of~\Cref{thm:sum-mdp-gap}, we have that $\pi_{XZ}$ admits an $\alpha_{xz}$-pointwise approximation for
\[\alpha_{xz} = \max_{t=(x,z)} \frac{\surr_{xz}(t)}{\surr_{zx}(t)} = \max_{t=(x,z)}\frac{\max(g_{xz}, g_x + z , z+x)}{\max( g_{zx}, g_z + x, z+x)} \leq \max_{t=(x,z)}\max (1, \frac{g_{xz}}{g_{zx}} , \frac{g_x + z}{\max( g_{zx}, g_z + x, z+x)})\]
and for the last term in the maximum, since $x\geq 0$, we have that
\[\alpha_{xz}  \leq \max_{t=(x,z)}\max (1, \frac{g_{xz}}{g_{zx}} , \frac{g_x + z}{\max( g_{zx},z)}) \leq \max ( \frac{g_{xz}}{g_{zx}} , 1 + \frac{g_x}{ g_{zx}})\]
and likewise we obtain that 
\[\alpha_{zx}   \leq \max ( \frac{g_{zx}}{g_{xz}} , 1 + \frac{g_z}{ g_{xz}}).\]

Finally, we will now prove that $\min(\alpha_{zx},\alpha_{zx})\leq \phi$ always. By re-labeling and re-scaling we can assume without loss that $g_{xz}\leq g_{zx}=1$. Furthermore, we have that $g_{xz} \geq g_x + g_z$. To see why this is true, consider an amortization $\mc^{\pi_{xz}}$ where each (height $1$) $Z$-vertex sends the same cost shares $b_z$ as in the water filling amortization of $(Z,c_z)$ to any terminal $(x,z)$, and the (unique) $X$-vertex sends the same cost shares $b_x$ as in the water filling amortization of $(X,c_x)$ to any terminal $(x,z)$ - this amortization covers the opening costs by definition, and the surrogate cost of terminal $(x,z)$ is $x+z + b_x + b_z = \surr_x(x) + \surr_z(z) \geq g_x + g_z$. Since water-filling amortization maximizes the minimum surrogate cost across all amortizations, we indeed have $g_{xz}\geq g_x+g_z$. Thus, we have $g_x + g_z \leq g_{xz} \leq g_{zx} = 1$. These in turn imply that $\alpha_{xz} \leq 1 + g_x$ and $\alpha_{zx} \leq \frac{1}{g_x + g_z}\max(1, g_x + 2g_z)$. The proof is then completed by showing
\[\alpha:= \min\bigg(1 + g_x, \frac{1}{g_x + g_z}\max(1, g_x + 2g_z)\bigg)\leq \phi\]
for all $g_x,g_z\geq 0$ with $g_x+g_z\leq 1$. By fixing $g_x$, it is not hard to see that the second term gets maximized at either $g_z=0$ or $g_z = 1-g_x$. Thus, we have
\[\alpha \leq \max_{x\in [0,1]}\min\bigg( x+1 , \max(\frac{1}{x}, 2-x)\bigg) = \phi.\]

\smallskip
\paragraph{The $k\geq 3$ case.} Let $\pi\in\comset{\mc}$ be any commitment of index $g^\pi$ for an additive-MDP $\mc$ over $k$-components and let $\alpha = \lfloor \frac{k+1}{2}\rfloor$. We will show that there exists a static committing policy $\pi'$ such that $g^{\pi'}\leq \alpha\cdot g^\pi$. Then, by instantiating $\pi$ as any minimum index policy and observing that all static policies are by definition root dominant (i.e. the root state has maximum index across all states), step 1 in the proof of~\Cref{thm:sum-mdp-gap} immediately gives us the proof of~\Cref{thm:sum-mdp-static}.

Consider the Markov chain $\mc^\pi$. Like in the proof of~\Cref{thm:sum-mdp-gap}, we define the index $g^\pi(s)$ of a state $s$ as the water filling index of the sub-chain rooted at $s$ if all previously probed values are realized at $0$, and its value $V^\pi(s)$ as the sum of all realized values in the trajectory from the root to $s$. Observe that these indices and values uniquely determine the water filling cost shares that each state $s$ sends to its downwards terminal states in $\mc^\pi$; in particular, we have that
\begin{align*}
    b_{st} &= \max\bigg( 0, g^\pi(s) + V^\pi(s) - \max_{s'\in P(s,t)}(g^\pi(s') + V^\pi(s'))\bigg) \\
    &= \max\bigg( 0, g^\pi(s)  - \max_{s'\in P(s,t)}(g^\pi(s') + V^\pi(s,s'))\bigg)
\end{align*}
where the inner maximum is taken over all states in the (unique) path from $s$ to $t$, excluding state $s$, and $V^\pi(s,s')$ denotes the sum of all realizations in the trajectory from $s$ to $s'$. The above equality is a direct consequence of the fact that the minimum surrogate cost in the sub-tree of state $s$ immediately after $s$ gets amortized is by definition $g^\pi(s) + V^\pi(s)$.

We will now define the static policy $\pi'$. Observe that any trajectory from the root of $\mc^\pi$ to a terminal state defines a fixed ordering of the $k$-components; we define $\pi'$ as the ordering corresponding to a trajectory of lexicographically maximum index. In other words, $\pi'$ probes the same first box as $\pi$, then the second box corresponding to a child of the root with maximum index etc. By re-labeling, we assume that $\pi'$ is the static ordering $k\mapsto (k-1)\mapsto \cdots \mapsto 1$. 

We now turn our attention to the Markov chain $\mc^{\pi'}$. We use $h(s)$ to denote the height of state $s$; this is $0$ for terminal states and $k$ for the root. Notice that in $\mc^{\pi'}$, all the states $s$ with $h(s)=j\in [k]$ correspond to the same (remaining) fixed probing order $j\mapsto (j-1)\mapsto \cdots \mapsto 1$, and will have the same index $g^{\pi'}(s) = g(j)$, with $g^{\pi'}=g(k)$. Furthermore, each state $s$ of $\mc^{\pi'}$ at height $j\in [k]$ can be mapped to the unique height $j$ state of $\mc^\pi$ that was chosen by our process; we use $s_j$ to denote this state.

We will prove that for all $j\in [k]$, it holds that $g(j)\leq \alpha(j)\cdot g^\pi(s_j)$ for $\alpha(j):=\lfloor \frac{j+1}{2}\rfloor$. This would directly imply that $g^\pi = g(k) \leq \alpha(k)\cdot g^\pi(s_k) = \alpha(k)\cdot g^\pi$, completing the proof. In order to upper bound the indices $g(j)$ of $\mc^{\pi'}$, it suffices to argue that the suggested upper bounds generate enough cost shares to cover the amortization of the costs. Since each state $s$ of $\mc^{\pi'}$ at height $j\in [k]$ shares the same distribution over terminal states with the state $s_j$ of $\mc^\pi$, it will suffice to point-wise compare these cost shares. In other words, it suffices to prove that 
\[\max\bigg( 0, \alpha(j)\cdot g^\pi(s_j)  - \max_{s'\in P'(s,t)}(\alpha(h(s'))\cdot g^\pi(s_{h(s')}) + V^{\pi'}(s,s'))\bigg) \geq 
\max\bigg( 0, g^\pi(s_j)  - \max_{s'\in P(s_j,t)}(g^\pi(s') + V^\pi(s_j,s'))\bigg)\]
for all $j\in [k]$. Our proof will be completed by arguing that our specific definition of $\alpha(j)$ satisfies
\[\max_{s'\in P'(s,t)}\bigg( \alpha(h(s'))\cdot g^\pi(s_{h(s')}) + V^{\pi'}(s,s')\bigg) \leq \alpha(j)\cdot \max_{s'\in P(s_j,t)}\bigg( g^\pi(s') + V^\pi(s_j,s')\bigg)\]
for all $j\in [k]$, all states $s$ with $h(s)=j$ and all terminal states $t$.

The maximum in the left hand side will get realized at some $s'\in P'(s,t)$ with $h(s') := j'\leq j-1$ (as $h(s)=j$). If $j' < j-1$, then we can upper bound the left hand side by $a(j')\cdot g^\pi(s_{j'}) + V^{\pi'}(s,t)$. Since $\pi'$ is a static policy, it is not hard to see that the indices in any trajectory from the root to a terminal will form a decreasing sequence, so $g^\pi(s_{j'}) \leq g^\pi(s_j)$. Furthermore, we have that $V^{\pi'}(s,t) = V^\pi (s_j,t)$ and thus our inequality holds as long as $\alpha(j)\geq \alpha(j')+1$ which is always the case for $j'< j-1$.

It remains to argue about $j'=j-1$. In that case, the left hand side is $\alpha(j-1)\cdot g^\pi (s_{j-1}) + V_j$ for some realization $V_j$ of the $j$-th component. We will upper bound this through the first term in the right hand side maximum, corresponding to some child vertex $\hat{s}$ of $s_j$; in that case, we need to show that
\[\alpha(j-1)\cdot g^\pi (s_{j-1}) + V_j \leq \alpha(j)\cdot g^\pi (\hat{s})+ V_j \]
and the proof is completed by the fact that our policy selected a trajectory of lexicographically maximum indices and thus $g^\pi (s_{j-1}) \leq g^\pi (s')$ as $s_{j-1}$ and $s'$ are both children of the same state $s_j$ in $\mc^\pi$.

\newpage
\newpage
\section{Omitted Proofs from~\Cref{sec:ws} (Weighing Scale Problem)}\label{app:ws}
In this chapter of the appendix, we present all omitted proofs from~\Cref{sec:ws}. We restate all our results for the reader's convenience.

\subsection{Proof of~\Cref{lem:osh-policy}}

\wsupptag*
\noindent We fix a WS alternative $\mc$ corresponding to random value $X$ of support $\X:=\mathrm{support}(X)$ and weighing cost $c$. We use $\mu , M , g, h$ and $\kappa$ to denote the parameters of the alternative, as previously defined. In order to develop our pointwise approximation guarantees, we first need to obtain expressions for the surrogate costs.

\paragraph{Surrogate cost of $\mc$.} We begin by noting that the optimality curve of $\mc$ has a very simple form. Indeed, consider the local game $(\mc , y)$; the optimal adaptive policy has only three choices available: either accept the outside option at a cost of $y$, or accept the alternative without performing any weighings at an expected cost of $\mu$ or perform a single weighing of the alternative against $y$ to determine which of the two costs is smaller. Thus, we immediately obtain that $\optf_\mc(y) = \min(y, \mu , c + \expect{\min\{y,X\}})$.

Observe that by definition, $g$ corresponds to the maximal threshold for which $y\leq c + \expect{\min\{y,X\}}$ for all $y\leq g$. Thus, if $\mu\leq g$, then the cost $c + \expect{\min\{y,X\}}$ will always be dominated by either $y$ or $\mu$, making the weighing action universally sub-optimal. If that's the case, then we can safely commit to blindly accepting the alternative without performing any weighings, achieving a $1$-local approximation. Thus, from now on we will be assuming that $g<\mu$. Recall that $g$ satisfies $c=\expect{(g-X)^+}$; $h$ satisfies $c=\expect{(X-h)^+}$; and $\mu$ satisfies $\expect{(\mu-X)^+}=\expect{(X-\mu)^+}$ by its definition. Then we can deduce that $g<\mu$ implies $\mu<h$.
We can therefore re-write the optimality curve of $\mc$ as

\begin{equation*}
 \optf_\mc (y) =
    \begin{cases}
    y & \text{if $y < g$}\\
    \expect{\min(y,\max (g,X))} & \text{if $y\in [g,h]$}\\
    \mu & \text{if $y> h$}
    \end{cases}       
\end{equation*}

From this, we obtain the following characterization of the surrogate costs. We note that the same result is proven by~\cite{DS24}, as the optimality curve for the alternative $\mc$ coincides with the optimality curve of a Pandora's Box with optional inspection in the minimization setting.
\begin{fact}\label{cl:ws-opt-sur}
    The water filling surrogate cost $W^*_{\mc}$ of $\mc$ corresponds to sampling $x\sim X$ and returning
    \[\surr^*(x):= \min(h , \max(g,x)).\]
\end{fact}


\paragraph{Committing policies.} Now, consider any commitment $\pi\in\comset{\mc}$. Observe that $\pi$ corresponds to a protocol that determines in advance a decision tree (or a distribution over decision trees) over pre-specified thresholds against which it will weigh $X$, resulting in a Markov chain $\mc^\pi$. Importantly, any commitment will end up partitioning the support $\X$ of distribution $\dist$ into a set of intervals, corresponding to its terminal states. We use $\I^\pi$ to denote this set of intervals, and $t(I)$ to denote the terminal state of $\mc^\pi$ corresponding to interval $I\in \I^\pi$. Furthermore, the probability of running the Markov chain $\mc^\pi$ and ending up in a terminal state $t(I)$ will be precisely $\prob{X\in I}$. This allows us to obtain the following characterization of surrogate costs for committing policies.
\begin{fact}\label{cl:ws-com-sur}
    For any committing policy $\pi\in\comset{\mc}$, the water filling surrogate cost $W^*_{\mc^\pi}$ corresponds to sampling $x\sim X$ and returning
    \[\surr^\pi(x):= W^*_{\mc^\pi}(t(I))\]
    for the unique interval $I\in\I^\pi$ that contains $x$.
\end{fact}


\noindent
We are now ready to prove~\Cref{lem:osh-policy}. Let $\pi$ be the committing policy described in~\Cref{lem:osh-policy}. We have already handled the case of $g>\mu$. Next, consider $g>M$. In this case $\pi$ commits to no weighings, $\mc^\pi$ is simply a terminal state of value $\mu$ and thus $\surr^\pi(x) = \mu$ for all $x\in \X$. Since $\surr^*(x)\geq g$ for all $x\in \X$, this trivially implies a $\alpha = \mu/g \leq \mu/M$ pointwise approximation and the lemma follows.

If $M\geq g$, then $\pi$ corresponds to the one-sided halving algorithm with $t_1=g$ and $t_2 = \min(M,h)$. Observe that by definition, the minimum threshold used by the policy will be some $t_f\in [g,2g]$. Thus, the first interval in $\I^\pi$ will be $I_0=(-\infty , t_f]$. Note that \[\prob{X\in I_0} = \prob{X\leq t_f}\geq \prob{X\geq g} \geq \frac{c}{g}\]
where the last inequality follows from the fact that $c = \expect{(g-X)^+} \leq g \cdot \prob{X\leq g}$. Furthermore, recall that by the definition of (any) amortization, the cost shares for the amortization of a state $s$ satisfy
\[p(s)\cdot c(s) = \sum_{t\in T(S)}p(t)\cdot b_{st}\]
and since $\prob{X\in I_0} \geq c/g$ and all the action costs are $c(s)=c$, this implies that in any step during the amortization of $\mc^\pi$, the surrogate cost of terminal $t(I_0)$ increases by at most $g$. Finally, we note that the horizon of $\mc^\pi$ will be at most
\[k:=\log\frac{t_2}{t_f}\leq \log\frac{M}{g} \leq \log\frac{2\mu}{g}\] 
and thus we conclude that the total number of amortization steps will be $k$, and that the total increase in the value of the first terminal will be at most $kg$.

Up next, we will argue that the terminal $t(I_0)$ will be the terminal that suffers the maximum increase during the water-filling amortization of $\mc^\pi$. Note that by definition of $\mc^\pi$, the initial value of a terminal state $t(I)$ corresponding to some $I\in\I^\pi$ will be precisely $\mu(I):=\expect{X|X\in I}$. Thus, terminal $t(I_0)$ starts with the minimum value. Furthermore, by the one-sided structure of $\mc^\pi$, observe that $t(I_0)$ is a terminal state for all intermediate action states, and thus it will participate in all the stages of the water filling amortization. These two facts immediately prove the claim.

From the above, we can summarize that for all $I\in \I^\pi$ we have
\[W^*_{\mc^\pi}(t(I)) \leq \mu(I) + k\cdot g\]
and thus for any $x\in\X$, we have $\surr^\pi(x) \leq  \mu(I) + k g$ for the unique $I\in \I^\pi$ for which $x\in I$. It remains to upper bound the expectations $\mu(I)$. For $I_0=(-\infty , t_f] \subseteq (-\infty, 2g]$ we have that $\mu(I_0) \leq 2g$. For the maximum interval $I=(t_2,\infty)$, we have that $\mu(I)\leq 2\mu$; this is a consequence of the fact that $t_2\leq M$ and thus $\expect{X|X>t_2} \leq \expect{X|X\geq M}\leq 2\mu$.
Finally, for any other interval $I$, we know that the ratio between its endpoints will be precisely $2$ due to the halving and thus $\mu(I)\leq 2x$ for any $x\in I$. 

To summarize, we have shown that for all $x\in\X$:
\begin{equation*}
  \surr^\pi(x) \leq
    \begin{cases}
    kg + 2g & \text{if $x < g$}\\
    kg + 2x & \text{if $x\in [g,t_2]$}\\
    kg + 2\mu & \text{if $x>t_2$}
    \end{cases}       
\end{equation*}
Observe that since $t_2\leq M \leq 2\mu$, we have that for any $x\in\X$, \[\surr^\pi(x) \leq u(x):= kg + 2\cdot\min(2\mu , \max(x,g)).\] 

Notice that the upper bound $u(x)$ is non-decreasing. Now, recall that $\surr^*(x)=\min(h,\max(g,x))$ is also a non-decreasing mapping. This, policy $\pi$ will $\alpha$-pointwise approximate $\mc$ for
\[a = \max_{x\in X}\frac{u(x)}{\surr^*(x)}\]
and since $\mu\leq h$ and $k=O(\log\frac{\mu}{g})$ this implies a $O(\log\frac{\mu}{g} + \frac{\mu}{M})$-pointwise approximation.

\subsection{Proof of~\Cref{thm:ws-lb}}

\wslbtag*
\noindent For ease of notation, let $k:=\alpha+1$ and $B:= 2^{k^2}$. We consider the alternative with weighing cost $c=1$ and random cost $X$ that is continuously distributed in interval $[1,B]$ and has two point masses on $0$ and $(k+1)B$, namely:
    \[
    X =
    \begin{cases}
      0 & \text{with probability $1 - \frac{1}{k}$}\\
      x\in [1,B] & \text{with density $f(x) = \frac{1}{kx^2}$}\\
      (k+1)B & \text{with probability $\frac{1}{kB}$}
    \end{cases}       
    \]
    Note that since $B > k > 1$ and $\int_{x=1}^B\frac{1}{kx^2}dx = \frac{1}{k} - \frac{1}{kB}$, $X$ is indeed a valid random variable. We proceed by computing the relevant parameters of $X$.
    \begin{itemize}
        \smallskip
        \item The expected value of $X$ is $\mu = \expect{X} = k+1+\frac{1}{k}$ by definition. Thus, $\mu\in [k+1,k+2]$.
        
        \smallskip
        \item The $g$-index of $\mc$ satisfies $g\in [1,\frac{k}{k-1}]$. Define $\text{exc}(z):= \expect{(z-X)^+}$. The claim follows by noting that $\text{exc}(\cdot)$ is non-decreasing; $\text{exc}(1) = 1-\frac{1}{k}<1$ and $\text{exc}(\frac{k}{k-1}) > 1$; whereas $\text{exc}(g)=1$ by definition.
        
        \smallskip
        \item The $h$-index of $\mc$ is $h=B$, since $\expect{(X-B)^+} = \frac{1}{kB}\cdot [(k+1)B-B] = 1$.  
    \end{itemize}

\smallskip\noindent
By definition of $\surr^*(x)$, we have that $\surr^*(0) = g$,  $\surr^*((k+1)B) = B$ and $\surr^*(x) = \max(g,x)$ for all $x\in [1,B]$. We will now show that no committing policy can achieve an $\alpha'$-pointwise approximation for the alternative $(X,1)$ for any $\alpha'<\alpha$.

Fix any commitment $\pi$ and let $\ell(x):= \expect{X\in I}$ where $I\in\I^\pi$ is the unique interval of the policy's partition that contains $x$. Clearly, $\ell(x)$ is a lower bound to $\surr^\pi(x)$ and it is also a non-decreasing function. Thus, in order to prove impossibility of pointwise approximation for any $\alpha'<\alpha$, it suffices to show that there exists some $x\in\X$ for which
\[\frac{\ell(x)}{\surr^*(x)}\geq \alpha.\]

We begin by considering the committing policy $\pi$ that does not perform any weighings. Then, we simply have $\ell(x) = \surr^\pi(x)=\mu$ for all $x\in\X$ and since $\surr^*(x)\geq g$ for all $x\in X$, and the claim follows from $\alpha \leq \mu/g$. Next, consider any commitment that performs weighings, and let $t$ be the maximum threshold that it uses. Clearly, $t<(k+1)B$ otherwise there is no point in the weighing. This means that the final interval in $\I^\pi$ will be $I_\infty := (t,\infty)$. Let $\mu_\infty := \mu(I_\infty) = \expect{X|X>t}$. We distinguish between the following cases:

\smallskip
\begin{itemize}
\smallskip

    \item If $t\geq B$, then $\mu_\infty = (k+1)B$ and $\surr^*(t)\leq h = B$. In that case, the ratio is at least $k>\alpha$.
\smallskip

    \item If $t\leq 1$, then $\mu_\infty \geq \expect{X|X\geq 1} = k^2 + k + 1$ and $\surr^*(t) \leq \surr^*(1) = \min(h, \max(g,1)) = g$. In that case, the ratio is at least $k^2-k > \alpha$.
\smallskip

    \item Finally, if $t\in (1,B)$, then $\prob{X\geq t} = \frac{1}{kt}$ and thus $\mu_\infty \geq kt$. Since $\surr^*(t) = \max(g,t)$, the ratio is at least $k-1 = \alpha$. 

\end{itemize}
Thus, in any case there exists some $x$ for which $\ell(x) \geq \surr^*(x)$. As already mentioned, by the fact that both mappings are non-decreasing and by definition of pointwise-approximation, this proves~\Cref{thm:ws-lb}.

\newpage
\newpage
\section{Omitted Proofs from~\Cref{sec:gvo} (Optional Inspection PB)}\label{app:gvo}

In this chapter of the appendix, we present all omitted proofs from~\Cref{sec:gvo}. We restate all our results for the reader's convenience.

\subsection{Proof of~\Cref{thm:semilocal_composition}}
\semilocalcompo*
\begin{proof}
We consider the following algorithm, restated from~\Cref{sec:gvo}:   
\begin{algorithm}
    \SetAlgoLined
    \SetKwInOut{Input}{Input}
    \SetKwFor{For}{For }{\unskip:}{}{}
    \SetKwIF{If}{ElseIf}{Else}{If}{\unskip:}{Else if}{Else}{}
    \SetKwComment{Comment}{}{}
    \DontPrintSemicolon
    \caption{Semilocal Approximation Composition Algorithm}
    \label{algo:pboi_repeat}
    \Input{%
        A normalized matroid-max-PBOI instance $\instance = (\pbox_1, \dots, \pbox_n, \feas)$. \\
        A vector of probabilities $(p_1, \dots, p_n)$.
    }
    \smallskip
    Relabel the boxes such that $\mu_1 \geq \dots \geq \mu_n$
    \;
    \smallskip
    $S^\grab \gets \{\}$
    \Comment*{\color{gray}\normalfont%
        set of boxes marked as ``grab''}
    \smallskip
    \For{\upshape$i \gets 1, \dots, n$}{%
        \smallskip
        
        $L_i \gets \textbf{if } S^\grab \cup \{i\} \in \feas \textbf{ then } 1 \textbf{ else } 0$
        \Comment*{\color{gray}\normalfont%
            $L_i = 0$ means we never want to grab box~$i$}
        \label{line:semilocal_composition:use_local}
        \smallskip
        Sample $K_i \gets \Bernoulli(p_i)$
        \Comment*{\color{gray}\normalfont%
            $K_i = 1$ means \emph{provisionally} mark box~$i$ ``grab''}
        \smallskip
        \If{\upshape$K_i = 1 \textbf{ and } L_i = 1$}{
            \smallskip
            $S^\grab \gets S^\grab \cup \{i\}$
            \Comment*{\color{gray}\normalfont%
                fully mark box~$i$ as ``grab''}
        }
    }
    \smallskip
    
    Commit to grabbing the boxes $S^\grab$ and opening the boxes in $[n]\setminus S^\grab$.
    \;
    \label{line:semilocal_composition:commitment}
    \smallskip
    Run the optimal (index) policy under the resulting commitment.
    \;
\end{algorithm}

To simplify the notation, we let $W^\insp_i:= W^*_{o,i}$ be the surrogate cost of the opening action for box $\pbox_i$. There are two main steps of the proof, each stated and proved in a lemma below. We express both steps in terms of the random variables
    \[
        W^\alg_i = \begin{cases}
            W^\insp_i & \text{if } K_i L_i = 0 \\
            \mu_i & \text{if } K_i L_i = 1,
        \end{cases}
    \]
    where $K_i$ and $L_i$ are as defined in \Cref{algo:pboi_repeat}.
    One can think of $W^\alg_i$ as the surrogate value of box~$i$ conditional on the state of the algorithm at \cref{line:semilocal_composition:commitment}. The first step, \cref{thm:semilocal_composition:achieved}, is to express the value achieved by \cref{algo:pboi_repeat} in terms of~$W^\alg_i$:

    \begin{lemma}
        \label{thm:semilocal_composition:achieved}
        For any max-matroid Pandora's box instance $\instance = (\pbox_1, \dots, \pbox_n, \feas)$ and any vector of probabilities $(p_1, \dots, p_n)$, the expected value achieved by \cref{algo:pboi_repeat} is
        \[
            \E{\textnormal{value achieved by \cref{algo:pboi_repeat}}}
            = \E[\bigg]{\max_{S \in \feas} \sum_{i \in S} W^\alg_i}.
        \]
    \end{lemma}

    The second step, \cref{thm:semilocal_composition:comparison}, is to compare the resulting expression to an upper bound on the optimal value. This step uses the semilocal $(\alpha, \beta)$-approximation guarantee from \cref{def:semilocal_approx}, which gives us a relationship between $W^\alg_i$ and~$W^*_i$.
    \begin{lemma}
        \label{thm:semilocal_composition:comparison}
        Under the hypotheses of \cref{thm:semilocal_composition},
        \begin{equation*}
            \E[\bigg]{\max_{S \in \feas} \sum_{i \in S} W^\alg_i}
            \geq \E[\bigg]{\max_{S \in \feas} \sum_{i \in S} \alpha W^*_i - \smashoperator{\sum_{i \in S^\grab}} \beta \mu_i},
        \end{equation*}
        where $S^\grab$ refers to the value of the $S^\grab$ variable from \cref{algo:pboi_repeat} at \cref{line:semilocal_composition:commitment}.
    \end{lemma}

    Combining the lemmas yields
    \[
        \E{\text{value achieved by \cref{alg:semilocal_composition}}}
        \geq \E[\bigg]{\max_{S \in \feas} \sum_{i \in S} \alpha W^*_i - \smashoperator{\sum_{i \in S^\grab}} \beta \mu_i},
    \]
    where $S^\grab$ is the set of boxes marked ``grab'' at \cref{line:semilocal_composition:commitment}, i.e. at the end of the algorithm.
    It remains only to relate the two terms on the right-hand side to the optimal expected value.
    \* From the analogue of~\Cref{thm:MDP-lb} for the maximization setting (see~\Cref{thm:MDP-lb-max} in~\Cref{app:maxim}), we know that the water filling surrogate costs provide an upper bound on the utility of the optimal adaptive policy for any max-CICS instance. Thus:
    \[
        \E[\bigg]{\max_{S \in \feas} \sum_{i \in S}  W^*_i} \geq \E{\text{value of optimal policy}}.
    \]
    \* Because \cref{algo:pboi_repeat} ensures $S^\grab \in \feas$ by construction, the following policy is feasible: ``Compute $S^\grab$ as in \cref{algo:pboi_repeat}, but then simply grab the boxes in $S^\grab$.'' This algorithm achieves value $\E[\big]{\sum_{i \in S^\grab} \mu_i}$, which means
    \[
        \E[\bigg]{\smashoperator[r]{\sum_{i \in S^\grab}} \mu_i} \leq \E{\text{value of optimal policy}}.
    \]
    \*/
    Therefore, as desired, $\E{\text{value achieved by \cref{algo:pboi_repeat}}} \geq (\alpha - \beta)\cdot  \E{\text{value of optimal policy}}$ which implies a lower bound of $\alpha-\beta$ for the commitment gap.
\end{proof}

\paragraph{Proof of~\Cref{thm:semilocal_composition:achieved}.} Consider running \cref{algo:pboi_repeat} through \cref{line:semilocal_composition:commitment}, but not further. All of the randomness thus far comes from the coin flips~$K_i$, and no boxes have been opened yet. This means that conditional on the coin flips~$K_i$, the expected value achieved is that of a \emph{mandatory-inspection} instance $\instance' = (\pbox'_1, \dots, \pbox'_n, \feas)$ whose $i$th box $\pbox'_i$ is defined as follows:
    \* If $K_i L_i = 0$ (marked ``open''), $\pbox'_i = (\dist_i, c_i)$, i.e. the box is the same as the original instance.
    \* If $K_i L_i = 1$ (marked ``grab''), $\pbox'_i = (\mu_i, 0)$, i.e. the box is free to open and always contains value~$\mu_i$.
    \*/
    Under instance~$\instance'$, box~$i$'s surrogate value is given by $W^\alg_i$. Since this is an instance of max-CICS over Markov chains, from the counterpart of~\Cref{thm:MC-opt} for maximization (see~\Cref{thm:MC-opt-maxim} in~\Cref{app:maxim}) we have
    \[
        \E{\text{value achieved by \cref{algo:pboi_repeat}} \given K_1, \dots, K_n}
        = \E[\bigg]{\max_{S \in \feas} \sum_{i \in S} W^\alg_i \given K_1, \dots, K_n}.
    \]
    The lemma then follows by the law of total expectation.

\paragraph{Proof of~\Cref{thm:semilocal_composition:comparison}.} We want to prove that
   
   \begin{equation}
    \label{eq:semilocal_composition:comparison}
            \E[\bigg]{\max_{S \in \feas} \sum_{i \in S} W^\alg_i}
            \geq \E[\bigg]{\max_{S \in \feas} \sum_{i \in S} \alpha W^*_i - \smashoperator{\sum_{i \in S^\grab}} \beta \mu_i}.
    \end{equation}
        
    The outline of the proof is as follows. We begin with the left-hand side of~\cref{eq:semilocal_composition:comparison}.
    Then, for each box~$i$, we swap $W^\alg_i$ with $\alpha W^*_i$, and also subtract $\beta \mu_i$ if box~$i$ is marked ``grab''.
    Because each box admits a semilocal $(\alpha, \beta)$-approximation, each of these replacements only decreases the expression's expected value.
    After all $n$ replacements, we are left with the right-hand side of~\cref{eq:semilocal_composition:comparison}, as desired.
    This is the same strategy used by \citet[Theorem~5.4]{DS24}, but some careful conditioning is needed to account for the $\beta \mu_i$ subtractions.

    In order to notate the one-by-one replacement outlined above, let
    \begin{align*}
        W^{(j)}_i
        &= \begin{cases}
            W_i^\alg & \text{if } i \leq j \\
            \alpha W_i^* & \text{if } i > j,
        \end{cases}
        &
        U^{(j)}
        &= \max_{S \in \feas} \sum_{i \in S} W^{(j)}_i.
    \end{align*}
    Using this notation and recalling how $S^\grab$ is defined in \cref{algo:pboi_repeat}, we can rewrite our goal \cref{eq:semilocal_composition:comparison} as
    \[
        \E{U^{(n)}} \geq \E[\bigg]{U^{(0)} - \sum_{i = 1}^n \beta K_i L_i \mu_i}.
    \]
    Therefore, it suffices to show that for all $j \in \{1, \dots, m\}$,
    \begin{equation}
        \label{eq:semilocal_composition:suffices_1}
        \E{U^{(j)}}
        \geq \E[\bigg]{U^{(j - 1)} - \beta K_j L_j \mu_j}.
    \end{equation}

    We will show \cref{eq:semilocal_composition:suffices_1} using the definition of semilocal approximation (\cref{def:semilocal_approx}).
    But in order to do so, we need to express each side in terms of a maximum between $W^\alg_j$ or $\alpha W^*_j$ and a quantity that is independent of box~$j$'s value~$X_j$ and coin flip~$K_j$ (this will take the place of the outside option in the definition of semilocal approximation).
    We express the latter quantity in terms of
    \begin{align*}
        Y_{\neq j}
        &= \max_{S \in \feas : j \not\in S} \sum_{i \in S} W^{(j)}_i,
        &
        Z_{\neq j}
        &= \max_{S \in \feas : j \in S} \smashoperator[r]{\sum_{i \in S \setminus \{j\}}} W^{(j)}_i.
    \end{align*}
    These can both be seen as optimal total surrogate values achievable without box~$j$.
    The difference is that $Y_{\neq j}$ optimizes over sets that exclude~$j$, whereas $Z_{\neq j}$ optimizes over sets that include~$j$ (but still excludes box~$j$'s surrogate value from the sum). With the definitions of $Y_{\neq j}$ and $Z_{\neq j}$ in hand, we can express $U^{(j)}$ and $U^{(j - 1)}$ as
    \begin{align*}
        U^{(j)}
        &= \max\Bgp{Y_{\neq j}, Z_{\neq j} + W^\alg_j} = Z_{\neq j} + \max\Bgp{W^\alg_j, Y_{\neq j} - Z_{\neq j}},
        \\
        U^{(j - 1)}
        &= \max\Bgp{Y_{\neq j}, Z_{\neq j} + \alpha W^*_j} = Z_{\neq j} + \max\Bgp{\alpha W^*_j, Y_{\neq j} - Z_{\neq j}}.
    \end{align*}
    So to show \cref{eq:semilocal_composition:suffices_1}, it suffices to show
    \[
        \E{\max\Bgp{W^\alg_j, Y_{\neq j} - Z_{\neq j}}}
        \geq \E{\max\Bgp{\alpha W^*_j, Y_{\neq j} - Z_{\neq j}} - \beta K_j L_j \mu_j}.
    \]
    Letting $K_{< j} = (K_1, \dots, K_{j - 1})$, by the law of total expectation, it suffices to show
    \begin{equation}
        \label{eq:semilocal_composition:suffices_2}
        \E{\max\Bgp{W^\alg_j, Y_{\neq j} - Z_{\neq j}} \given K_{< j}, Y_{\neq j}, Z_{\neq j}}
        \geq \E{\max\Bgp{\alpha W^*_j, Y_{\neq j} - Z_{\neq j}} - \beta K_j L_j \mu_j \given K_{< j}, Y_{\neq j}, Z_{\neq j}}.
    \end{equation}
    The key to showing \cref{eq:semilocal_composition:suffices_2} is observing the following independence facts:
    \* $(K_j, X_j)$ is independent of $K_{< j}$. This is because the coin flips~$K_{< j}$ affect neither the coin flip~$K_j$ nor the box value~$X_j$.
    \* $(K_j, X_j)$ is conditionally independent of $(Y_{\neq j}, Z_{\neq j})$ given $K_{< j}$. This is because once $K_{< j}$ are fixed, the values $(Y_{\neq j}, Z_{\neq j})$ are a function of the values of boxes other than~$j$, and the box values are mutually independent.
    \*/
    The main obstacle to applying the semilocal approximation condition to \cref{eq:semilocal_composition:suffices_2} is that $W^\alg_j$ depends on $L_j$, which in turn depends on $K_{< j}$.
    Fortunately, we see from \cref{algo:pboi_repeat} that
    \[
        L_{\leq j} = (L_1, \dots, L_j) \text{ is a deterministic function of } K_{< j} = (K_1, \dots, K_{j - 1}).
    \]
    This is because for all~$i$, when executing \cref{line:semilocal_composition:use_local}, the only randomness the algorithm has used is the past coin flips~$K_{< i}$.
    So to show \cref{eq:semilocal_composition:suffices_2}, we split into cases based on whether $L_j = 0$ or $L_j = 1$.

    Suppose that $L_j = 1$.
    More precisely, suppose $K_{< j} = k_{< j}$, where $k_{< j}$ is any bit vector such that $K_{< j} = k_{< j}$ induces $L_j = 1$ in \cref{algo:pboi_repeat}.
    In this case, the coin flip~$K_j$ impacts whether we mark box~$j$ as ``grab'' or ``open'', so we will use the semilocal approximation guarantee from \cref{def:semilocal_approx}.
    By the assumption on $k_{< j}$ and the fact that $K_j \sim \Bernoulli(p_j)$ independently of $K_{< j}$,
    \[
        \E{K_j L_j \given K_{< j} = k_{< j}} = \E{K_j} = p_j.
    \]
    So, by \cref{def:semilocal_approx},
    \begin{align*}
        \E{\max\Bgp{W^\alg_j, y} \given K_{< j} = k_{< j}}
        &= (1 - p_j) \E{\max\Bgp{W^\insp_j, y}} + p_j \max\Bgp{\mu_j, y} \\
        &\geq \E{\max\Bgp{\alpha W^*_j, y}} - \beta p_j \mu_j \\
        &= \E{\max\Bgp{\alpha W^*_j, y} - \beta K_j \mu_j} \\
        &= \E{\max\Bgp{\alpha W^*_j, y} - \beta K_j L_j \mu_j \given K_{< j} = k_{< j}}.
    \end{align*}
    The fact that $(Y_{\neq j}, Z_{\neq j})$ is conditionally independent of $(K_j, X_j)$ given $K_{< j}$ completes the proof of \cref{eq:semilocal_composition:suffices_2} on the event $L_j = 1$.

    Suppose now that $L_j = 0$.
    In this case, we mark box~$j$ as ``open'' regardless of the coin flip~$K_j$, so instead of using the semilocal approximation condition, we will show that marking box~$j$ as ``open'' does not lose any any potential value.
    Specifically, we will show the following:
    \*[(a)] For all $y \geq h_j$, we have $\E{\max\Bgp{W^\insp_j, y}} = \E{\max\Bgp{W^*_j, y}}$.
    \* If $L_j = 0$, then $Y_{\neq j} - Z_{\neq j} \geq \mu_j$.
    \*/
    Together with the fact that $\mu_j \geq h_j$ (which holds by the assumption that all degenerate boxes have been normalized), facts~(a) and~(b) imply that for any $k_{< j}$ such that $K_{< j} = k_{< j}$ induces $L_j = 0$ in \cref{algo:pboi_repeat},
    \begin{align*}
        \E{\max\Bgp{W^\alg_j, Y_{\neq j} - Z_{\neq j}} \given K_{< j} = k_{< j}}
        &= \E{\max\Bgp{W^\insp_j, Y_{\neq j} - Z_{\neq j}} \given K_{< j} = k_{< j}} \\
        &= \E{\max\Bgp{W^*_j, Y_{\neq j} - Z_{\neq j}} \given K_{< j} = k_{< j}} \\
        &\geq \E{\max\Bgp{\alpha W^*_j, Y_{\neq j} - Z_{\neq j}} \given K_{< j} = k_{< j}} \\
        &= \E{\max\Bgp{\alpha W^*_j, Y_{\neq j} - Z_{\neq j}} - \beta K_j L_j \mu_j \given K_{< j} = k_{< j}},
    \end{align*}
    which completes the proof of \cref{eq:semilocal_composition:suffices_2} on the event $L_j = 0$.
    It remains only to show (a) and~(b).
    Fact~(a) holds by definition: recall that $h_j$ denotes the smallest value of the outside option $y$ for which the optimal policy in the local game will prefer opening to grabbing. For~(b), we will show that if $L_j = 0$, then $Y_{\neq j} \geq \mu_j + Z_{\neq j}$.
    Let
    \* $B^Z \in \argmax_{S \in \feas : j \in S} \sum_{i \in S \setminus \{j\}} W^{(j)}_i$ be a maximizing basis in the definition of $Z_{\neq j}$,
    \* $S^\grab_{< j} = S^\grab \cap \{1, \dots, j - 1\}$ be the boxes marked ``grab'' before the $j$-th step in the loop, and
    \* $B^\grab$ be $S^\grab_{< j}$ extended to a basis by elements of $B^Z$, so that $S^\grab_{< j} \subseteq B^\grab \subseteq S^\grab_{< j} \cup B^Z$.
    \*/
    Because $L_j = 0$, we have $S^\grab_{< j} \cup \{j\} \not\in \feas$, which means $j \not\in B^\grab$.
    But $j \in B^Z$ by definition, so $j \in B^Z \setminus B^\grab$.
    By the basis exchange property, there exists $k \in B^\grab \setminus B^Z$ such that the following is a basis:
    \[
        B^Y = (B^Z \setminus \{j\}) \cup \{k\}.
    \]
    But $B^\grab \setminus B^Z \subseteq S^\grab_{< j}$, which means $W^{(j)}_k = \mu_k \geq \mu_j$ since~\Cref{algo:pboi_repeat} iterates over the boxes in decreasing order of mean value.
    This means
    \[
        Y_{\neq j}
        \geq \smashoperator{\sum_{i \in B_Y}} W^{(j)}_i
        = \mu_k + \smashoperator{\sum_{i \in B_Z \setminus \{j\}}} W^{(j)}_i
        \geq 
        \mu_j + \smashoperator{\sum_{i \in B_Z \setminus \{j\}}} W^{(j)}_i
        = \mu_j + Z_{\neq j}.
        \qedhere
    \]
    and the proof is completed.

\subsection{Proof of~\Cref{lem:semilocal-constant}}

\semilocalconstant*
\begin{proof}
    To simplify notation we denote $\alpha(\beta)$ as $\alpha$ throughout the proof. By~\Cref{def:semilocal_approx}, proving this lemma amounts to showing that these $\alpha$ and $\beta$ satisfy:
    \begin{equation}\label{eq:def-semilocal-approx}    
        (1 - p) \E{\max\Bgp{W^*_o, y}} + p \max\Bgp{\mu, y}
        \geq \E{\max\Bgp{\alpha W^*, y}} - p \beta \mu
    \end{equation}
    for all $y \geq 0$. This is equivalent to showing that
    \[
        f(y) := (1 - p) \E{\max\Bgp{W^*_o, y}} + p \max\Bgp{\mu, y} - \E{\max\Bgp{\alpha W^*, y}} + p \beta \mu
    \]
    satisfies $f(y) \geq 0$ for all $y \geq 0$. We will first show that it is sufficient to check that $f(0) \geq 0$ and that $f(y) \geq 0$ for $y \geq \mu$. It follows from \cref{def:pboi_sur_cost} that at all values $y$ for which $f'(y)$ is defined, it is equal to
    \[
        (1-p) \prob{y \geq W^*_o} + p \mathbbm{1}(y > \mu) - \mathbbm{1}(y > \alpha h)\prob{y > \alpha W^*_o}.
    \]
    Observe that
    \begin{itemize}
        \item if $y < \alpha h < \mu$, then $f'(y) \geq 0$,
        \item if $\alpha h < y < \mu$, then since $\alpha \leq 1$, we must have $f'(y) \leq 0$.
    \end{itemize}
    Thus a global minimum of $f(y)$ must be at $y = 0$ or on $y \geq \mu$, so it is sufficient to check that \cref{eq:def-semilocal-approx} is satisfied when $y=0$ and $y \geq \mu$. When $y=0$, \cref{eq:def-semilocal-approx} reduces to
    \begin{equation*}
        (1-p)(\mu-c)+p\mu \geq \alpha \mu-p\beta \mu.
    \end{equation*}
    which holds for the values of $\alpha$ and $p$ given in the lemma.
    When $y \geq \mu$, \cref{eq:def-semilocal-approx} reduces to
    \begin{equation}\label{eq:semilocal-approx-suff-inequality-y>m}
        (1-p) \E{\max\Bgp{W^*_o, y}} + py
        \geq \E{\max\Bgp{\alpha W^*, y}} - p \beta \mu = \alpha\E{\max\Bgp{W^*_o, \frac{y}{\alpha}}} - p \beta \mu 
    \end{equation}
    where the last equality follows from the fact that $\E{\max\Bgp{\alpha W^*, y}} = \max\{\E{\max\Bgp{\alpha W^*_o, y}}, \alpha \mu\}$. The slope of $\E{\max\{W^*_o, y\}}$ is bounded above by $1$, so 
    \[
        \alpha\E{\max\Bgp{W^*_o, \frac{y}{\alpha}}} \leq \alpha \E{\max\Bgp{W^*_o, y}} + \alpha \left(
            \frac{y}{\alpha} - y
        \right).
    \]
    Applying this bound to \cref{eq:semilocal-approx-suff-inequality-y>m}, we find that \cref{eq:semilocal-approx-suff-inequality-y>m} holds if,
    \begin{equation}\label{eq:semilocal-approx-suff-inequal-y>m-after-bound}
        (\alpha+p-1) (\E{\max\Bgp{W^*_o, y}}-y)
        \leq p \beta \mu. 
    \end{equation}
    Since the slope of $\E{\max\{W^*_o, y\}}$ is bounded above by $1$, it must be that $\E{\max\Bgp{W^*_o, y}}-y$ is maximized at $y=0$ where it is equal to $\mu - c$, so \cref{eq:semilocal-approx-suff-inequal-y>m-after-bound} holds if,
    \[
        (\alpha+p-1) (\mu-c)
        \leq p \beta \mu. 
    \]
    Observe that if $p = \alpha\frac{c}{\mu}$, we can rewrite this inequality as
    \[
        \alpha\left(1 + \frac{c}{\mu} - \frac{c}{\mu-c}\beta \right) \leq 1,
    \]
    which is satisfied by the $\alpha$ given in the lemma.
\end{proof}


\newpage
\bibliographystyle{plainnat}
\bibliography{references}

\end{document}